\documentclass{article}
\input{tikzit.tikzdefs}

\title{\textbf{Causal and Compositional Abstraction}}
\date{}
\author{Robin Lorenz  \ \ \ \  Sean Tull\\[0.04cm]
{\small  \{\nolinkurl{robin.lorenz, sean.tull}\}\nolinkurl{@quantinuum.com}} \\[0.03cm]
{\small Quantinuum, London, United Kingdom}}

\begin{document}

\maketitle

\begin{abstract}
\rlb{Abstracting from a low level to a \emph{more explanatory high level} of description, and ideally in a way that is in keeping with the respective \emph{causal structures} of the variables, is arguably fundamental to scientific practice, to causal inference problems, and to robust, efficient and interpretable AI.  
We present a general account of abstractions between low and high level models as \emph{natural transformations}, focusing on the case of causal models. 
This provides a new formalisation of \emph{causal abstraction}, unifying several notions in the literature, including constructive causal abstraction, Q-$\tau$ consistency, abstractions based on interchange interventions, and `distributed' causal abstractions.  
Our approach is formalised in terms of \emph{category theory}, and uses the general notion of a \emph{compositional model} with a given set of queries and semantics in a monoidal, cd- or Markov category; causal models and their queries such as interventions being special cases. 
We identify two basic notions of abstraction each described by natural transformations: \emph{downward abstractions}, applying to abstract diagrammatic queries, mapping them from high to low level; and \emph{upward abstractions}, applying to concrete queries such as specific Do-interventions, mapping them from low to high level. Although usually presented as the latter kind, we show how common notions of causal abstraction may, more fundamentally, be understood in terms of the former kind. 
In either case our categorical approach helps bring the structural essence of (causal) abstraction to the forefront. 
It also leads us to consider a new stronger notion of `component-level' abstraction, applying not only to queries but individual components of a model. In particular, this yields a novel, strengthened form of \emph{constructive causal  abstraction at the mechanism-level}, for which we prove characterisation results. 
Finally, by generalising beyond the causal case we show that abstraction can be defined for further compositional models, including those with a \emph{quantum semantics} implemented by quantum circuits, and we take first steps in exploring abstractions between quantum compositional circuit models and high-level classical causal models as a means to developing explainable quantum AI.} 
\end{abstract}

\section{Introduction} \label{sec:intro}

Abstracting away details in order to reason in terms of explanatory high-level concepts is a hallmark of both human cognition and scientific practice; from microscopic vs thermodynamic descriptions of a gas, to extracting macrovariables from climate data or medical scans, all the way to economics. In each of these cases, and most if not all situations of scientific practice, issues of abstraction intersect with those of \emph{causality} and causal reasoning. 

This intuition has been made precise through notions of \emph{causal abstraction}, pioneered over the past decade in a number of works \cite{ChalupkaEtAl_2017_CausalFeatureLearning, rubenstein2017causal, BeckersEtAl_2019_AbstractingCausalModles, beckers2020approximate, massidda2023causal, geiger2023causal}. These begin with the causal model framework \cite{pearl2009causality}, which formalises causal models and reasoning at a given level of abstraction. A `high-level' causal model $\modelMH$ then stands in a causal abstraction relation to a `low-level' one $\modelML$, when there is an abstraction map $\tau$ from low to high-level variables and a relation between interventions at each level such that intervening at the high-level is causally \emph{consistent} with respect to interventions at the low-level. 

Causal abstraction is receiving increasing attention for both foundational and practical reasons in fields ranging from causal inference to AI and philosophy. Despite the varying assumptions, goals and forms of causal model, in each scenario the same basic relation is needed. There are at least three distinct contexts that motivate its study. Firstly, a generalisation of \emph{causal identifiability problems} in which the level of our causal hypotheses and quantities \rlb{of interest}\ differs from the level at which we have empirical data. In this case the low-level causal model is (an approximation of) a ground-truth model which we neither have access to, nor usually care about. 
Examples include questions varying from weather phenomena \rlb{to human brains \cite{ChalupkaEtAl_2016_MultiLevelCauseEffectSystems, ChalupkaEtAl_2017_CausalFeatureLearning, xia2024neural, huang2025causal}. Relatedly, abstraction is also key for a causal study of racial discrimination to avoid the `loss of modularity' from wrongly mixing relata of different levels of abstraction \cite{mosse2025modeling}.}

Second is the study of \emph{interpretability and explainability} in artificial intelligence and deep learning. Similarly, here the aim is to relate a complex and obscure low-level model, such as a neural network, to a human-interpretable high-level model. Now, however, the low-level model is perfectly known and not considered (an approximation of) a ground-truth model, but a computational model trained for some task which may or may not be explicitly causal; yet the abstraction relation to a manifestly interpretable, high-level causal model can provide a gold standard kind of interpretability. This has been pursued for post-hoc XAI methods by a number of works \cite{geiger-etal-2020-neural, GeigerEtAl_2021_CausalAbstracttoinOfNN, hu2022neuron, Beckers_2022_CausalExplanationsAndXAI, WuEtAl_2023_InterpretabilityatScale, geiger2024finding, geiger2023causal, wu2024reply, pislar2025combining}, with \cite{geiger2023causal} in fact arguing that any method in \emph{mechanistic interpretability} is best understood from a causal abstraction perspective.  
Abstraction has also been used for not-post-hoc XAI, as an inductive bias to obtain models which are interpretable by design \cite{GeigerEtAl_2022_InducingCausalStructureForInterpretableNN}. 
\rlc{Moreover, it has been argued to provide an account of computational explanation more generally \cite{geiger2025_CausalAbstarctionUnderpinsCOmpExplaantion}.} 

Third is the problem of \emph{causal representation learning}, which aims to obtain AI models that have learned appropriate causal variables and relations directly from `raw', low-level data, without (all of) the information as to the appropriate dimensions or number of such variables. This generalises both causal discovery and other forms of (disentangled) representation learning, being argued to be essential for robustness and strong and efficient generalisabilty in AI \cite{scholkopf2021toward}.  An increasing number of works have explored questions of its theoretical and practical (im-)possibility, often relying on causal abstraction as a key ingredient \cite{brehmer2022weakly, yao2023multi, kekic2023targeted,  yao2024unifying, kori2025unifying, li2025identifiability, zhu2024unsupervised}. Relatedly, it has been shown that embodied AI relies on veridical world models that have learned causal structure at the right level of abstraction \cite{gupta2024essential, richens2024robust, richens2025general}.

\paragraph{A language for abstraction.}

Despite forming a basic and essential notion for science, there is currently no unified formalism for causal abstraction, or any consensus as to which of a suite of closely related notions of abstraction are appropriate in which situation. For both future theory and experiments in abstraction, an appropriate conceptual language is needed.

Now, causal reasoning is all about the interplay between syntactic causal structure and semantic probabilistic data, while abstraction is about how a model at one level of detail relates to another. \emph{Category theory} is the mathematics of processes, structure and composition \cite{coecke2006introducing,leinster2014basic}; precisely the language that brings structural aspects and their relations to the forefront so they can be studied directly. Indeed categories are natural for describing both the structure of a model, through \emph{monoidal categories} and the intuitive graphical language of \emph{string diagrams} \cite{piedeleu2023introduction,lorenz2023causal,tull2024towards}, as well as relationships between models, through the classic notions of \emph{functors} and \emph{natural transformations} \cite{mac2013categories}.

A presentation of the causal model framework in categorical terms, using string diagrams, was given in \cite{lorenz2021causal}, building on \cite{jacobs2019causal,fritz2023d}.  In \cite{tull2024towards} this approach was extended to more general \emph{compositional models}, specified by functors from structure to semantics, showing that the structure and interpretability of a wide range of AI models (including causal models) can be expressed in these terms.

In this work we extend this approach to give a precise and general formalisation of causal abstraction, finding that the essence of abstraction can be expressed more generally at the level of compositional models with given sets of queries. More precisely, we show that abstraction can be formalised categorically in terms of natural transformations between models:

\vspace*{0.15cm}
\begin{center}
	\fbox{ \vspace*{0.18cm} Abstractions are natural transformations between queries.
	\vspace*{0.15cm} }
\end{center}
\vspace*{0.15cm}
In more detail, our contributions are the following: 
\begin{itemize}

	\item \emph{Abstraction abstractly}:  
	We present a general definition of abstraction between compositional models with given sets of queries and show that this includes as special cases many notions of causal abstraction from the literature such as constructive abstraction \cite{BeckersEtAl_2019_AbstractingCausalModles}, exact transformations \cite{rubenstein2017causal}, `$Q-\tau$-consistency' \cite{xia2024neural}, and `distributed' abstractions \cite{geiger2023causal,geiger2024finding}. 

	\item  \emph{The power of diagrams}: We show that the consistency conditions for each form of (causal) abstraction can be expressed straightforwardly as string diagrams, bringing their essence to the forefront.

	\item  \emph{Up and down}: We find that there are in fact two \rlb{basic kinds of abstractions}, one mapping queries from low to high level, and one from high to low; we find that while usually presented in the former `concrete' kind, the essence of many paradigmatic notions of causal abstraction lies in the more structural latter kind. 

	\item \emph{General interventions}: We provide a unified treatment of causal abstraction which allows for more general interventions, not necessarily restricted to Do-interventions, and which allows us to give a new precise notion of distributed causal abstractions (in keeping with e.g. \cite{massidda2023causal, geiger2023causal}).

	\item \emph{A new, strong notion}: We introduce \emph{component-level} abstractions, which are natural from the categorical perspective, and in the causal case yield a new stronger notion of \emph{mechanism-level} constructive abstraction which raises new questions and which we characterise mathematically.\footnote{\st{The latter appeared implicitly in \cite{englberger2025causal}; see the related work and Appendix \ref{sec:Engl-comparison}.}}  

	\item \emph{Quantum generalisation}: Extending abstraction to compositional models allows us to introduce a  definition of abstraction for quantum compositional models, such as those based on quantum circuits, and we use this to initiate the study of quantum abstraction for interpretable quantum AI.   
  
\end{itemize}

\paragraph{Related work.}
\rlc{Specific notions of causal abstraction, not using a category theoretic perspective, have}\ been defined formally in a number of works, with our view mostly based on \cite{rubenstein2017causal,BeckersEtAl_2019_AbstractingCausalModles,geiger2023causal,xia2024neural}, from which we focus on the latter \rlc{three}. \rlc{(Also see \cite{zennaro2022abstraction, zhang2024causal} for reviews.)} 

Our approach builds on the account of structured (AI) models as compositional models described by string diagrams, which underlies for example, DisCo-models in (Q)NLP \cite{coecke2010mathematical,coecke2021mathematics,wang2023distilling} and the `discopy' and `Lambeq' python toolkits \cite{de2020discopy,kartsaklis2021lambeq}. The account of causal models in this framework, building on \cite{jacobs2019causal,fritz2023d}, was spelled out in detail in \cite{lorenz2023causal}. A broad view of AI models as compositional models and the advantages of \emph{interpretable} compositional structure for XAI was given in \cite{tull2024towards} (see also \cite{de2022categorical}). Here we extend this approach, beyond describing models as functors, we now see how describing relations between models as natural transformations indeed lead to meaningful notions of abstraction which have arisen independently in the literature. 

\st{The idea that causal abstractions should form natural transformations between causal models seen as Markov functors was also captured independently in the work of Englberger and Dhami \cite{englberger2025causal}, which we recommend for further reading. Their definition, after a necessary fix, in fact corresponds to the special case of mechanism-level constructive causal abstraction, as we outline in detail in Appendix \ref{sec:Engl-comparison}. Our approach based on queries allows the difference between constructive abstraction and its mechanism-level version to be made apparent, and combined with our notion of upward abstraction, allows us to unify a wider range of further abstractions including counterfactual abstractions, exact transformations, a fully formalised notion of distributed abstractions, and quantum generalisations.}

In \cite{otsuka2022equivalence} the authors \st{also} studied a more restricted form of causal abstraction using natural transformations, however most suitable for `equivalent' models rather than genuine abstraction. Our approach is strictly more general and captures known forms of causal abstraction. 
 
\rl{Also related is the work in \cite{rischel2021compositional}, which studies the composition of abstraction relations in the language of enriched categories and especially how bounds on errors behave under composition.} 

\rlc{Finally, the set of works \cite{DAcunto_2025causal, DAcunto2025relativity, DAcunto2025causal_semantic_embedding, DAcunto2026learning}, amongst other things, also formalises abstraction between a certain class of SCMs in terms of natural transformations. However their formalisation of causal models as functors is different from ours, i.e.~that of \cite{lorenz2023causal,jacobs2019causal,fritz2023d}, as is the notion of abstraction they capture, namely $\alpha$-abstraction \cite{rischel2020category}. In future work it would be interesting to more fully understand the connections to this approach.}

\subsection{Overview} \label{subsec:overview}

Let us now give a high-level overview of our account of abstraction. Recall that a symmetric monoidal category $\catC$ consists of a collection of objects and processes between them, which we can compose using graphical string diagrams. For causal models, we use the category $\FStoch$ of finite sets and probability channels (Stochastic matrices) between them, which forms a \emph{Markov category} with distinguished `copy' and `discard' processes \cite{fritz2020synthetic}. \maybe{In Section \ref{sec:setup} we introduce this necessary background on category theory and string diagrams.}

\paragraph{Models and queries.} A \emph{compositional model} $\modelM$ assigns semantics in $\catC$ to a collection of \emph{variables} and \emph{components}, where each component is an abstract process coming with a specified list of input and output variables. Each variable $\syn{V}$ is represented as an object $V=\modelMgen{\syn{V}}$ and each component $\syn{c}$ as a process $c = \modelMgen{c}$ in $\catC$ \rlb{(note the change of font going from syntax to semantics)}. For a causal model, $V$ is the finite set of values of variable $\syn{V}$ and each $c=c_V$ is the probability channel for a causal \emph{mechanism} $P(V \mid \Pa(V))$. 
Formally, the variables and components form a \emph{signature} $\SigS$ generating a free category $\strucS$, and the model a functor $\modelMgen{-} \colon \strucS \to \catC$.

Beyond these, a model also gives semantics to a collection of \emph{queries}, describing how we use it in practice. Each query is mapped to a process $Q=\semSM{\syn{Q}}$ in $\catC$, typically given by a diagram of components. Formally these form another signature $\SigQ$, \rlb{whose objects we call types, and the model defines a}\ functor $\semMfunc \colon \strucQ \to \catC$. 

For a causal model $\modelM$ with input variables $\Vin$ and output variables $\Vout$ there are various standard examples of queries. A basic example is the query $\modelMio = P(\Vout \mid \Vin)$ for its overall probability channel from inputs to outputs. Another is the query $P(\Vout \mid \Vin, \Do(Z))$, for a yet-to-be specified Do-intervention on variable $Z$, with $Z$ as an extra input, for which we introduce a graphical convention shown right-hand below. 

For a model $\modelM$ with input $\syn{X}$, outputs $\syn{Y,Z}$ and causal relations $\syn{X \to Y, Z}$, and $\syn{Z \to Y}$, the queries $P(Y, Z \mid X)$ and $P(Y, Z \mid X, \Do(Z))$ are represented by the 
following diagrams respectively, read from bottom to top. 
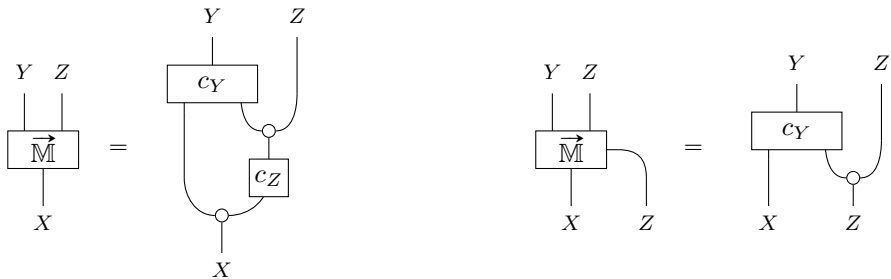
\[
% \semM{P(Y,Z \mid X)} \quad = \quad 
\input{CM-ex-2a.tikz}
\qquad \qquad \qquad \qquad 
% \semM{P(Y \mid X, \Do(Z))} \quad = \quad 
\input{CM-ex-3a.tikz}
\]
\maybe{In Section \ref{sec:comp-models} we define compositional models and queries, and spell out causal models as a special case.}

\paragraph{Abstraction.} Given a `low-level' model $\modelML$ and a `high-level' model $\modelMH$, an abstraction relation specifies, firstly, for each (list of) high-level type(s) $\syn{X}$ a list of low-level types $\HtoL(\syn{X})$ and a `surjective' (epic) morphism $\LtoH_{\syn{X}} \colon \HtoL(X) \to X$, which intuitively maps low-level values to high-level values. Secondly, it relates each high-level query $\syn{Q}_H$ to a collection of low-level queries $\syn{Q}_L$, such that the following holds in $\catC$. 
\begin{equation} \label{eq:abstr-fund-intro}
 \input{abstraction-fundamental.tikz}
%\tikzfig{abstraction-fundamental-nodots} %Optional-nodots
%\rlb{\tikzfig{abstraction-fundamental-nodots_LH}}
\end{equation}
The above \emph{consistency} equation is the fundamental rule of abstraction. It specifies that, up to the abstraction \rlb{maps}, the models' behaviours coincide, in that we can implement each high-level query on a low-level input either directly after $\LtoH$, or by first applying a corresponding low-level query and then $\LtoH$. A key observation is that, formally, \eqref{eq:abstr-fund-intro} precisely tells us that the maps $\LtoH$ amount to a natural transformation between functors representing the queries. There are two main types of abstraction, depending on the type of relation $(\syn{Q}_H, \syn{Q}_L)$. 

\begin{enumerate}
\item 
A \emph{\qdownfull{}} has a map $\HtoL \colon \SigQH \to \SigQL$ from high to low-level queries $\syn{Q}_L = \HtoL(\syn{Q}_H)$.
\item 
An \emph{\qupfull{}} has a partial map $\LtoHquery \colon \SigQL \to \SigQH$ from low to high-level queries $\synQ_H = \LtoHquery(\syn{Q}_L)$. 
\end{enumerate}
\maybe{In Section \ref{sec:abstraction} we give the full definitions of abstractions between models.}

A key example are \emph{constructive causal abstractions}, which are \qdownfull{}s between \rlc{causal models with the queries of the kind}\  $P(\Vout \mid \Vin, \Do(S))$, establishing consistency between Do-interventions at the low-level 
\rlb{$\Do(\HtoL(S))$ and high-level $\Do(S)$}.
Using our convention above, the consistency condition becomes: 
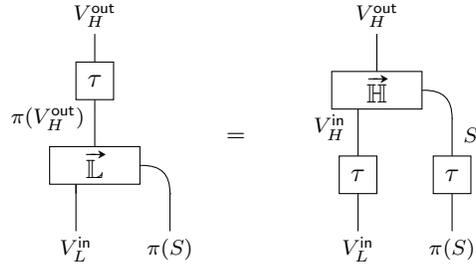
\[ 
% \tikzfig{causal-abstr-intro}
% \tikzfig{CA-intro-expl}
\input{CA-intro-expl-a.tikz}
\] 
For queries given by specific Do-interventions $P(\Vout \mid \Vin, \Do(X=x))$, there are typically many low-level queries which could achieve a given high-level one. Abstractions based on such concrete interventions, Do-style or more general, are instead described by \qupfull{}s, and referred to in the literature as `exact transformations'. In general, for any `abstract' \qdownfull{} at the structural level there is a corresponding notion of `concrete' \qupfull{}, we express this formally in Proposition \ref{prop:Exact-trans-from-ref}; \rlb{indeed, many known notions of causal abstraction can be captured in both forms.}

\maybe{
{\rlb{Section \ref{sec:causal-abstraction} discusses causal abstraction in detail}}; it goes through key notions from the literature, observes how these are special cases of the definitions in {\rlb{Section \ref{sec:abstraction}}} and further develops the theory of causal abstraction. 
These include exact transformations (in \ref{subsec:exact-intv}), constructive abstraction (in \ref{subsec:CCA}), \restrictedCCA{} (in \ref{subsec:weak-CA}), \CF{} abstraction (in \ref{subsec:CF-abstraction}) and a formal account of the  `distributed' abstractions treated implicitly in \cite{geiger2023causal} (in \ref{sec:distributed-abstraction}). 
}

\paragraph{Component-level abstraction.}

The categorical perspective suggests considering another stronger form of abstraction, acting not only on queries but \rlb{on}\ components directly.  By a \emph{\strucdown{}} we mean a \qdownfull{} equipped with a further functor $\HtoLS \colon \strucSH \to \strucSL$ sending each high-level component (such as a causal mechanism) $\syn{c}$ to a diagram of low-level components $\HtoLS(\syn{c})$, along with a natural transformation $\LtoH$ between functors at this component level, which respects queries in a suitable sense. 

In the simplest case of a \emph{strict} \strucdown{} we have $\LtoH=\id{}$ and so $c = \HtoLS(c)$ in $\catC$, simply writing each high-level component as an equivalent diagram of low-level components. Causal examples include refining a causal Bayesian network to a structural causal model (SCM), or implementing each \rlb{mechanism}\ directly as a neural network:  
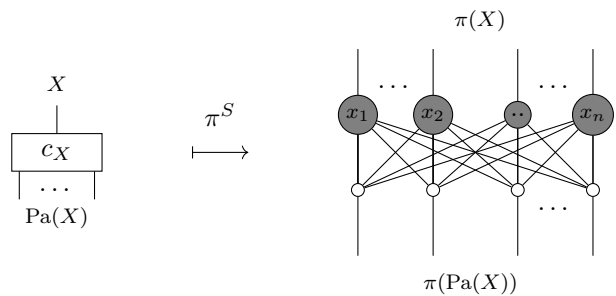
\[
\input{cX.tikz}
\qquad \quad 
\begin{tikzpicture}
	\begin{pgfonlayer}{nodelayer}
		\node [style=none] (0) at (-0.75, 0) {};
		\node [style=none] (1) at (0.75, 0) {};
		\node [style=none] (2) at (0, 1) {$\HtoLS$};
	\end{pgfonlayer}
	\begin{pgfonlayer}{edgelayer}
		\draw [style=maptostyle] (0.center) to (1.center);
	\end{pgfonlayer}
\end{tikzpicture}

% \mapsto
\qquad \quad 
% \tikzfig{DX2}
% , 
% \tikzfig{cX}
% \qquad 
% \overset{\HtoLS}\mapsto
% \tikzfig{piS}
% \qquad 
\input{DX4.tikz}
\]

More generally, for a (not necessarily strict) \strucdown{} the original component $c$ and diagram $\HtoLS(c)$
must satisfy a naturality condition $c \circ \LtoH = \LtoH \circ \HtoLS(c)$ in $\catC$. We call a constructive causal abstraction which extends to a \strucdown{} in this way a \emph{\strongCCA{}}. We believe this to be a new notion and show in Theorem \ref{Thm:strongCCAnew} that it corresponds to certain conditions on the partition $\HtoL$ of a constructive abstraction. In future work it would be interesting to understand the significance of \strongCCA{}s and how often they apply to \rlb{typical examples of abstractions.} 

\maybe{Section \ref{sec:structure-refinements} presents the new notion of \strucdownfull{} for compositional models {\rlb{and in particular discusses the case of \rlc{\strongCCA{}}\ in Section \ref{sec:strong-constructive-abstraction}}}.}

\paragraph{Quantum abstraction.}
Defining abstraction at the level of compositional models allows it to be applied beyond `classical' causal models, and in particular to compositional models with a \emph{quantum} semantics. 
\rlb{We can now consider}\
the case where our low-level compositional model $\modelML$ has \rlb{an input-output process given}\ by a quantum circuit such as in the simple example below, where thin wires denote classical inputs and outputs and thick wires denote quantum systems such as qubits. Here $E$ is an encoder from classical inputs to qubits, $U, V$ are unitary gates, and the upper boxes denote measurements.
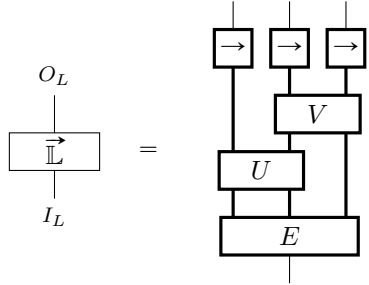
\[
\input{qcirc4.tikz}
\]
It is also possible to generalise both (abstract) Do-interventions, and weaker `interchange' interventions to quantum models defined by such circuits. We can then consider abstraction relations between a quantum circuit model $\modelML$ and classical causal model $\modelMH$ based on any such queries. Such abstractions may allow us to give classical (causal) \emph{explanations} of or \emph{interpretations} to aspects of quantum circuits, similarly to the use of abstraction in classical XAI; for example in quantum AI contexts where circuits arise as `black-boxes' from training. 

\maybe{
In Section \ref{sec:quantum-abstraction} we study abstraction for quantum compositional models, motivated by explainability in quantum AI. 
We close with a discussion of future directions in Section \ref{sec:discussion}. 
Some supplementary material on the theory and some of the proofs can be found in the appendix.}

\section{Categories and string diagrams} \label{sec:setup}

We will make use of category theory and its associated graphical language of \emph{string diagrams} \cite{coecke2006introducing,piedeleu2023introduction,selinger2011survey}. More specifically we will throughout work with a \emph{symmetric monoidal category (SMC)} $(\catC, \otimes)$ \rl{\cite{coecke2006introducing}}. 

Recall that a category consists of a collection of objects $\A, \B, \dots$ and morphisms or processes $f \colon \A \to \B$ between them. In string diagrams an object $\A$ is depicted as a wire and a morphism as a box, read from bottom to top. 
\[
\begin{tikzpicture}
	\begin{pgfonlayer}{nodelayer}
		\node [style=map] (0) at (0, 0) {$f$};
		\node [style=none] (1) at (0, 1) {};
		\node [style=none] (2) at (0, -1) {};
		\node [style=label] (3) at (0, 1.5) {$\B$};
		\node [style=label] (4) at (0, -1.5) {$\A$};
	\end{pgfonlayer}
	\begin{pgfonlayer}{edgelayer}
		\draw (2.center) to (1.center);
	\end{pgfonlayer}
\end{tikzpicture}
 
\]
For any morphisms $f \colon \A \to \B$ and $g \colon \B \to \C$ we can form their sequential composite $g \circ f \colon \A \to \C$, depicted as below. The identity morphism $\id{\A}$ on $\A$, with $\id{\B} \circ f = f  = f \circ \id{\A}$ for any $f \colon \A \to \B$ is depicted as a blank wire.
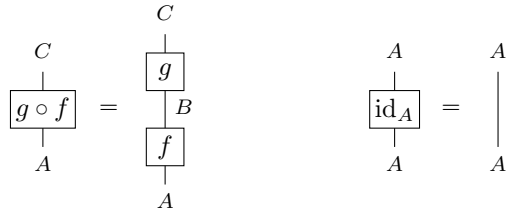
\[
\input{composite-1.tikz}
\qquad  \qquad \qquad 
\input{identity.tikz}
\]
In a monoidal category, for any pair of objects $\A, \B$ we can form their \emph{monoidal product} $\A \otimes \B$, with $\id{\A \otimes \B} = \id{\A} \otimes \id{\B}$ depicted by drawing wires side-by-side as shown left-hand below. More generally, for morphisms $f \colon \A \to \C$ and $g \colon \B \to \D$ we can form their monoidal product $f \otimes g \colon \A \otimes \B \to \C \otimes \D$, depicted center below. 
In general morphisms can have multiple inputs and/or outputs in diagrams, formally meaning their input is given as such a product. For example we depict a morphism $f \colon \A_1 \otimes \dots \otimes \A_n \to \B_1 \otimes \dots \otimes \B_m$ as shown right-hand below. 
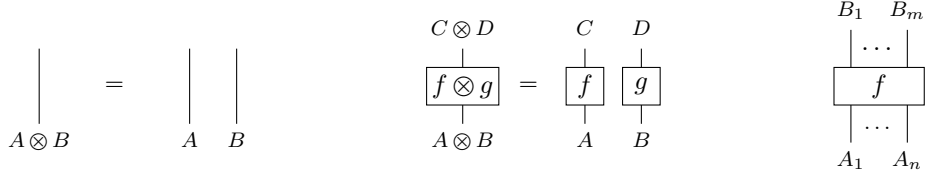
\[
\input{tensor-ob.tikz}
\qquad \qquad \qquad 
\input{tensor.tikz}
\qquad \qquad \qquad 
\input{fAB.tikz}
\]
Since our category is symmetric we also have morphisms $\tinyswap$ allowing us to `swap' pairs of wires over each other, such that $\tinyswap \circ \tinyswap = \id{}$, and boxes can `slide' along the swaps:  
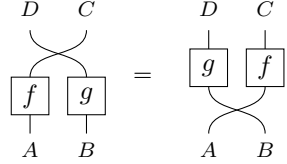
\[
% \tikzfig{doubleswap} 
% \qquad \qquad \qquad \qquad
\input{symmetry.tikz}
\]
Recall that there is also always a \emph{unit object} $I$, depicted simply as `empty space', where taking a monoidal product with $I$ simply leaves any object invariant.\footnote{Formally this is expressed via `coherence isomorphisms' $\A \otimes I \simeq \A \simeq I \otimes \A$.} We write $1 := \id{I}$. A noteworthy class of morphisms are \emph{states} of $\A$, which are morphisms $\omega \colon I \to \A$,  appearing with `no input'. 
\[
\begin{tikzpicture}
	\begin{pgfonlayer}{nodelayer}
		\node [style=map] (0) at (0, 0) {$\omega$};
		\node [style=none] (1) at (0, 1.25) {};
		\node [style=label] (2) at (0, 1.75) {$\A$};
	\end{pgfonlayer}
	\begin{pgfonlayer}{edgelayer}
		\draw (1.center) to (0);
	\end{pgfonlayer}
\end{tikzpicture}
 
\]

The composition in a category satisfies numerous other axioms which are self-evident in string diagrams; for example associativity of composition $(h \circ g) \circ f = h \circ g \circ f = h \circ (g \circ f)$ is enforced as all three composites amount to drawing $f, g, h$ in sequence.

\begin{definition}[Terminal category]\cite{coecke2008axiomatic,coecke2014terminality} 
An SMC $\catC$ is \emph{terminal} when $I$ is a terminal object. That is, for every object $A$ there is a unique morphism $A \to I$, which we call \emph{discarding} $\discard{A}$ and depict with a ground symbol:
\[
\begin{tikzpicture}
	\begin{pgfonlayer}{nodelayer}
		\node [style=upground] (0) at (4.5, 0.25) {};
		\node [style=none] (1) at (4.5, -1) {};
		\node [style=label] (3) at (4.5, -1.5) {$A$};
	\end{pgfonlayer}
	\begin{pgfonlayer}{edgelayer}
		\draw (1.center) to (0);
	\end{pgfonlayer}
\end{tikzpicture}

\]
% $\discard{A} \colon \A \to I$ with a ground symbol and call it \emph{discarding}. 
% $I$ is a terminal object. That is, there is a unique morphism $\discard{A} \colon \A \to I$ on each object $A$, depicted with a ground symbol and called \emph{discarding}. 
% \[
% \tikzfig{discarding}
% \]
We call a morphism obtained by discarding any output(s) of a morphism $f$ a \emph{marginal} of $f$. 
\end{definition}

%Now, the categories we consider here will typically come with further structure allowing us to `ignore' or `throw away' objects \cite{coecke2008axiomatic,coecke2014terminality}. 

% \begin{definition} \label{def:d-category}
% By a  \emph{discard-category (d-category)} we mean a SMC in which every object $A$ comes with a chosen \emph{discarding} morphism $\discard{\A} \colon \A \to I$, depicted with a `ground' symbol, 
% such that that the following hold: 
% \begin{equation} \label{eq:disc-nat}
% \tikzfig{disc-nat} \qquad \qquad  \qquad \tikzfig{disc-I}
% \end{equation}
% We call a morphism obtained by discarding any output(s) of a morphism $f$ a \emph{marginal} of $f$. 
%  A morphism $f$ is called a \emph{channel}\footnote{Channels are also known as `causal' \cite{coecke2014terminality} or `total' \cite{cho2015introduction} morphisms.} when it preserves discarding: 
% \[
% \tikzfig{causal2}
% \] 
% We call a d-category \emph{terminal} when every morphism is a channel. Equivalently, $\discard{\A}$ is the unique morphism $\A \to I$ for each object $\A$, making $I$ a \emph{terminal object}. 
% \end{definition}

We will also call morphisms in a terminal category \emph{channels}. 

Our main example categories in this article will be terminal. The final piece of extra structure we will use, present for `classical' \rl{(as opposed to quantum)}\ processes, such as those of causal models, is the following.

% \begin{definition}
% \cite{cho2019disintegration}\label{def:cd_category}
% A \emph{cd-category} (\emph{copy-discard category}) is a discard-category in which every object $A$ comes with a distinguished \emph{copying} morphism $\tinycopy \colon \A \to \A \otimes \A$ 
% satisfying the following:
% \[
% \tikzfig{markov-axioms}
% \] 
% \begin{equation} \label{eq:copy-nat-rules}
% \tikzfig{copy-nat}
%  % \qquad \qquad \tikzfig{disc-nat} \qquad \qquad \tikzfig{disc-I}
% \end{equation}
% A \emph{Markov category} is a terminal cd-category \cite{fritz2020synthetic}. In any cd-category we call a morphism $f \colon \A \to \B $ \emph{deterministic} when it satisfies the following. 
% \begin{equation} \label{eq:deterministic}
% \tikzfig{deterministic}
% \end{equation}
% In particular, a state $\x$ of $\A$ is then deterministic precisely when it is copied by $\tinycopy$; we also call such states \emph{sharp}. 
% \begin{equation} \label{eq:copy-points}
% 	\tikzfig{copy-points}
% \end{equation} 
% \end{definition}

\begin{definition}[Markov category]
\cite{fritz2020synthetic} \label{def:Markov_category}
% \cite{cho2019disintegration}
A \emph{Markov category} is a terminal SMC $\catC$ in which every object $A$ comes with a chosen \emph{copying} morphism $\tinycopy \colon \A \to \A \otimes \A$ satisfying:
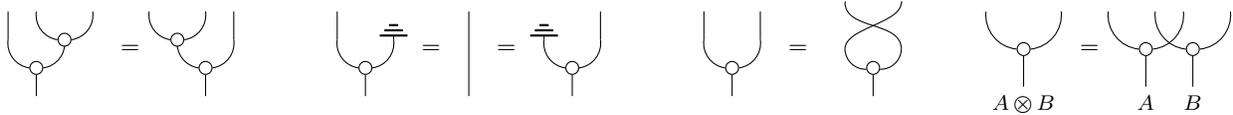
\[
%\tikzfig{markov-axioms}
\input{markov-2.tikz}
\] 
% \begin{equation} \label{eq:copy-nat-rules}
% \tikzfig{copy-nat}
%  % \qquad \qquad \tikzfig{disc-nat} \qquad \qquad \tikzfig{disc-I}
% \end{equation}
%A \emph{Markov category}is a terminal SMC in which every object $A$ comes with a distinguished \emph{copying} morphism $\tinycopy \colon \A \to \A \otimes \A$
%  depicted as: 
% \[
% \tikzfig{just-copy}
% \]
% and s
%satisfying the following:
% \[
% \tikzfig{markov-axioms-alt}
% \] 
%More generally, a \emph{cd-category} is an SMC with discarding and such copy morphisms. 
%A \emph{Markov category} is a terminal cd-category. 
% and that: 
% Formally, this says that copying and discarding form a `commutative comonoid'. The copying morphisms are moreover required to be `natural' in that the following holds for all objects $A, B$.
% \begin{equation} \label{eq:copy-nat-rules}
% \tikzfig{copy-nat}
%  % \qquad \qquad \tikzfig{disc-nat} \qquad \qquad \tikzfig{disc-I}
% \end{equation}
% precisely when it is copied by $\tinycopy$; we also call such states \emph{sharp}. 
% % We also call deterministic states $x$ \emph{sharp}, satisfying the following.
% % In particular every state $x$ is deterministic, also called \emph{sharp}, meaning it is copied by the copy morphism:
% \begin{equation} \label{eq:copy-points}
% 	\tikzfig{copy-points}
% \end{equation}
% A terminal cd-category in which every morphism is a channel is called a \emph{Markov category} 
\end{definition}
In such a category we call a morphism $f \colon \A \to \B $ \emph{deterministic} when it satisfies the following. 
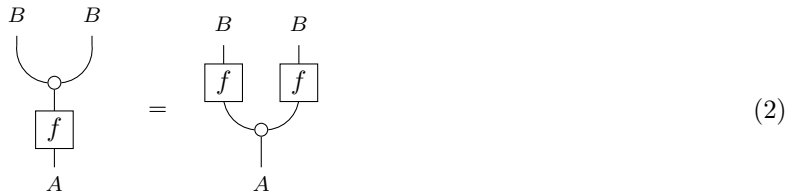
\begin{equation} \label{eq:deterministic}
\input{deterministic.tikz}
\end{equation}
In particular, a state $\x$ of $\A$ is deterministic or \emph{sharp} precisely when it is copied by $\tinycopy$:
\begin{equation} \label{eq:copy-points}
	\input{copy-points.tikz}
\end{equation}

For causal models, the only category required in this article is the following.

\begin{example}[$\FStoch$]
In the Markov category $\FStoch$ the objects are finite sets $X, Y, \dots$. A morphism $f \colon X \to Y$ is a probability channel from $X$ to $Y$, specifying a probability $f(y \mid x) \in [0,1]$ for each $x \in X$ and $y \in Y$, such that: 
\[
\sum_{y \in Y} f(y \mid x) = 1
\] 
for all $x \in X$. For any $f \colon X \to Y$ and $g \colon Y \to Z$, the sequential composition $g \circ f \colon X \to Z$ is given by \rl{matrix composition:}
\[
(g \circ f) (z \mid x) := \sum_{y \in Y} g(z \mid y) f (y \mid x)
\]
The identity channel on $X$ is given by $\id{}(y \mid x) = \delta_{x,y}$. The monoidal product is given on objects by $X \otimes Y = X \times Y$, and on channels $f \colon X \to W$ and $g \colon Y \to Z$ by 
\[
(f \otimes g)(w, z \mid x, y) = f(w \mid x) g(z \mid y)
\]
corresponding to two independent probability channels. Here $I=\{\star\}$ a singleton, so that states $\omega$ on $X$ correspond to distributions over $X$ via $\omega(x) := \omega(x \mid \star)$.   

The copy morphism on $X$ is $\tinymultflip[whitedot](y,z \mid x) = \delta_{x,y,z}$.  A state is sharp precisely when it is a point distribution $\delta_x$ for some $x \in X$. We denote such a state simply by $x$, so that sharp states $x$ of $X$ correspond to elements of $X$. More generally, deterministic morphisms $f \colon X \to Y$ correspond to functions $f \colon X \to Y$ via $f \circ x = f(x)$. Discarding $\discard{X}$ on $X$ corresponds to the unique function $X \to I$.

Many notions from probability theory can be described in terms of string diagrams in $\FStoch$ and more general Markov categories \rl{\cite{fritz2020synthetic}}.\footnote{When carrying out such probabilistic reasoning it can also be useful to work in the broader category $\MatR$ of positive real matrices, with $\FStoch$ being the subcategory of channels therein, see for example \cite{lorenz2023causal}; here the latter category suffices.} In particular, a morphism as below is precisely a conditional probability distribution $P(Y_1,\dots,Y_m \mid X_1,\dots, X_n)$:  
\[
\input{condprob.tikz}
\] 
while composition with discarding amounts to marginalisation in the usual sense: 
\[
\input{FStoch-marginal_v2.tikz}
=
\begin{minipage}{9.8cm} \centering $\sum\limits_{ \begin{minipage}{1.5cm} \centering \tiny $y_k,...y_m \in$ \\[0.05cm] $Y_k\times ... \times Y_m$  \end{minipage}  }$ {\small $P(Y_1,\dots,Y_{k-1} , Y_k=y_k,\dots,Y_m=y_m\mid X_1,\dots, X_n)$} \end{minipage}
\]
\end{example}

% \paragraph{Convention.}
% In what follows, categories are always taken to be of some fixed class, either that of terminal SMCs or Markov categories.

% Recall that a functor $F \colon \catC \to \catD$ consists of a mapping on objects $\A \mapsto F(\A)$ and morphisms $(f \colon \A \to \B) \mapsto F(f) \colon F(\A) \to F(\B)$, respecting composition and identities. Here we will always mean one which is \emph{strong symmetric monoidal}, preserving $I, \otimes$ and $\tinyswap$, and in the Markov case also preserving copy morphisms and with deterministic structure isomorphisms. 

% A \emph{natural transformation} $\alpha \colon F \implies G$ is given by a family of morphisms $\alpha_\A \colon F(\A) \to G(\A)$ such that \ $\alpha_\B \circ F(f) = G(f) \circ \alpha_\A$ for all $f \colon \A \to \B$. We will always mean one which is \emph{monoidal}, with $\alpha_{A \otimes B} = \alpha_A \otimes \alpha_B$, and each $\alpha_A$ deterministic in the Markov case.

\paragraph{Functors and natural transformations.} 
In what follows, categories are always taken to be of some fixed class, either that of terminal SMCs or Markov categories. Recall that a \emph{functor} $F \colon \catC \to \catD$  between categories $\catC, \catD$ is given by a mapping on objects $\A \mapsto F(\A)$ and morphisms $(f \colon \A \to \B) \mapsto F(f) \colon F(\A) \to F(\B)$ from $\catC$ to $\catD$, respecting composition and identities. Here we will always mean one which is \emph{strong symmetric monoidal}, meaning that it preserves $I, \otimes$ and $\tinyswap$. In the Markov case we also always mean that functors preserve copy morphisms and have deterministic structure isomorphisms.%$F$ is \emph{(strong) symmetric monoidal} when it preserves $I, \otimes$ and $\tinyswap$. By a \emph{(c)d-functor} between (c)d-categories we mean a strong symmetric monoidal functor which moreover preserves discard morphisms (and copy morphisms).
\footnote{In more detail, functors satisfy $F(g \circ f) = F(g) \circ F(f)$ and $F(\id{}) = \id{}$. Strong monoidal functors comes with an isomorphism $F(I) \simeq I$ and natural isomorphisms $F(\A \otimes \B) \simeq F(\A) \otimes F(\B)$ satisfying coherence conditions. A Markov functor is such that the equations \eqref{eq:functor-rules} with $\otimes$ and $\tinycopy$ hold up to structure isomorphisms.
%in for all objects $\A$, up to its structure isomorphisms, 
%$F(\discard{\A}) = \discard{F(\A)}$, and
 %$F(\tinycopy_\A) = \tinycopy_{F(\A)}$ in the cd-case, for all objects $A$.
  In diagrams we omit to draw the structure isomorphisms.} 
\begin{equation} \label{eq:functor-rules}
\input{Ffgsimpler.tikz}
\qquad \qquad
\input{Fftensg.tikz}
\qquad \qquad 
%F\left( \tikzfig{tinycopy} \right) = \tikzfig{tinycopy}
F\left( \input{tinycopy_withobject.tikz} \right) = \input{tinycopy_F_object.tikz}
%\qquad
%F\left( \tikzfig{tinydisc} \right) = \tikzfig{tinydisc}
%F\left( \tikzfig{tinydisc_with_object} \right) = \tikzfig{tinydisc_F_object}
\end{equation}

Given functors $F, G \colon \catC \to \catD$ a \emph{natural transformation} 
$\alpha \colon F \implies G$ is given by a family of morphisms $\alpha_\A \colon F(\A) \to G(\A)$, in diagrams all written simply as $\alpha$ due to their wire labels, such that for all $f \colon A \to B$ we have: 
\begin{equation} \label{eq:naturality}
% \tikzfig{alpha} 
% \qquad \qquad 
% \text{such that}
% \qquad \qquad 
\input{naturality.tikz}
\end{equation}
We will always implicitly mean natural transformations which are \emph{monoidal} meaning that:
%\emph{monoidal}, with $\alpha_{A \otimes B} = \alpha_A \otimes \alpha_B$, and each $\alpha_A$ deterministic in the Markov case.
%A \emph{monoidal} natural transformation between strong monoidal functors is one that respects $\otimes$, in that:
\footnote{And also $I$ in that $\alpha_I \colon F(I) \to G(I)$ is the composite of the structure isomorphisms for $F, G$.}
\[
% \tikzfig{emptydiagnat} \qquad \qquad \qquad 
% \qquad \qquad 
\input{monnat.tikz} 
\]
and with each $\alpha_A$ deterministic in the Markov case. %Finally, a \emph{(c)d-natural transformation} between (c)d-functors is such that furthermore each component $\alpha_\A$ is a (deterministic) channel. 

\begin{remark}
Intuitively, for each kind of category the corresponding kind of functor $F \colon \catC \to \catD$ maps any diagram in $\catC$ to one of the same shape in $\catD$, preserving $\discard{}$ and $\tinycopy$ when they are present.
\[
\input{functor-map.tikz}
\]
A natural transformation $\alpha$ has the necessary properties allowing us to  `pull' it through any such diagrams of the form $F(f)$, $G(f)$, as below. 
\[
% \tikzfig{nat-rewrite-long}
\input{nat-rewrite-short.tikz}
\]
\end{remark}

\section{Compositional models and queries} \label{sec:comp-models}

We can now define the general forms of models we will consider in this article. Firstly, we begin by specifying their abstract structure, as follows.

\begin{definition}[Signature]
A \emph{(monoidal) signature} \rl{$\Gensig$}\ consists of:
\begin{itemize}
\item 
a set \rl{$\Gensigob$}\ of abstract `objects' $\syn{A,B,C,\dots}$;
\item 
a set \rl{$\Gensigmor$}\ of abstract `morphisms' $\syn{f} \colon (\syn{A_i})^n_{i=1} \to (\syn{B_j})^m_{j=1}$, where each morphism is given with a lists of input \rl{$(\syn{A_i})^n_{i=1}$ and output $(\syn{B_j})^m_{j=1}$}\ objects from \rlc{$\Gensigob$}. 
\end{itemize}
\end{definition} 

We fix some class of monoidal categories and associated string diagrams, either that of terminal SMCs or Markov categories.\footnote{In Section \ref{sec:strong-constructive-abstraction} we also consider cd- and Cartesian categories.} Any signature $\Gensig$ then freely 
generates a \rl{category $\Gensigcat$ of that kind}: objects are lists from \rl{$\Gensigob$}\ and morphisms $\diagD \colon (\syn{A_i})^n_{i=1} \to (\syn{B_j})^m_{j=1}$ are \rl{string diagrams of the chosen kind}\ built from \rl{$\Gensigmor$}\ with inputs $(\syn{A_i})^n_{i=1}$ and outputs $(\syn{B_j})^m_{j=1}$.\footnote{Composition is composing diagrams `on the page', the empty list is the unit and each identity is the diagram of plain wires. A signature can also include a set \rl{$\Gensigeq$}\ of \emph{equations} $\diagD_1 = \diagD_2$  where $\diagD_1, \diagD_2$ are string diagrams built from the generators, but we will not focus on this here. Diagrams are equivalent when they can be re-written \rlb{into each other}\ using $\Gensigeq$. Models should preserve any equations in $\Gensigeq$.} 
We consider two diagrams $\diagD, \diagD'$ equal as morphisms if one can be rewritten to the other in finitely many steps using the rules of such diagrams.

We can now define what it means to model such structure in a category. The following notion takes centre stage in this work.

\begin{definition}[\rl{Model}]
A \rl{\emph{model} $\modelM$ of a signature $\Gensig$ in a \emph{semantics}}\ category $\catC$ consists of a functor: 
\[
\modelMgenfunc \colon \rl{\Gensigcat} \to \catC
\]
Equivalently, this consists in specifying an object $A := \modelMgen{\syn{A}}$ of $\catC$ for each $\syn{A} \in \rl{\Gensigob}$ and a morphism $f := \modelMgen{{\syn{f}}}$ in $\catC$ for each $\syn{f}$ in \rl{$\Gensigmor$}, of the appropriate type as below.
%\footnote{Moreover the functor should preserve any equations in \rl{$\Gensigeq$}.} 
\begin{equation} \label{eq:generator-map} 
\input{gen-map-2.tikz}
\end{equation}
\end{definition} 

In other words, a \rl{model}\ chooses specific objects and morphisms in $\catC$ corresponding to the `abstract' ones in \rl{$\Gensig$}. As above we typically omit to write $\modelMgenfunc$ and use distinct fonts to distinguish any object $\syn{A}$ or morphism $\syn{f}$ in $\Gensigcat$ from its semantics $A, f$ in $\catC$.

\rlb{The}\ models we consider here will in fact typically come with \emph{two} distinct, \rlb{but related, signatures. Although}\ both simply cases of the above definition, it is helpful to introduce distinct terminology for them.

\begin{definition}[Compositional Model]
By a \emph{compositional model} $\modelM$ we mean a signature $\SigS$ along with a 
\rl{model of $\SigS$}\ in a \emph{semantics} category $\catC$:
\[
\rl{\semcompM{\modelM}{-}} \colon \strucS \to \catC
\] 
where we call $\SigS$ the \emph{structure} of the model, the elements of $\Sigob$ the \emph{variables}, and those of $\Sigmor$ the \rl{\emph{components}}\ of the model. 
\end{definition} 

Though simply terminology, referring to a model as a compositional model is to convey that its signature $\Sig$ captures \emph{compositional structure}, i.e.~the basic underlying components, mechanisms or building-block processes of interest to the phenomenon being modelled. Typically there \rlb{also are}\ further processes of interest constructed from these, which the aim of the model is to compute; for example specific questions for which the model is to provide the answers. We will refer to these as `queries', organised into another signature as follows. 

\begin{definition}[Model of Queries]
By a \rlb{\emph{signature of queries}}\ we simply mean a signature $\SigQ$ whose morphisms we refer to as \emph{queries}. We call (lists of) objects in $\SigQ$ \emph{query types}. By a \emph{model of queries} we simply \rlb{mean}\ a model $\modelM$ of $\SigQ$ in $\catC$, with \rlc{its}\ functor denoted:
\[
\semMfunc_{\modelM} \colon \strucQ \to \catC
\]
By a \emph{compositional model of queries} $(\modelM, \SigS, \SigQ, \catC)$ we simply mean a compositional model $\semcompM{\modelM}{-}$ of structure $\SigS$ in $\catC$ along with a model $\semMfunc_{\modelM}$ of queries $\SigQ$ in $\catC$. 
\end{definition}

Thus again, each query specifies an abstract input-output process $\syn{Q}$, which a model $\modelM$ gives semantics $Q=\semM{\syn{Q}}_\modelM$ denoted by changing font. 
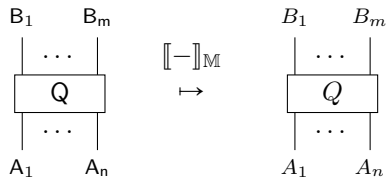
\[
\input{gen-map-queries.tikz}
\]

\rlb{Typically, we will work with a compositional model of queries $(\modelM, \SigS, \SigQ, \catC)$, rather than just a model of queries. There may of course be many compositional models of the same \rlb{queries}\ $\SigQ$, but with distinct structures $\SigS$. Moreover, for any such tuple $(\modelM, \SigS, \SigQ, \catC)$}, 
the two functors are usually not arbitrary but closely related; a particularly strong link between \rlc{them}\ is when they factorise as follows.

\begin{definition}[Abstract Queries] \label{def:abstract-query}
For any \rl{compositional model of queries $(\modelM,\SigS, \SigQ, \catC)$, we call the queries $\SigQ$}\ \emph{abstract} when they come with a \rl{functor $\qabs{-} \colon \strucQ \to \strucS$}, sending each query to a diagram in $\strucS$, such that the following commutes.  
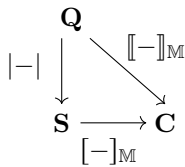
\[
\input{abstract-query-4.tikz}
\] 
\end{definition}

\rlb{Given a compositional model of queries $(\modelM, \SigS, \SigQ, \catC)$, even}\ 
when queries are not technically abstract, the computation of each query $\semM{Q}_{\modelM}$ will typically make use of the structure $\strucS$ and semantics $\semcompM{\modelM}{-}$ of the compositional model. Finally, when in context the particular model $\modelM$ is clear we may drop the subscripts and just write $\semcompM{}{-}$ and $\semMfunc_{}$ for the semantics functors of structure and queries, respectively.

\subsection{Causal models} \label{subsec:causal-models}

Let us now illustrate our set-up via our prime examples of compositional models and queries; those of causal models. \rl{The following presentation will closely follow that of \cite{lorenz2023causal} (and \cite{jacobs2019causal}).}\  

By an \emph{open DAG} \rl{$\dagG=(\dagG,\Vin)$ we mean a DAG $\dagG$ over vertices $\VG$ along with subset $\Vin \subseteq \VG$ of \emph{input} vertices}\ such that each input has no parents in $\dagG$. We then define a signature $\SigS_\dagG$ containing an object, \st{i.e.~variable}, \rlb{$\syn{X}$ for each vertex $\syn{X} \in \VG$ and for each non-input vertex $X \in \Vnin := \VG \setminus \Vin$}\ a generator called the \rl{(causal)}\ \emph{mechanism} for $\syn{X}$: 
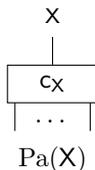
\[
\input{causal-generator.tikz}
\]
where $\Pa(\syn{X})$ denotes the parents of $\syn{X}$ in $G$. We write $\strucS_G$ for the free Markov category generated by $\SigS_G$, i.e.~consisting of all diagrams we can build from these generators along with copy and discard maps.

\begin{definition} \label{def:causal-model}
Let $\catC$ be a Markov category and $G$ an open DAG. An \emph{(open) causal model} $\modelM$ of $G$ in $\catC$ is given by a compositional model of $\SigS_G$ in $\catC$, along with a specification of 
\rlb{output variables $\Vout \subseteq \VG$}\footnote{The output variables may optionally be thought of as the variables considered `observed'.}. 
\end{definition}

A causal model $\modelM$ of $G$ in $\FStoch$ amounts to specifying a finite set $X := \semM{\syn{X}}$ of values for each variable $\syn{X}$ as well as a probability channel $c_X \colon \Pa(X) \to X$ for each  
\rlb{non-input variable $X \in \Vnin$}, more conventionally denoted:\footnote{\rl{For $\catC = \FStoch$ we adopt the common notation of $P( - | - )$, but the mechanism $c_X $ really just is a stochastic channel, which does not have to arise as a conditional from a joint distribution; also see \cite{lorenz2023causal}.}} 
\[
P(X \mid \Pa(X))
\]
Hence such a model is precisely an \emph{(open) causal Bayesian network (CBN)} for the DAG $G$, the typical notion of causal model on finite sets. Here `open' refers to the fact that $G$ can come with inputs, for which no mechanism is specified.\footnote{\rl{Although the definition of \emph{open} causal model from \cite{lorenz2023causal} that has inputs is not standard in the literature, it is actually very natural seeing as, e.g., `unspecified' exogenous variables or do-interventions give rise to such open models (see below).}} 

The following will also be useful. By a set of \emph{input-output variables} $\syn{V}$ we mean a set of (abstract) variables given with subsets 
\st{$\syn{\Vin}, \syn{\Vout} \subseteq \syn{V}$, }\
of \emph{input} and \emph{output} variables, respectively. 
By a set of \emph{concrete variables} $\modelV=(\syn{V}, \sem{-})$ in $\catC$ we mean a set of input-output variables $\syn{V}$ along with a given semantics $X=\sem{\syn{X}}$ in $\catC$ for each $\syn{X} \in \syn{V}$.\footnote{\rl{The notation reflects that $\modelV$ is formally a model of $\syn{V}$ seen as a signature with no morphisms. Concrete variables $\modelV$ correspond to what is sometimes called a `signature' in the causal model literature.}} At times we write simply $V$ for $\modelV$. A \emph{causal model $\modelM$ over $\modelV$} is then one for some $\SigS_G$ with variables $\syn{V}$ and inputs $\syn{\Vin}$, outputs $\syn{\Vout}$, and where $X = \sem{\syn{X}}_{\modelM}$ for all $X \in \syn{V}$.

\begin{example} \label{ex:open-DAGs} 
Consider the open DAG $G$ over $\syn{V=\{A,B,C,D,E\}}$ with \rl{inputs $\syn{\Vin=\{B,C\}}$}, 
shown on the LHS below, \rlb{where we used arrows `without a source' to indicate the subset $\syn{\Vin}$}. 
\rlb{In conventional presentation, an open}\ CBN with this causal structure \rlb{specifies}\ (finite) sets $A,B,C,D,E$ along with the probability channels listed on the RHS, \rl{as well as a subset of variables considered outputs, e.g. \rlb{$\syn{\Vout = \{C,D,E\}}$.}} 
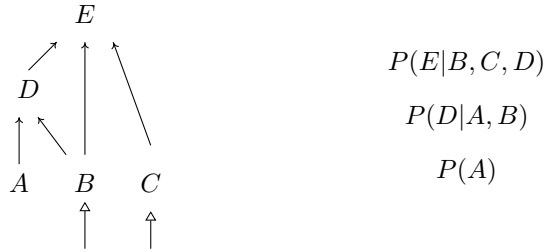
\[
%\tikzfig{open-DAG-rotate3} 
\input{open-DAG-rotate3_v2.tikz} 
\qquad 	\qquad 	
\begin{minipage}{5cm}
			\centering
			$P( E | B, C, D)$ \\[0.3cm]
			$P(D | A, B)$ \\[0.3cm]
			$P(A)$	
\end{minipage}
\] 
\rlb{In the language of compositional models, the}\ 
corresponding signature $\Sig_G$ consists of objects (variables) $\syn{A,B,C,D,E}$ and morphisms (mechanisms): 
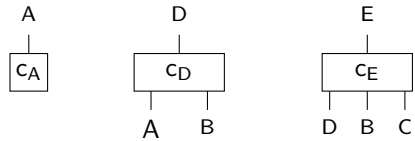
\[
\input{ex-mechanisms.tikz}
\]
A causal model $\modelM$ of $G$ in a Markov category $\catC$ is given by a functor $\semMfunc$ sending each variable to \rlc{an object}\
$A=\semM{\syn{A}}$, B = \semM{\syn{B}}, \dots and each mechanism to its corresponding 
\rlb{channel, i.e. for $\catC = \FStoch$, for instance, $\semM{\syn{c_D}} = c_D =  P(D \mid A, B)$}.  
The free Markov category $\strucS_G$ generated by $\Sig_G$ contains as morphisms all Markov diagrams we can build from these components, such as: 
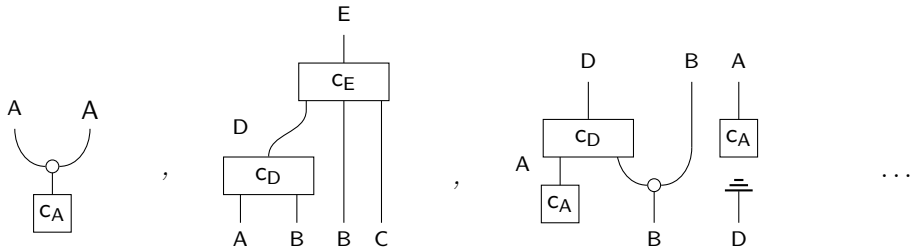
\[
\input{ex-diags.tikz}
\]
The functor $\semMfunc$ sends any such diagram $\diagD$ to a channel $\semM{\diagD}$ from its inputs to its outputs.
\end{example}

\subsection{Causal queries}

For a causal model there are various possible types of queries, corresponding to kinds of `causal questions' we can ask, usually based on interventions. These queries are `answered', i.e. modelled in $\catC$, by computing corresponding channels from the model's mechanisms. We will \rlb{now}\ introduce several examples of queries for causal models, starting with the most basic.

\subsubsection{Input-output queries}

Most simply, we can consider as a query the overall channel from inputs $\Vin$ to any subset $\syn{O \subseteq \Vout}$ of outputs. For a CBN this is the marginal $P(O \mid \Vinconc)$ on $O$ of the overall channel: 
\begin{align}
%\rl{P(\Ven \mid \Vin)} := \prod_{X \in \rl{\Ven}} P(X \mid \Pa(X))
%\rlb{P(\Vout \mid \Vin) \ := \ \Big( \prod_{X \in \Vout \setminus \Vin} P(X \mid \Pa(X)) \Big) \ \Big( \prod_{X \in \Vout \cap \Vin} \delta_{X,X} \Big) \ , }
\rlb{P(\VGconc \mid \Vinconc) \ := \ \Big( \prod_{X \in \Vnin} P(X \mid \Pa(X)) \Big) \ \Big( \prod_{Y \in \Vin} P(Y \mid Y) \Big) \ , } \label{eq:CBN-io-channel}
\end{align}
\rlb{where the second factor could be dropped, seeing as $P(Y \mid Y) = \delta_{Y,Y}$ but is used to emphasise the fact that we in general allow for $\Vout \cap \Vin \neq \emptyset$, in which case the corresponding inputs are `copied out' and regarded outputs, too. 
More generally, for}\ 
any causal model \rl{$\modelM$}\ over $\dagG$ and subset \rl{$O\subseteq \Vout$ of outputs}, $\SigS_G$ contains a diagram called the (normalised) \emph{network diagram} $\netdiag{\modelM} \colon \syn{\Vin} \to \syn{O}$, 
%\rl{\netdiag{G,O}}
given by wiring together the mechanisms of every variable $\syn{X}$ with a path to $\syn{O}$ in $\dagG$ using copy maps. Each $\syn{X} \in \Vin$ is an input to the diagram and each $\syn{X} \in O$ is copied an additional time to form an output to the diagram.\footnote{In contrast to~\cite{lorenz2023causal} here we use the normalised network diagram, where every \rl{non-input}\  variable has a path to an output. 
\rl{This in particular means that any input $X \in \Vin$ that in $G$ has no path to $O$, in $\netdiag{\modelM}$ 
is the input into a disconnected discard $\discard{X}$. 
Note that leaving $O$ implicit in writing $\netdiag{\modelM}$ to avoid clutter is harmless, because $O$ will always be explicit as type in diagrams.}}   
We denote the resulting channel $\modelMio := \semcompM{\modelM}{\netdiag{\modelM}}$ in $\catC$ as:
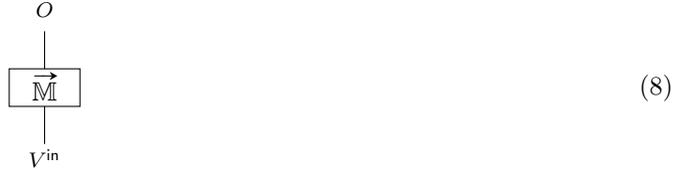
\begin{equation} \label{eq:query-map-explicit}
% \tikzfig{query-expl}
% \tikzfig{query-expl-simpler}
% \tikzfig{Mio2}
\input{Mio2-nodots.tikz} %Optional-nodots
\end{equation} 
For a CBN, $\modelMio$ is then precisely the probability channel $P(O \mid \Vin)$. See the below example and for full details \rl{\cite{jacobs2019causal, lorenz2023causal, fritz2023free}}.

Formally, we define a \rlb{signature of queries}\ $\ioqueries(\syn{V})$ with types $\syn{V}$ and a query $\syn{(\ioquery{O}{\Vin}) \colon \Vin \to O}$ for each subset $\syn{O \subseteq \Vout}$. Then $\modelM$ models these as abstract queries via  
$\qabs{V} = V$,  $\qabs{(\ioquery{O}{\Vin})} = \netdiag{\modelM}$,    
so that $\semM{(\ioquery{O}{\Vin})} = \eqref{eq:query-map-explicit}$.

\begin{example}
For the DAG $G$ and model $\modelM$ in Example \ref{ex:open-DAGs}, the network diagram 
for \rlb{$\syn{O}=\syn{\{E,C\}} \subset \syn{\Vout}$}\ is: 
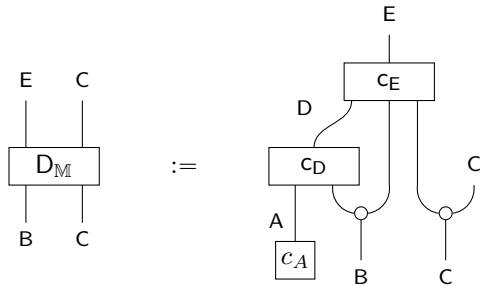
\[
\input{netdiag.tikz} 
\qquad  := \qquad 
\input{net-diag-simpler.tikz}
\]
Hence the corresponding query $\syn{(\ioquery{O}{\Vin})} \colon \syn{\Vin} \to \syn{O}$ in $\ioqueries(\syn{V})$ 
is mapped to the equivalent composite $\modelMio$ of channels in $\catC$ (note the \rlb{change of font}\ representing passing to specific channels in $\catC$): 
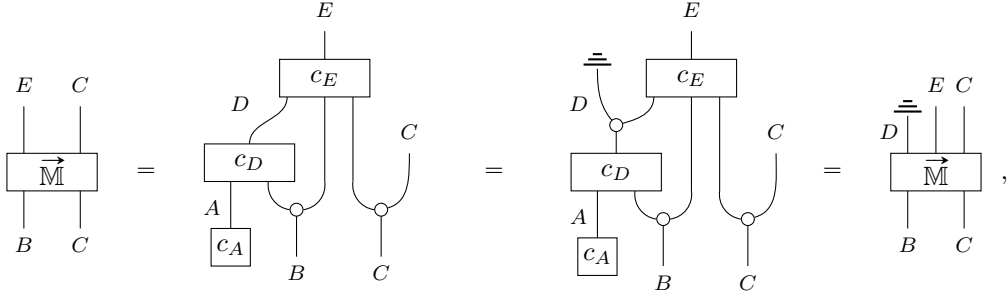
\[
%\tikzfig{query-io-simplest}
\rlb{\input{query-io-simplest-added-marginalisation.tikz} \ ,}
\]
\rlb{where we have also shown how it arises as a marginal from $\modelMio$ for the full $\syn{(\ioquery{\Vout}{\Vin})}$}. 
For an open CBN in $\catC=\FStoch$, we have $\modelMio =  P(E,C \mid B ,C)$, and the above is equivalent to the formula \rlb{(compared with Eq.~\eqref{eq:CBN-io-channel}, dropping the factor $P(C \mid C)$)}: 
\[
%P(E, C \mid B, C) = \rl{\sum_{A,D} \ P(E \mid D,B,C) \ P(D \mid A,B) \ P(A)}
P(E, C \mid B, C) = \rl{\sum_{a\in A,d\in D} \ P(E \mid D=d,B,C) \ P(D=d \mid A=a,B) \ P(A=a)}
\]
\end{example}

\subsubsection{General interventions as queries} \label{subsub:general-int}

Given \rlb{any causal}\ model $\modelM$ in $\catC$, an \emph{intervention} $\inti$ on $\modelM$ is a specification of a new \rlb{causal}\ model $\transform{\inti}{\modelM}$ with the same variables, inputs and outputs, including their \rl{semantics}\ $V = \semM{\syn{V}}$ in $\catC$. We can think of the intervention as a map 
assigning new mechanisms $\{c'_{X}\}_{X \in A_{\inti}}$, to a subset of variables \st{$A_{\inti} \subseteq \Vint$}.\footnote{\rl{Seeing as this work focuses on acyclic causal models, the new mechanisms that specify an intervention on a given model $\modelM$ are unconstrained apart from that they mustn't introduce directed cycles in the causal structure.}} 
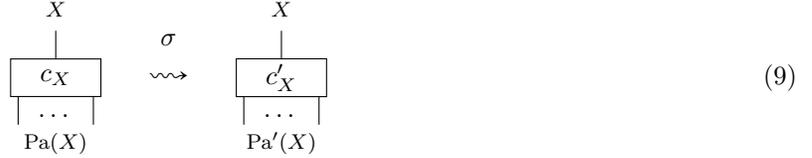
\begin{equation} \label{eq:intervention}
\input{intervention.tikz}
\end{equation}

Formally, 
given concrete variables $\modelV$, a set of interventions $\Intset$ can be seen as a query signature with types $\syn{V}$ and a query $\inti \colon \syn{\Vin} \to \syn{\Vout}$ for each  $\inti \in \Intset$ (corresponding to the data $\{c'_{X}\}_{X \in A_{\inti}}$). 
Any causal model $\modelM$ over $\modelV$ for which all $\sigma \in \Intset$ are well-defined interventions, that is do not introduce directed cycles, yields a model of $\sigfont{I}$ in $\catC$ via:  
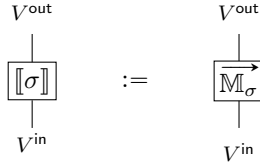
\[ 
	% \tikzfig{inti-query-3_b}
		\input{inti-query-3_b-nodots.tikz}
\]
For later use we note that there is a second, only slightly different way in which a set of interventions $\Intset$ defines a corresponding query signature:\footnote{\rl{Which of the two is appropriate depends on the presence of a `variable alignment' in the abstraction, see Section~\ref{sec:causal-abstraction}.}} 
let $\Intsetio$ denote the signature with \emph{only two} types $\Invar, \Outvar$ and a query $\inti \colon \Invar \to \Outvar$ for each $\inti \in \Intset$, with $\semM{\Invar} = \Vinconc$, $\semM{\Outvar} = \Voutconc$. \

\subsubsection{Do-queries} 
 
The quintessential form of interventions are Do-interventions. Given concrete variables $\modelV$ 
in $\catC$, a  \emph{Do-intervention} $\Do(X=x)$ on 
\st{$\syn{X} \in \Vint$}\
is the special case of an intervention of the form:
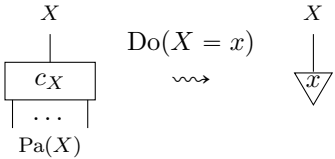
\[
\input{Do-intervention.tikz}
\]
where $x$ is a sharp state of $X$. 
More generally, a Do-intervention $\Do(S=s)$ for \st{$\syn{S} \subseteq \Vint$}\ and 
a sharp state $s$ of $S$ 
replaces the mechanism of each $\syn{X \in S}$ with the corresponding factor of $s$ on $X$.  

Formalising Do-interventions as queries is straightforward. Given concrete variables $\modelV$, the query set $\Do(\modelV)$ has types $\syn{V}$ and a query $\Do(S=s) : \Vin \rightarrow O$ for every pair of subsets \rlb{$\syn{S \subseteq \Vint}, \syn{O \subseteq \Vout}$}\ and sharp state $s$ of $S$. A causal model $\modelM$ over $\modelV$ then gives a model of $\Do(\modelV)$ by
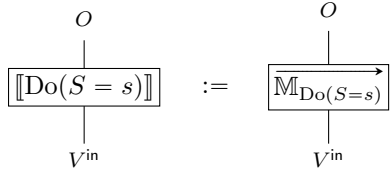
\[
\input{model-of-concrete-Do.tikz}
\]
For a CBN in $\catC=\FStoch$ the above channel is more commonly denoted $P(O \mid \Vin, \Do(S=s))$; see below for an example.\footnote{In the literature models tend not to have inputs and the common expression is then $P(O \mid \Do(S=s))$.}  In general one may consider query signatures $\Intset$ with the same types as $\Do(\modelV)$, but a strict subset of its queries, corresponding to the `sets of interventions' common in the literature \cite{BeckersEtAl_2019_AbstractingCausalModles, rubenstein2017causal}.

\paragraph{Abstract Do-queries.} Arguably, the essence of a Do-intervention is `purely syntactic' and independent from the specific value $s$ assigned to $S$. Motivated by this, given a causal model $\modelM$ and subset \st{$S \subseteq \Vint$}, we define $\opentransform{S}{\modelM}$ to be the causal model with the same variables (and outputs) given by simply deleting the mechanisms $\syn{c_X}$ for $X \in S$, making each $X \in S$ now an input.\footnote{Thus $\opentransform{S}{\modelM}$ is a causal model over the open DAG $\opentransform{S}{G}$ defined from $G$ by deleting all edges coming into $S$, and labelling all vertices in $S$ as inputs. $\Do(S)$ is a transition between causal models, but not an intervention in the above sense as it changes the set of input variables. See \cite{lorenz2023causal} for further details, therein called `opening, and \cite{richardson2023nested} for the related `fixing' operation.} 
As a graphical convention, we depict each resulting input-output diagram and channel in $\catC$ as follows.
%\footnote{\rlb{Note that our graphical convention for $\Do(S)$ just is that the wire corresponding to $S$ comes out of one of the two sides, that is, we allow it to come out of the left instead when this is more convenient for the diagram's overall typing such as in Ex.~\ref{ex:do-int}.}} 	
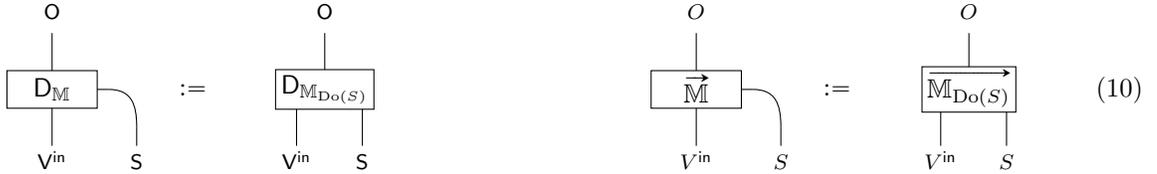
\begin{equation} \label{eq:opening-diagrams}
\input{opening-convention-D.tikz}
\qquad \qquad \qquad \qquad \qquad 
\input{opening-convention-two.tikz}
\end{equation}
For a CBN in $\catC=\FStoch$ the right-hand channel above is the probability channel commonly denoted: 
\begin{equation} \label{eq:P-Do}
P(O \mid \Vinconc, \Do(S))
\end{equation}

Formally, given input-output variables $\syn{V}$, we define the set of \emph{abstract Do-queries} $\Openqueries(\syn{V})$ with types $\syn{V}$ and a query $\syn{\openqueryshort{O}{\Vin}{S}} \colon \syn{\Vin \otimes S \to O}$ for each pair of subsets \st{$\syn{S} \subseteq \Vint, \syn{O} \subseteq \Vout$}. 
A causal model $\modelM$ over corresponding concrete variables $\modelV$ then gives a model of $\Openqueries(\syn{V})$ as abstract queries, where $\qabsMfunc$ and $\semMfunc$ send each such query $\syn{\openqueryshort{O}{\Vin}{S}}$ to the LHS and RHS of \eqref{eq:opening-diagrams}, respectively. 

Note that a Do-query $\Do(S=s)$, which we also call a \emph{concrete} Do-query, is then given by composing the abstract Do-query $\Do(S)$ with the sharp state $s$ as below.  
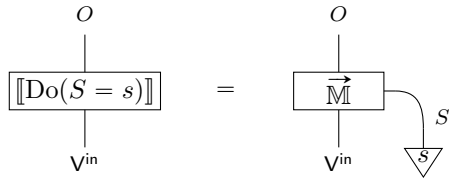
\[
\input{concrete-do_v2.tikz}
\]

\begin{example} \label{ex:do-int}
Continuing our example, the query ${\openqueryshort{O}{\Vin}{S}}$ for $\syn{S=\{D\}}$ \rlb{and}\ $\syn{O=\{E,C\}}$, has semantics: 
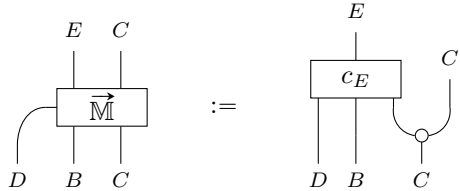
\[
%\tikzfig{open-ex-shorter}
\rlb{\input{open-ex-shorter-left-opening.tikz}}
% \tikzfig{open-ex}
% \qquad
% =
% \qquad 
% \tikzfig{open-ex-2}
\]
For any sharp state $d$ of $D$, the concrete Do-query $\Do(D=d)$ then has semantics: 
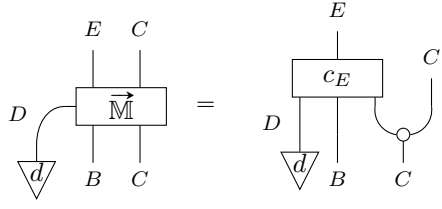
\[
% \tikzfig{do-expl}
%\tikzfig{do-expl-short}
\rlb{\input{do-expl-short-left-opening.tikz}}
\]
For a CBN these are the channels often denoted $P(E,C \mid B, C, \Do(D))$ and $P(E,C \mid B, C, \Do(D=d))$, respectively.  
\end{example}

\color{black}

\subsubsection{Interchange queries}

For causal models used in AI, such as neural networks,\footnote{Any neural network forms a causal model in a purely formal way; a separate question is whether these are `interesting' or `interpretable' causal variables.} 
it may be more natural to only consider interventions which set values to those actually arising from the model.  

Given input-output variables $\varV$, we define the query set of (abstract) \emph{interchange queries} $\WOpenqueries(\varV)$\ to have types $\varV$ and a query $\wopenquery{O}{S_1,\dots,S_n}$ from $n+1$ copies of $\syn{\Vin}$ to $\syn{O}$, for each subset 
$\syn{O \subseteq \Vout}$ and set of disjoint subsets \rlb{$\syn{S_1,\dots,S_n \subseteq \Vint}$}. 
Any causal model $\modelM$ with variables $\varV$ yields a model \rlb{of $\WOpenqueries(\varV)$}\ as abstract queries, with:  
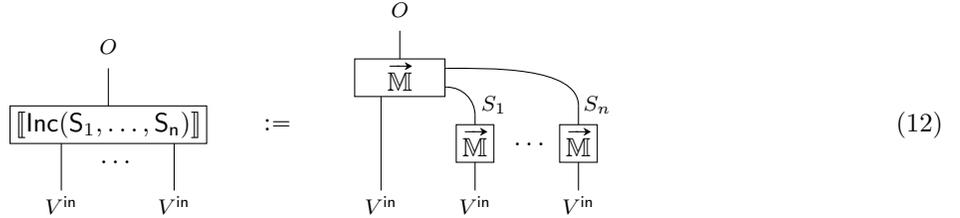
\begin{equation} \label{eq:interchange-query}
% \tikzfig{weak-open-queries}
\input{weak-open-direct.tikz}
% \tikzfig{weak-open-queries-simpler}
\end{equation}
Just as for Do-queries there are also concrete forms of interchange queries, 
\rl{which in fact are special cases of ordinary Do-interventions $\Do(S=s)$.}\ 
Given disjoint subsets \rlb{$\syn{S_1,\dots,S_n \subseteq \Vint}$}\ and sharp states $x_1,\dots,x_n$ of $\Vin$, an \emph{interchange intervention} on $\modelM$ is an intervention of the form: 
\[
\II(S_j, x_j)\rlc{^n_{j=1}} := \Do(S_j = \rl{\io{\modelM}} \circ x_j)^n_{j=1}
\]

\rl{Given concrete variable $\modelV$, we}\ write $\II(\modelV)$ 
for the \rlb{signature of queries with types $\syn{V}$ and a query for each possible interchange intervention and set of outputs $\syn{O}$, also referred to as a \emph{concrete interchange query}.
A causal model $\modelM$ represents any such query by applying}\ an interchange query to a collection of sharp states $x_1, \dots, x_n$. 
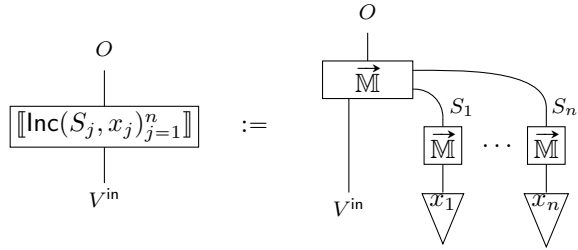
\[
% \tikzfig{II-open} 
\input{II-open-2.tikz} 
\]

\section{Abstraction} \label{sec:abstraction}

We now introduce the central notion of abstractions between models. Fix a (c)d-category $\catC$ throughout and write $\SigModel{\SigQ}{\modelM}$ for a pair consisting of a query signature $\SigQ$ along with a model $\modelM$ in $\catC$, with functor $\sem{-}_\modelM \colon \strucQ \to \catC$.\footnote{While this is all the data explicitly required in this section, in practice we usually have a compositional model $(\modelM,\SigS, \SigQ, \catC)$ of further structure $\SigS$, as for all later examples in \rlb{Sections~\ref{sec:causal-abstraction} - \ref{sec:quantum-abstraction}}.} We will in fact meet two notions of abstraction, both relating a \emph{low-level} model $\SigModel{\sigQL}{\refdefML}$ to a \emph{high-level} model ${\SigModel{\sigQH}{\refdefMH}}$, and sharing the following basic ingredient. 
Recall that, in any category, a morphism $f$ is \emph{epic} (an \emph{epimorphism}) if whenever $g \circ f = h \circ f$ we have $g = h$; \rlb{this is}\ the categorical generalisation of surjective functions between sets.

\begin{definition}[\typedowncaps] 
A \rl{\emph{\typedownfull{}}}\ 
\rlb{$\typealign{\HtoL}{\LtoH}{\SigModel{\sigQL}{\refdefML}}{\SigModel{\sigQH}{\refdefMH}}$}\
is given by specifying for each type \rl{$\syn{X}$ of $\sigQH$ a list of types  $\HtoL(\syn{X})$}\ of $\sigQL$ along with an epic \rl{deterministic}\ channel in $\catC$: 
\begin{equation} \label{eq:tau-one}
\input{tau-pic.tikz}
\end{equation}
where \rl{$X = \semMH{\syn{X}}$, $\pi(X) := \semML{\pi(\syn{X})}$.} 
\end{definition}

Henceforth in diagrams we simply write $\LtoH$ for each $\LtoH_X$, distinguishing them by their wire labels. 
Intuitively, a \typedown{} constitutes data to go between a pair of models, by assigning a query type $\HtoL(\syn{X})$ of the low-level model $\modelML$ to each query type $\syn{X}$ of the high-level model $\modelMH$, and conversely a map $\LtoH$ from low-level states of the former to high-level states of the latter. 
We extend it to any list of types by defining \rl{$\HtoL(\syn{X_1},\dots,\syn{X_n})$ as the concatenation of the $\HtoL(\syn{X_j})$}\ and $\LtoH$ as the monoidal product:\footnote{We also require that every such \rlb{product}\ morphism \eqref{eq:tau-factors} is epic; this is automatic in categories such as $\FStoch$. \rl{Also note that while Eq.~\eqref{eq:tau-factors} does not require the $\HtoL(\syn{X_j})$ to be disjoint for different $j$, in our examples they actually will always be disjoint, though without implying a \emph{\varalign} for an underlying compositional model (see Sec.~\ref{sec:causal-abstraction}) since $\sigQL$ and $\sigQH$ may have only two types $\syn{in}_L, \syn{out}_L$ and $\syn{in}_H, \syn{out}_H$, respectively, and $\HtoL(\syn{in}_H) = \syn{in}_L$ and $\HtoL(\syn{out}_H) = \syn{out}_L$.}} 
\begin{equation} \label{eq:tau-factors}
\rl{\input{tau-factorises_b.tikz}}
\end{equation}

While automatic for epic deterministic morphisms in $\FStoch$, in general we say \rlc{the}\ $\LtoH$ are \emph{surjective} when every sharp state \rlc{of}\ any $X \in \varV_H$ is of the form $\LtoH \circ s$ for some sharp state $s$ of $\HtoL(X)$.

Now, in an abstraction we also wish to relate queries themselves at each level, and there are two main ways we may do so. In each case, we relate each high-level query $\QH \colon \syn{X} \to \syn{Y}$ to at least one low-level query $\QL \colon \HtoL(\syn{X}) \to \HtoL(\syn{Y})$, such that \emph{consistency} holds in $\catC$:
\begin{equation} \label{eq:query-relation}
%\tikzfig{consistency-intext}
\input{consistency-intext_tau-convention.tikz} %Optional-nodots
% \tikzfig{consistency-intext_tau-convention-nodots}
\end{equation}
where $Q_H = \semMH{\QH}$, $Q_L=\semML{\QL}$.

\paragraph{Downward abstractions.}
 The first such relation involves a mapping $\HtoL$ on queries from \emph{high to low} level, and comes with a straightforward categorical description. 

\begin{definition}[\qdowncaps] \label{def:query-ref}
A \emph{\rlc{\qdownshort{}}
%(\qdownfull{})
} $\qdownabs{\HtoL}{\LtoH}{\SigModel{\sigQL}{\refdefML}}{\SigModel{\sigQH}{\refdefMH}}$ 
is given by a functor $\HtoL$ and epic natural transformation $\LtoH$ of the form: 
\[
\input{qref.tikz}
\]
such that for each query $\syn{Q}$ of $\SigQH$, $\HtoL(\syn{Q})$ is a query of $\sigQL$.\footnote{\rl{that is, of the signature $\sigQL$, rather than \rlb{just}\ of the category $\strucQL$}} 
\end{definition} 

Unpacking the definition, we see that a \qdown{} is precisely a \typedown{} (given by the map $\HtoL$ on types and transformation $\LtoH$) along with a further map $\HtoL$ sending each high-level query $\QH \colon \syn{X} \to \syn{Y}$ to a low-level query $\QL := \HtoL(\QH) \colon \HtoL(\syn{X}) \to \HtoL(\syn{Y})$. Naturality of $\LtoH$ is none other than the requirement that consistency \eqref{eq:query-relation} holds for each query $\QH$ in $\sigQH$, where $\QL = \HtoL(\QH)$. That is, in $\catC$ we have: 
\[
% \tikzfig{query-ref-sem}
% \tikzfig{query-refinement-equation}
% \qquad \qquad 
\input{query-refinement-equation-alt.tikz}
% \tikzfig{query-refinement-equation-alt-nodots} %Optional-nodots
\]

\paragraph{Upward abstractions.} 
For many `concrete' queries there are however typically many low-level queries associated with a single high-level one (e.g. as for specific Do-interventions and the notion of exact transformations of causal models \rl{\cite{rubenstein2017causal, BeckersEtAl_2019_AbstractingCausalModles}; see Sec.~\ref{sec:causal-abstraction})}. To capture this we consider our second notion of abstraction, instead using a partial mapping $\LtoHquery$ of queries from the low to high level. Intuitively, $\LtoH$ then forms a natural transformation as below. 
\begin{equation} \label{eq:exact-partial}
\input{exact-nat-trans-2.tikz}
\end{equation}

More precisely, we define the following.

\begin{definition} [\qupcaps] \label{def:up-abstraction}
\rlc{An}\ \emph{\rlc{\qupshort{}} 
%(\qupfull)
} 
$\qupabs{\HtoL}{\LtoH}{\queryLtoH}{\SigModel{\sigQL}{\modelML}}{\SigModel{\sigQH}{\modelMH}}$
is given by a \typedown{} $(\HtoL,\LtoH)$ along with a partial surjective mapping $\queryLtoH \colon \SigQL \pto \SigQH$ on queries of the form: 
\begin{equation} \label{eq:ex-trans-map}
\input{omega-queries.tikz}
% \tikzfig{omega-queries-nodots} %Optional-nodots
\end{equation}
such that consistency \eqref{eq:query-relation} holds with $\QH = \queryLtoH(\QL)$.  
\end{definition}

Thus, in \rlc{an}\ \qup, whenever $\queryLtoH(\syn{Q}_L)$ is defined and has type $\syn{X} \to \syn{Y}$ then $\syn{Q}_L$ has type $\HtoL(\syn{X}) \to \HtoL(\syn{Y})$, and the following holds in $\catC$:  
\begin{equation} \label{eq:consistency-condition}
\input{Q-tau-consistency.tikz}
% \tikzfig{Q-tau-consistency-nodots} %Optional-nodots
\end{equation}

Again this tells us that \qup s are some kind of natural transformation as in the picture \eqref{eq:exact-partial}. We make this precise in Proposition \ref{prop:ex-trans} in Appendix \ref{app:up-abs}, by relating \qup s to \qdown s. 
 
A useful fact is that both notions of abstraction are compositional in the following sense.

\begin{proposition}  
Downward and upward abstractions are each closed under composition. That is, if:
\[
\input{abstr-closed.tikz}
\] 
are \qdown{}s (resp. \qup{}s), then so is: 
\[
\input{abstr-closed-2.tikz}
\]
\end{proposition}

As a result, type alignments are closed under composition also. 

\begin{proof}
	Straightforward to verify. 
\end{proof}

\paragraph{From downward to upward abstractions.} As we will see repeatedly in upcoming examples, \rlc{an}\ \qup{} often arises simply by restricting \rlb{the consistency conditions of}\ a \qdown{} to certain input states. 
\rl{In the next section we will meet several examples of such pairs of causal \qdown{} and \qup{} that are related through corresponding abstract and concrete versions of the same kind of queries, summarised in Figure~\ref{fig:overview-table}.}

These can be seen to arise from the following general result. Consider a model of queries $\SigModel{\SigQ}{\modelM}$ where each query $\syn{Q}$ comes with a chosen (set of) input(s) $\syn{Z_Q}$. 
%Suppose we also have a  choice $\Stchoice$ of a set  $\Stchoice(\syn{Q})$ of states of $\syn{Z_Q}$ in $\catC$, for each $\syn{Q}$. 
We define a \rlb{signature of queries}\ $\SigQ^{\Sh}$ with the same types as $\SigQ$ and a query $(\syn{Q},s)$ for each $\syn{Q} \in \SigQ$ and sharp state $s$ of $\syn{Z_Q}$, with 
% \rl{$s \in \Stchoice(Q)$}, with 
$\semM{(\syn{Q},s)}$ given in $\catC$ by:
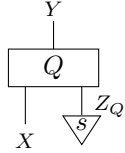
\[
\input{Qssmall.tikz}
\]

\begin{proposition} \label{prop:Exact-trans-from-ref}
Let be
$\qdownabs{\HtoL}{\LtoH}{\SigModel{\SigQL}{\modelML}}{\SigModel{\SigQH}{\modelMH}}$ a \qdown{}, where $\sigQH$ 
is of the above form, with surjective $\LtoH$. We obtain \rlc{an}\ \qup{} $\qupabsshort{\SigModel{\SigQL^\Sh}{\modelML}}{\SigModel{\SigQH^\Sh}{\modelMH}}$ via: 
\[
\left(\input{QS1a.tikz} \ , \ \begin{tikzpicture}
	\begin{pgfonlayer}{nodelayer}
		\node [style=none] (0) at (0.5, 1.25) {};
		\node [style=none] (1) at (0.5, 0.25) {};
		\node [style=point] (2) at (0.5, 0) {$s$};
		\node [style=label] (3) at (0.5, 1.75) {\rl{$\HtoL(Z_Q)$}};
	\end{pgfonlayer}
	\begin{pgfonlayer}{edgelayer}
		\draw (0.center) to (1.center);
	\end{pgfonlayer}
\end{tikzpicture}
 \right)
\qquad
\input{omega.tikz}
\qquad
\left(\input{QS2a.tikz} \ , \ \begin{tikzpicture}
	\begin{pgfonlayer}{nodelayer}
		\node [style=none] (0) at (14.25, 1.25) {};
		\node [style=label] (1) at (14.25, 1.75) {\rl{$Z_Q$}};
		\node [style=point] (2) at (14.25, -1.5) {$s$};
		\node [style=map] (3) at (14.25, 0) {$\LtoH$};
	\end{pgfonlayer}
	\begin{pgfonlayer}{edgelayer}
		\draw (2) to (0.center);
	\end{pgfonlayer}
\end{tikzpicture}
 \right)
\]
with $\syn{Z}_{\HtoL(\syn{Q})} := \HtoL(\syn{Z_Q})$.
% and 
%$\LtoH(\Stchoice)(\HtoL(\syn{Q})) := \{ \LtoH \circ s \mid s \in \Stchoice({\HtoL(Z_Q)})\}$}, 
%for each $\syn{Q} \in \SigQH$.  
\end{proposition}

\begin{proof}
Surjectivity holds by construction and \rlb{consistency directly from naturality of the \qdown{}, applied to each state $s$: 
\begin{align} 
%	\tikzfig{consistency-up-from-down}
	\input{consistency-up-from-down-boxed-up.tikz} 
%	\tikzfig{consistency-up-down-small}
	\label{eq:proof-prop-up-from-down}
\end{align}
}
\end{proof}

\begin{figure}[H]
\centering
\rl{
\begin{tabular}{c|cc|c}
Level     & Abstract      & Concrete  
& \begin{minipage}{4.3cm} \centering abstraction type \end{minipage} \\[0.2cm]  \hline 
$\Ltwo$   & Do-query      
& \begin{minipage}{4cm} \centering \vspace*{0.3cm} Do-intervention  \\ {\footnotesize (concrete Do-query) } \end{minipage}   
& \begin{minipage}{4.3cm} \centering constructive abstraction  \end{minipage} \\[0.7cm] 
%----
$\Ltwo$   & interchange query  & 
\begin{minipage}{5cm} \centering interchange intervention \\ {\footnotesize (concrete interchange query ) } \end{minipage} 
& \begin{minipage}{4.3cm} \centering \restrictedCCA  \end{minipage} \\[0.5cm]  
%----
$\Lthree$ & \CF{} query 
& \begin{minipage}{5.5cm} \centering concrete \CF{} query  \end{minipage}    
& \begin{minipage}{4.3cm} \centering \CF{} abstraction \end{minipage}  
\end{tabular}
}
\caption{\rl{Overview of several forms of causal abstraction from Section~\ref{sec:causal-abstraction}; $\Ltwo$ and $\Lthree$ denote the second (interventional) and third (counterfactual) level of the causal hierarchy (see, e.g., \cite{pearl2009causality, bareinboim2022pearl}).}\label{fig:overview-table}} 
\end{figure}

\section{Causal abstraction} \label{sec:causal-abstraction}

In this section we will now see how many notions of causal abstraction from the literature form special cases of our definitions, i.e.~natural transformations between \rlb{models of}\ queries. Throughout, we consider a pair of \emph{causal} models $\modelML$ and $\modelMH$ in $\catC$ with respective variables $\syn{V^{\text{in}}_{L/H}},\syn{V^{\text{out}}_{L/H}} \subseteq \syn{V_{L/H}}$.

In each case, we will first define the abstraction directly in a form closest to the literature, before then seeing how it falls as a special case of our definitions for a particular choice of causal queries from Section \ref{subsec:causal-models}. In several cases, the following form of 
\rlb{a mapping}\ will be present.

\begin{definition}[\rl{\varaligncaps}]
\cite[Def. 3.19]{BeckersEtAl_2019_AbstractingCausalModles} 
\cite[Def 31]{geiger2023causal}  
A \emph{\varalign} $\dvaralign{\HtoL}{\LtoH}{\modelVL}{\modelVH}$ 
between concrete variables in $\catC$ is given by:
\begin{itemize}
\item 
\rlc{A collection of disjoint subsets $(\HtoL(\syn{X}))_{\syn{\rlc{X \in \VH}}}$ of $\VL$, 
satisfying $\HtoL(\VinH) = \VinL$ and $\HtoL(\VoutH) \subseteq \VoutL$;}
\item 
For each $\rl{\syn{X}} \in \varV_H$ a deterministic epic channel in $\catC$:  
\[
%\tikzfig{tau-V}
\input{tau-X.tikz}
\] 
\end{itemize}
\end{definition}

Thus \rl{\varalign s}~are similar to \typedown s, with the variables as types and now the property that the \rl{$\HtoL(\syn{X})$}\ form disjoint subsets. Again in $\catC=\FStoch$ the conditions mean that each map is a surjective function \rl{$\LtoH \colon \HtoL(X) \to X$}. 

\begin{remark} \label{rem:var-alignment-ref} 
For any concrete variables $\modelV$, write $\varV^*$ for the set of abstract queries with types $\varV$ and a query $\synQ_\syn{S} \colon \syn{S} \to \syn{S}$ for each subset $\syn{S} \subseteq \varV$, with $\qabs{\synQ_\syn{S}}_\modelV = \id{\syn{S}}$. 
Then a \rl{\varalign\ $\dvaralign{\HtoL}{\LtoH}{\modelVL}{\modelVH}$ }\ is in fact a \qdown{} 
$\qdownabsshort{\SigModel{\varV^*_L}{\modelVL}}{\SigModel{\varV^*_H}{\modelVH}}$ satisfying $\HtoL(\VinH) = \VinL$ and $\HtoL(\VoutH) \subseteq \VoutL$; for a proof see Appendix \ref{app:proofs}. 
\end{remark}

\subsection{Exact transformations} \label{subsec:exact-intv}

Perhaps the \rl{first and weakest notion}\ 
of abstraction in the literature, \rl{relating concrete interventions}\ from low-level to high-level, 
is referred to as an `exact transformation' of models. We adapt 
\cite[Def 3]{rubenstein2017causal},
\cite[Def 3.1]{BeckersEtAl_2019_AbstractingCausalModles} and
\cite[Def 25]{geiger2023causal}
to our setting as follows \rl{(see Rem.~\ref{rem:ET-notions} later for a precise comparison).}

\begin{definition}[Exact Transformation] \label{def:exact-trans}
Given \rl{causal}\ models $\modelML, \modelMH$ with sets of interventions $\IntsetL$, $\IntsetH$, an \emph{exact transformation} $\modelML \to \modelMH$ is given by a pair of \rl{epic}\ deterministic channels $\LtoH \colon \VinLconc \to \VinHconc$ and $\LtoH \colon \VoutLconc \to \VoutHconc$ in $\catC$ (which we distinguish by their inputs) and a surjective partial function $\LtoHquery \colon \IntsetL \to \IntsetH$ such that whenever $\LtoHquery$ is defined on  $\inti \in \IntsetL$ we have:  
\begin{equation} \label{eq:int-exact-trans} 
% \tikzfig{int-exact-trans-2}
\input{int-exact-trans-2-nodots.tikz}
\end{equation}
\end{definition}

Recall that given set of interventions $\IntsetL$ the query signature $\IntsetioL$ contains a query $\inti : \syn{in_L} \rightarrow \syn{out_L}$ for every intervention $\inti \in \IntsetL$, and similarly for $\IntsetioH$. 

%\footnote{\rl{with $\sem{\syn{in_L}}_{\modelML} =\VinLconc$ and $\sem{\syn{out_L}}_{{\modelML}} =\VoutLconc$}}

\begin{proposition} \label{prop:ET-equiv-qup} 
An exact transformation $(\tau, \omega)$ from $\modelML$ to $\modelMH$ with intervention sets $\IntsetL$, $\IntsetH$ is equivalent to \rlc{an}\ 
\qup{} $\qupabs{\HtoL}{\LtoH}{\queryLtoH}{\SigModel{\IntsetioL}{\modelML}}{\SigModel{\IntsetioH}{\modelMH}}$.
\end{proposition}
\begin{proof}
Immediate from the definitions, noting that for any such \qup{} to have $\LtoHquery$  map queries to queries, the map on types $\HtoL$ must be given by $\HtoL(\syn{in_H}) := \syn{in_L}$ and $\HtoL(\syn{out_H}) := \syn{out_L}$.
\end{proof}

Note this notion does not assume a \varalign. A stronger version, which one can think of as a vanilla, general notion of causal abstraction, just short of demanding a \varalign, is the following. 
For any causal $\modelM$, by the \emph{trivial intervention} 
we mean the intervention $\noint$ which changes no mechanisms, so that $\modelMio_\noint = \modelMio$.

\begin{definition}[\strongCAcaps]  \label{def:strong-CA}
Given causal models $\modelML, \modelMH$ with sets of interventions $\IntsetL$, $\IntsetH$, an exact transformation $(\LtoH, \LtoHquery)$ from $\modelML$ to $\modelMH$ is called a \emph{\strongCA} if 
\begin{enumerate}
	\item $\VninH \subseteq \VoutH$ 
	i.e.~all non-input variables are outputs;
	\item $\IntsetH=\Do(\modelV_H)$, i.e.~contains \emph{all} concrete do-interventions \rlb{at the high level}; % (on all of $\VoutH$);
	\item 
	$\IntsetL$ and $\IntsetH$ contain the trivial interventions, with $\LtoHquery(\nointL)=\nointH$. 
\end{enumerate}
\end{definition}

Examples of exact transformations and \strongCA s will follow in the subsequent sections, with comments on the motivations and relation to the literature postponed to Rem.~\ref{rmk:strong-CA}.

\subsection{Constructive abstraction} \label{subsec:CCA}

We now reach a central, \rl{perhaps \emph{the} paradigmatic}, notion of causal abstraction, \rl{which forms the strongest notion}\ of abstraction given in \cite{BeckersEtAl_2019_AbstractingCausalModles}, and \rl{is}\ of particular use for AI models and other scenarios relating subsets of micro to macro variables.

\begin{definition}[Constructive Abstraction] \label{def:CCA}
\cite[Def 3.19]{BeckersEtAl_2019_AbstractingCausalModles} 
Let $\modelML, \modelMH$ be causal models in $\catC$, with respective concrete variables $\modelVL, \modelVH$ \rlb{and $\VninH \subseteq \VoutH$}. A  
\emph{constructive abstraction} $\modelML \to \modelMH$ is given by a \rl{\varalign\ $\dvaralign{\HtoL}{\LtoH}{\modelVL}{\modelVH}$} such that for all $\syn{S} \subseteq \VintH$  
the following holds in $\catC$: 
\begin{equation} \label{eq:CAabstraction}
\input{causal-abstr-open-simple.tikz}
\end{equation}
\end{definition} 

A constructive abstraction thus structurally relates (abstract) Do-interventions on $\modelH$ to those on $\modelL$. We can now characterise this categorically as follows.

\begin{theorem} \label{Thm:CCA-as-refinement}
A constructive abstraction $\modelML \to \modelMH$ 
\rlc{with \varalign\ $(\HtoL,\LtoH)$}\ is \rl{equivalent to}\ 
a \qdown{} $\qdownabs{\HtoL}{\LtoH}{\SigModel{\Openqueries(\varL)}{\modelML}}{\SigModel{\Openqueries(\varH)}{\modelMH}}$ 
such that $\HtoL(\VinH) = \VinL$. 
\end{theorem}
\begin{proof}
Appendix \ref{app:proofs}.
\end{proof}

Here we have exposed constructive abstraction as an entirely structural notion, in terms of abstract Do-queries, while the usual presentation is in terms of Do-interventions with concrete values. The link with the latter is however  straightforward via composition with input states. 
The following tells us that we can see constructive abstraction also as \rlc{an}\ \qup\ between Do-interventions 
which must surjectively map low-level Do-interventions $\Do(S_L=s_L)$ to high-level ones $\Do(S_H=s_H)$.

\begin{corollary} \label{cor:CCA-as-qup}
Any constructive abstraction 
%$(\HtoL,\LtoH) \colon 
$\modelML \to \modelMH$
%coming 
with surjective $\LtoH$ 
defines \rlc{an}\ \qup{}
of the form 
$\qupabs{\HtoL}{\LtoH}{\LtoHquery}{\SigModel{\rl{\Do(\modelV_L)}}{\modelML}}{\SigModel{\rl{\Do(\modelV_H)}}{\modelMH}}$ 
via: 
\[
\LtoHquery(\Do(\HtoL(S)=s))) := \Do(S=\LtoH \circ s) 
\]
\end{corollary}

As a result it also defines an exact transformation of Do-interventions, i.e.~\rlc{an}\ \qup{} of the form 
${\qupabs{\HtoL}{\LtoH}{\LtoHquery}
{\SigModel{\Do(\modelV_L)^{\mathsf{io}}}{\modelML}}
{\SigModel{\Do(\modelV_H)^{\mathsf{io}}}{\modelMH}}}$ with each $\Do(\modelV)^{\mathsf{io}}$ understood as a signature $\Intsetio$ with only two types. 
\rlc{In fact it defines a \strongCA\ (Def.~\ref{def:strong-CA}).}\

\begin{proof}
The consistency condition for such \rlc{an}\ \qup{} is precisely that: 
\begin{equation} \label{eq:constr-CA-explicit}
\input{constr-CA-explicit-2.tikz}
% \eqref{eq:constr-CA-explicit}
\end{equation}
Alternatively, this forms a special case of Proposition \ref{prop:Exact-trans-from-ref} with $\SigQH=\Openqueries(\VH)$ and $\SigQL=\Openqueries(\VL)$, where for $\syn{Q}=\Do{(\syn{S})}$ we take \rlb{$\syn{Z_Q} = \syn{S}$}. 
% along with the full set $\Sh(S)$ the set of sharp states $s$ of $S$; 
\rlb{We can then recognise Eq.~\eqref{eq:constr-CA-explicit} as an instance of Eq.~\eqref{eq:proof-prop-up-from-down}}. 
The statement where the query types are only $\syn{in,out}$ is immediate. 
\end{proof}

\begin{example} \label{ex:CCA} 
	The following simple example of a constructive abstraction is from \cite{BeckersEtAl_2019_AbstractingCausalModles}.\footnote{\rl{Originally from \cite{rubenstein2017causal} though discussed there with restricted sets of interventions exemplifying exact transformations.}} Consider open deterministic causal models $\modelL$ and $\modelH$ in $\FStoch$ shown left and right below, respectively. 
	\begin{equation} \label{eq:CCA-example-models}
		\input{CCA-example-L_b.tikz} 
		\hspace*{3cm}
%		\hspace*{1.0cm} , \hspace*{2.0cm}
		\input{CCA-example-H.tikz}
	\end{equation}
	Here $X_1,..., X_{99}$ are binary variables, valued in $\{0,1\}$, representing the votes on some petition by 99 individuals (`1' corresponds to `yes'), $T$ is the total number in support of the petition (i.e. $f_T = \sum_i X_i$) and $A_1$, $A_2$ are some discrete variables representing two possible advertisements the campaign for the petition can choose to run; $\{f_{X_i}\}_{i=1}^{99}$, $f_{A_1}$ and $f_{A_2}$ are some functions of the corresponding types with $U_{1},..., U_{101}$ the background variables considered (exogenous) inputs. 
	The high-level model arises from dividing the voters into three distinct groups 
	$G_1,G_2$ and $G_3$. 
	$G_1$ is valued in $[1,33]$ with $U_1' = \prod_{i=1}^{33} U_i$ and $f_{G_1}(a_1,a_2,u_1') = \sum_{i=1}^{33} f_{X_i}(a_1,a_2,(u_1')_i)$, analogously for $G_2, U_2',f_{G_2}$ and $G_3, U_3', f_{G_3}$; finally, $T'$ is binary, representing whether the petition is accepted or not, with $f_{T'}(g_1,g_2,g_3)$ yielding one if $g_1+g_2+g_3 > 49$ and zero otherwise. 
	
The relation between the two models can be recast in terms of a \varalign, with 
	$\HtoL(G_1)=\{X_1,...,X_{33}\}$, $\HtoL(U_1')=\{U_1,...,U_{33}\}$, analogously for $G_2,U_2'$ and $G_3,U_3'$ and with $\HtoL(T')=\{T\}$, $\HtoL(A_1)=\{A_1 \}$, $\HtoL(U_{100})=\{U_{100} \}$, $\HtoL(A_2)=\{A_2\}$, $\HtoL(U_{101})=\{U_{101}\}$; 
	the map $\LtoH$ is given by setting $\LtoH_{T'}(t) = 1$  if $t>49$ and $0$ otherwise, $\LtoH_{G_1} = \sum_{i=1}^{33} X_i$, analogously for $G_2$ and $G_3$, and setting $\LtoH$ to be the identity on the (exogenous) input variables and on $A_1,A_2$. 
	It then holds that: 
	\begin{equation} \label{eq:CCA-example-mechanism-consistency}
		\input{CCA-example-fTprime-consistency.tikz} 
		\hspace*{2.0cm}
%		\hspace*{0.5cm} , \hspace*{1.5cm}
		\input{CCA-example-fG-consistency.tikz}
		\hspace*{-0.5cm}
	\end{equation}
	and analogously for $f_{G_2}$ and $f_{G_3}$. 
	
	 One may check that this data indeed defines a constructive abstraction $\modelL \rightarrow \modelH$. 
	 As an illustrative example we verify consistency for an abstract Do-intervention on $S=\{G_1\}$; other ones go through analogously. 
	 The left-to-right order matches Eq.~\eqref{eq:CAabstraction}, but the computation is actually most easily done right to left: (a) and (d) hold by definition; (b) uses equalities like the RHS of Eq.~\eqref{eq:CCA-example-mechanism-consistency} for $f_{G_2}$ and $f_{G_3}$; (c) passes three instances of the deterministic morphisms $\LtoH$ through the copy maps 
	  and then rewrites using the LHS of Eq.~\eqref{eq:CCA-example-mechanism-consistency}. 
	\begin{center}
		\begin{minipage}{0.2\textwidth}
			\centering
			\input{CCA-example-consistency-example-1.tikz} 
		\end{minipage}
		\hspace*{0.5cm} $\stackrel{(d)}{=}$ \hspace*{0.5cm}
		\begin{minipage}{0.4\textwidth}
			\centering
			\input{CCA-example-consistency-example-2.tikz} 
		\end{minipage} 
		\hspace*{1.0cm}
		$\stackrel{(c)}{=}$ \\
		\begin{minipage}{0.41\textwidth}
			\centering
			\hspace*{-0.6cm}
			\input{CCA-example-consistency-example-3.tikz} 
		\end{minipage} 
		$\stackrel{(b)}{=}$
		\begin{minipage}{0.3\textwidth}
			\centering
			\vspace*{0.5cm}
			\hspace*{-0.6cm}
			\input{CCA-example-consistency-example-4.tikz} 
		\end{minipage}
		$\stackrel{(a)}{=}$
		\begin{minipage}{0.21\textwidth}
			\centering
%			\hspace*{0.2cm}
			\input{CCA-example-consistency-example-5.tikz} 
		\end{minipage}  
	\end{center} 
	\color{black}
\end{example}

\subsection{\restrictedCCAcaps} \label{subsec:weak-CA}

In practice, one \rl{may not be able}\ to verify a constructive \rl{causal abstraction}\ condition on arbitrary inputs, but only (a subset of) those which arise from the low-level model itself. This leads to the following weaker notion based instead on interchange queries \rl{\cite{GeigerEtAl_2021_CausalAbstracttoinOfNN, geiger-etal-2020-neural}.}  

\begin{definition}[\rl{\restrictedCCAcaps}]
\rl{An \emph{\restrictedCCA}}\ $\modelML \to \modelMH$ is given by a \rl{\varalign\ $\dvaralign{\HtoL}{\LtoH}{\modelVL}{\modelVH}$}, such that for all disjoint subsets $\syn{S_1,\dots,S_n \subseteq \VintH}$ the following holds in $\catC$:
\begin{equation} \label{eq:weak-CA-natural}
\input{weak-CA-natural.tikz}
\end{equation}
\end{definition}

The above corresponds to a consistency condition between interchange queries $\WOpenqueries(\varV)$, which we can formalise as follows.

\begin{proposition} \label{prop:interchange-queries}
\rl{An \restrictedCCA}\ 
$\modelML \to \modelMH$ 
\rlc{with \varalign\ $(\HtoL,\LtoH)$}\ 
is \rl{equivalent to}\  
a \qdown{}  
%of the form 
$\qdownabs{\HtoL}{\LtoH}{\SigModel{\WOpenqueries(\varL)}{\modelML}}{\SigModel{\WOpenqueries(\varH)}{\modelMH}}$ 
such that $\HtoL(\VinH) = \VinL$, 
$\HtoL(\VintH) \subseteq \VintL$ and: 
\begin{equation} \label{eq:weakCA-query-map}
\syn{
%\HtoL(\wopenquerylong{X}{\VinH}{S_1,\dots,S_n}) \ \ := \ \  (
%\syn{\wopenquerylong{\HtoL(X)}{\VinL}{\HtoL(S_1),\dots,\HtoL(S_n)}})
\rlb{\HtoL \big( \wopenquery{}{S_1,\dots,S_n} : \VinH \to X \big) \ = \ \wopenquery{}{\HtoL(S_1),\dots,\HtoL(S_n)} : \VinL \to \HtoL(X)}
}
\end{equation}
\end{proposition}

\begin{proof}
Appendix \ref{app:proofs}.
\end{proof}

Just as constructive abstraction relates to concrete Do-interventions, \rl{\restrictedCCA}\ relates to concrete interchange interventions. \rlc{An}\ \qup{} \rl{$\qupabsshort{\SigModel{\II(\modelVL)}{\modelML}}{\SigModel{\II(\modelVH)}{\modelMH}}$}\ must surjectively map low-level interchange interventions to high-level ones in a consistent manner, and \rl{\restrictedCCA}\ provides a structured way to do so.

\begin{corollary}
Any \rl{\restrictedCCA}\ 
$\modelML \to \modelMH$ 
%coming 
with surjective $\LtoH$ 
defines \rlc{an}\ \qup{} of the form 
%of the form 
\rl{$\qupabs{\HtoL}{\LtoH}{\LtoHquery}{\SigModel{\II(\modelVL)}{\modelML}}{\SigModel{\II(\modelVH)}{\modelMH}}$}\ via: 
\[
\LtoHquery ( \II(\HtoL(S_j), x_j)^n_{j=1}) :=  \rlc{\II(S_j,\LtoH \circ x_j)^n_{j=1}}
\]
\end{corollary}

As for constructive abstraction, such a u-abstraction also gives an exact transformation of concrete (interchange) interventions $\qupabs{\HtoL}{\LtoH}{\LtoHquery}
{\SigModel{\II(\modelVL)^{\mathsf{io}}}{\modelML}}
{\SigModel{\II(\modelVH)^{\mathsf{io}}}{\modelMH}}$.

\begin{proof}
The consistency condition on queries for such a u-abstraction simply becomes the statement: 
\begin{equation} \label{eq:interchange-ints} 
\input{II-concretely.tikz}
% \hspace*{-0.7cm}
% \tikzfig{II-concrete-moredetail_b}
\end{equation}
which follows immediately from \eqref{eq:weak-CA-natural}. Alternatively, apply Proposition \ref{prop:Exact-trans-from-ref} with 
%with $\Stchoice=\Sh$ the full set of sharp states, 
$\SigQH=\WOpenqueries(\VH)$, $\SigQL=\WOpenqueries(\VL)$. For $\syn{Q} = \wopenquery{X}{\syn{S_1,\dots,S_n}}$ we take \rlb{$Z_Q = S_1 \otimes \dots \otimes S_n$}. 
\end{proof}

\begin{remark} \label{remark:alt-view-weak-abstraction}
It follows from \rl{an \restrictedCCA}\ relation that we can re-write the \rl{RHS of}\ \eqref{eq:weak-CA-natural} as:\footnote{\rl{The re-write uses the special case of \eqref{eq:weak-CA-natural} for $\syn{S_1}=\dots = \syn{S_n} = \emptyset$.}} 
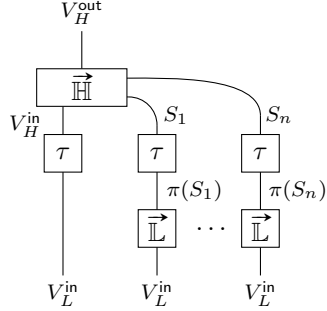
\[
\input{weak-CA-rewrite-LHS.tikz}
\]
Hence \rl{an \restrictedCCA}\ is precisely the statement that a constructive abstraction holds but restricted only to input states of $\HtoL(S)$ of the following form, with $\syn{S = S_1 \cup \dots \cup S_n}$ disjoint. 
\begin{equation} \label{eq:weak-abs-inputs}
\input{input-states.tikz}
\end{equation} 
\end{remark}

\subsection{Counterfactual abstraction} \label{subsec:CF-abstraction}

%\paragraph{Counterfactual queries} 
%\label{subsec:queries-for-FCMs}

\rl{Let us first name the kind of causal model with respect to which counterfactuals are defined.} 
\begin{definition} \label{def:FCM}
A \emph{Functional Causal Model (FCM)} in a Markov category $\catC$ is a causal model $\modelM$ with variables partitioned as $\syn{V = \FMCVen \cup U}$ with $\Vout = \syn{\FMCVen}$, $\Vin=\emptyset$. The \emph{exogenous} variables, $\syn{U}=\{\syn{U_i}\}^n_{i=1}$, are such that each $\syn{U_i}$ has no parents, with mechanisms as left below, and the \emph{endogenous} variables  $\syn{\FMCVen} = \{\syn{X_i}\}^n_{i=1}$ of the right-hand form below:
% being states $\lambda_i$. 
\begin{equation} \label{eq:noise-var}
\begin{tikzpicture}
	\begin{pgfonlayer}{nodelayer}
		\node [style=none] (0) at (0.5, 1.25) {};
		\node [style=map] (1) at (0.5, 0) {$\syn{\lambda_i}$};
		\node [style=label] (2) at (0.5, 1.75) {$\syn{U_i}$};
	\end{pgfonlayer}
	\begin{pgfonlayer}{edgelayer}
		\draw (0.center) to (1);
	\end{pgfonlayer}
\end{tikzpicture}

% \end{equation}
\qquad \qquad 
% % the \emph{endogenous} variables, where each $\syn{X_i}$ has a mechanism of the form:  
% \begin{equation} \label{eq:SCM-2}
\input{SCM-2.tikz}
\end{equation}
where $\Pa'(\syn{X_i}) \subseteq \syn{\FMCVen}$ and $f_i = \semSM{\syn{f_i}}$ is deterministic \eqref{eq:deterministic}.\afootnote{Here we follow the typical convention that an FCM has no inputs, though one could easily consider variants with \rlb{extra inputs}. Formally, the signature for an FCM includes equations specifying that each $\syn{f_i}$ is deterministic.} 
\end{definition}
	
% \rl{Recall from \eqref{eq:deterministic} what it means for a process in a cd-category to be deterministic.}\ 
%We use the (standard) name \emph{Structural Causal Model} (SCM) to refer to an FCM in $\catC = \FStoch$. An SCM is thus given by finite sets $X_i,U_i$ and for each $i$ a distribution $\lambda_i$ over $U_i$ and function $f_i \colon \Pa'(X_i) \times U_i \to X_i$ where \rlb{$\Pa'(X_i) \subseteq \FMCVen = \{X_1,\dots, X_n\}$}. 

\begin{example}
An FCM in $\catC = \FStoch$ is precisely a \emph{Structural Causal Model (SCM)} \rlc{in}\ the usual sense. It consists of finite sets $X_i,U_i$ and for each $i$ a distribution $\lambda_i$ over $U_i$ and function $f_i \colon \Pa'(X_i) \times U_i \to X_i$ where \rlb{$\Pa'(X_i) \subseteq \FMCVen = \{X_1,\dots, X_n\}$}. 
\end{example}

\begin{example} 
The following shows a network diagram for an FCM \rlc{with}\ \rlb{$\syn{\FMCVen=\{S,L,A\}}$}\ and \allowbreak\rlb{$\syn{U=\{U_S,U_L,U_A\}}$}. 
%with outputs $\syn{S,L,A}$. 
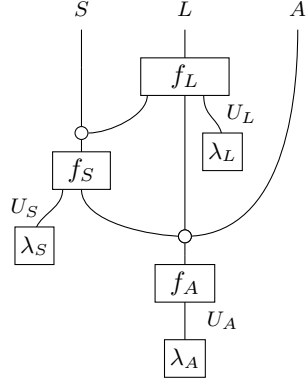
\[
\input{FCM-ex-smaller.tikz}
\]
\end{example} 

Closely related is the following. By a \emph{deterministic} causal model $\modelF$ we mean a causal model whose mechanisms $c_X$ are deterministic for all \rlb{$X \in \Vnin$}.\footnote{\rl{Note that being a deterministic causal model doesn't come with a condition on a particular partitioning of the variables, hence a root node variable either is an input or else has a sharp state (a point distribution) as its `mechanism'.}} 
Formally, we can view any FCM as a composite\footnote{\rlb{See \cite[Sec 5.2]{lorenz2023causal} for the details on the notion of composition between open causal models, which takes the union of the respective sets of mechanisms together with identifying respective input and output variables over which the composition happens.}} $\modelM = \modelF \circ \modelU$ of two open causal models: 
\rl{$\modelF$,}\ a deterministic model on \rlb{$\FMCVen, U$}\ with $U$ as inputs, containing all the deterministic mechanisms $f_i$, 
and 
\rl{$\modelU$, a}\ model containing only $U$ and their mechanisms, corresponding to the joint state:  
\[
\input{lambda-state.tikz}
\]

\paragraph{Counterfactual queries.}
Let us now consider \rl{queries definable}\ for FCMs specifically, following \cite{xia2024neural}. 
%Let \rlb{$\syn{\CFVen}$}\ be a set of 
%\rlb{abstract variables}
Given abstract variables $\syn{\CFVen}$, 
% and, borrowing notation from \cite{xia2024neural} which studies a similar definition, we}\ 
define the \rlb{signature $\absLthree(\syn{\CFVen})$}\ of \emph{\CF{} queries} to have types \rlb{$\syn{\CFVen}$}\ and a query \rlb{$\syn{\Lthreequery{Y_1|_{\CFOpenSet_1},\dots,Y_m|_{\CFOpenSet_m}}}$ with inputs $(\CFOpenSet_j)_{j=1}^m$ and outputs $(Y_j)_{j=1}^m$}, for each choice of subsets \rlb{$\CFOpenSet_j, Y_j \subseteq \syn{\CFVen}$}\ for $j=1,\dots,m$.
% modelled by an FCM $\modelM$ with endogenous variables $\CFVen$ as:
Any FCM $\modelM = \modelF \circ \modelU$ over $\syn{\CFVen}, \syn{U}$ yields a model of the queries via:
% of $\absLthree(\syn{\CFVen})$, as abstract queries, via} 
% Define the query \rlb{signature $\absLthree(\syn{\CFVen})$}
% %\ of \emph{\CF{} queries} 
% to have types \rlb{$\syn{\CFVen}$}\ and a query $\syn{\Lthreequery{Y_1|_{\CFOpenSet_1},\dots,Y_m|_{\CFOpenSet_m}}}$ \rlc{as below for any $\CFOpenSet_j, Y_j \subseteq \syn{\CFVen}$ for $j=1,\dots,m$, modelled by an FCM $\modelM$ with endogenous variables $\syn{\CFVen}$ as}:
\begin{equation} \label{eq:l3-expression-abstract}
\rlb{\input{L3-inC-simple_b.tikz}}
% \tikzfig{L3-abs-query}
% \qquad
% % \tikzfig{L3-abstract-proper}
% \tikzfig{L3-inC-proper}
\end{equation} 
Similarly, given \rl{concrete variables \rlb{$\model{\CFVen}$}\  in $\catC$,}\ 
we define the \rlb{signature $\Lthree(\model{\CFVen})$}\ of \emph{concrete \CF{} queries} 
to have types \rlb{$\syn{\CFVen}$ and a query $\Lthreequery{Y_1|_{\CFDostates_1},\dots,Y_m|_{\CFDostates_m}}$}\ with \rl{outputs $\syn{Y_1,\dots,Y_m}$, and no inputs}, for each choice of subsets  \rlb{$\CFOpenSet_j, Y_j \subseteq \syn{\CFVen}$ and sharp states $\CFDostates_j$ of $\CFOpenSet_j$},  for $j=1,\dots,m$,\footnote{\rl{Concrete variables come with specified input and output subsets; here we require \rlb{$\syn{\CFVen^{\text{in}}}=\emptyset$}\ since the types of \CF{} queries range over only the variables considered endogenous (in any FCM that yields a model of those queries).}} modelled by any FCM $\modelM = \modelF \circ \modelU$  over \rlb{$\CFVen, U$} as: 
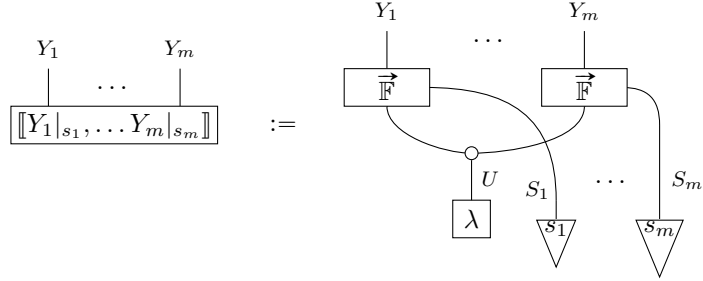
\[
%\tikzfig{conc-expression-2_b}
%\rlb{\tikzfig{conc-expression-2_b-expanded}}
\input{conc-exp-simpler.tikz}
% \qquad
% \tikzfig{conc-expression-2b}
% \semM{\Lthreequery{Y_1|_{x_1},\dots,Y_m|_{x_m}}} \qquad := \qquad  
% \tikzfig{conc-expression}
\]
Concrete \CF{} queries are closely related to \emph{counterfactuals}.  
Essentially, a counterfactual is given by 
a conditional probability distribution arising from a concrete \CF{} query; for details see
%\rlb{$\semM{\Lthreequery{Y_1|_{\CFDostates_1},\dots,Y_n|_{\CFDostates_n}}}_{\modelM}$}\ for a given FCM $\modelM$.  
\cite{lorenz2023causal}.\footnote{Including the relation between terms like \rlb{$\Lthreequery{Y_1|_{\CFDostates_1},\dots,Y_n|_{\CFDostates_n}}$}\ and `conjunctions of counterfactual expressions'.} %Now we can consider a notion of abstraction for functional causal models and \CF{} queries from \cite{xia2024neural}. 

\paragraph{Counterfactual abstraction.} We can now meet the related notion of abstraction, essentially from \cite{xia2024neural}. 

\begin{definition}[Counterfactual abstraction]
Let $\modelML, \modelMH$ be FCMs with concrete endogenous variables \rlb{$\model{\FMCVen}_L, \model{\FMCVen}_H$}. A \emph{\CF{} abstraction} $\modelML \to \modelMH$ is given by a \varalign\footnote{\st{Note that for \rlb{the endogenous variables of an FCM we have $\FMCVen=\Vout=\Vint$}\ and so the $\HtoL(\syn{X})$ must simply be disjoint subsets.}} 
\rlb{$\dvaralign{\HtoL}{\LtoH}{\model{\FMCVen}_L}{\model{\FMCVen}_H}$}\ 
such that for all \CF{} queries 
\rlb{$\Lthreequery{Y_1|_{\CFOpenSet_1},\dots,Y_m|_{\CFOpenSet_m}} \in \Lthree(\syn{\FMCVen_H})$}\ 
the following holds in $\catC$: 
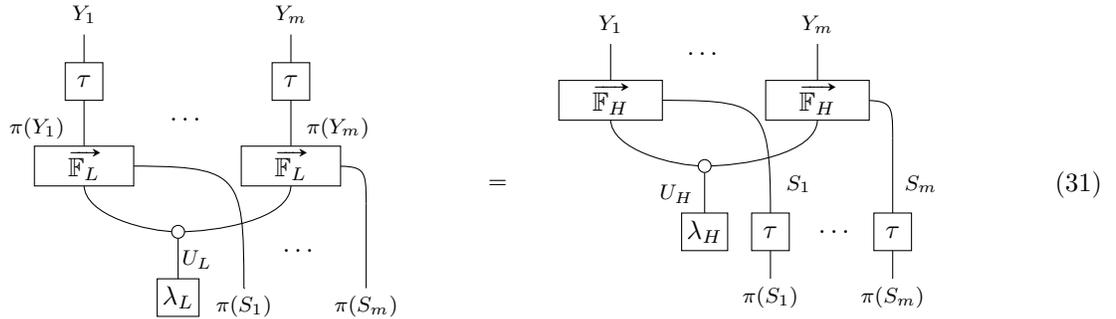
\begin{equation} \label{eq:compatibility-condition}
\rlb{
\input{compat-simpler_b_no_H.tikz}
%\tikzfig{compat-simpler_b}  %% With superscript H everywhere
}
%\tikzfig{compat-simpler}
%\tikzfig{query-simplest}
% \tikzfig{query-simplest-2}
\end{equation}
\end{definition}

We can by now immediately recognise the above as a naturality condition. 

\begin{proposition} \label{prop:consistent-alignment}
A \CF{} abstraction
$\modelML \to \modelMH$ 
is \rl{equivalent to}\ a \qdown{} of the form 
\allowbreak\rlb{$\qdownabs{\HtoL}{\LtoH}{\SigModel{\absLthree(\syn{\FMCVen_L)}}{\modelML}}{\SigModel{\absLthree(\syn{\FMCVen_H})}{\modelMH}}$}\ with:  
\begin{equation} \label{eq:CF-abs-query-map}
\rlb{\HtoL(\Lthreequery{Y_1|_{\CFOpenSet_1},\dots,Y_m|_{\CFOpenSet_m}}) := \Lthreequery{\HtoL(Y_1)|_{\HtoL(\CFOpenSet_1)},\dots,\HtoL(Y_m)|_{\HtoL(\CFOpenSet_m)}}}
\end{equation}
\end{proposition}

\begin{proof} 
Given a \CF{} abstraction, extend $\HtoL, \LtoH$ to products of variables as usual and define a map on queries via \eqref{eq:CF-abs-query-map}. Then \eqref{eq:compatibility-condition} states precisely that consistency (naturality) holds. Conversely, given such a \qdown, consistency ensures \eqref{eq:compatibility-condition} holds. By definition $\HtoL$ sends subsets of variables to subsets of variables, making $\dvaralignpair{\HtoL}{\LtoH}$ a \varalign.
\end{proof}

Again, there is a corresponding form of concrete \qup{}, here applying to concrete counterfactual queries \rlb{$\Lthree(\model{\CFVen})$}. 

\begin{corollary}
Any \CF{} abstraction $\modelML \to \modelMH$
with surjective $\LtoH$ determines \rlc{an}\ \qup{} \rlb{$\qupabsshort{\SigModel{\Lthree(\model{\FMCVen}_L)}{\modelML}}{\SigModel{\Lthree(\model{\FMCVen}_H)}{\modelMH}}$}\ via:
\begin{equation} \label{eq:conc-L3}
\rlb{\LtoHquery(\Lthreequery{\HtoL(Y_1)|_{\CFDostates_1},\dots,\HtoL(Y_m)|_{\CFDostates_m})}
:= 
\Lthreequery{Y_1|_{\LtoH \circ \CFDostates_1},\dots,Y_m|_{\LtoH \circ \CFDostates_m}}}
\end{equation}
\end{corollary}
\begin{proof}
The consistency condition for such \rlc{an}\ \qup{} is: 
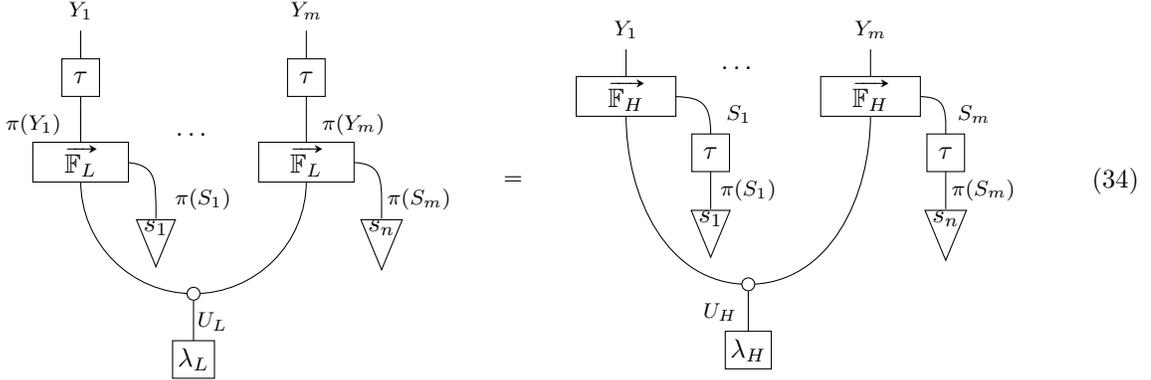
\begin{equation} \label{eq:consistency}
\rlb{
\input{compat-concrete_b_no_H.tikz}
%\tikzfig{compat-concrete_b} %% With superscript H everywhere
}
%\tikzfig{compat-concrete}
\end{equation}
which follows from \eqref{eq:compatibility-condition}. Alternatively, apply Proposition \ref{prop:Exact-trans-from-ref} with 
%$\Stchoice=\Sh$, 
\rlb{$\SigQH=\absLthree(\FMCVen_H)$, $\SigQL=\absLthree(\FMCVen_L)$}. For each such query $\syn{Q}$ as in \rl{\eqref{eq:l3-expression-abstract}}, \rlb{$\syn{Z_Q}$}\ is 
\rlb{$\syn{\CFOpenSet_1 \otimes \dots \otimes \CFOpenSet_n}$.}\ 
\end{proof}

As for other notions in the literature, counterfactual abstractions have previously been introduced as concrete \qup{} \cite{xia2024neural}. By a collection of \rlb{\emph{concrete \CF{} queries}}\ on \rlb{$\model{\CFVen}$}\ we mean a subset of queries \rlb{$\SigQ \subseteq \Lthree(\model{\CFVen})$}\ with types given by the variables \rlb{$\syn{\CFVen}$}. Given \rlb{concrete \CF{} queries $\SigQ$ on $\model{\FMCVen}_H$}\ 
\rl{and a \varalign\ \rlb{$\dvaralign{\HtoL}{\LtoH}{\model{\FMCVen}_L}{\model{\FMCVen}_H}$}\ we say that $\modelML, \modelMH$ satisfy \emph{$\SigQ$-$\LtoH$ consistency} }\  \cite[Def 7]{xia2024neural} when \eqref{eq:consistency} holds, for all queries \rlb{$\Lthreequery{Y_1|_{\LtoH \circ \CFDostates_1},\dots,Y_m|_{\LtoH \circ \CFDostates_m}}$}\ in $\SigQ$. Immediately we then have the following. 

\begin{corollary} 
\rlc{The FCMs}\ 
\rl{$\modelML, \modelMH$ satisfy $\SigQ$-$\LtoH$ consistency iff the data determines}\ \rlc{an}\ \qup{} 
\rlb{$\qupabs{\HtoL}{\LtoH}{\LtoHquery}{\SigModel{\Lthree(\model{\FMCVen}_L)}{\modelML}}{\SigModel{\rl{\sigQH}}{\modelMH}}$}\ via \eqref{eq:conc-L3},  
for all queries \rlb{$\Lthreequery{Y_1|_{\LtoH \circ \CFDostates_1},\dots,Y_m|_{\LtoH \circ \CFDostates_m}}$}\ in $\sigQH$. 
\end{corollary}

\subsection{Distributed causal abstraction} \label{sec:distributed-abstraction}

While notions of abstraction such as constructive abstraction are `local' in that they associate variables $\syn{V}$ at the high-level with \emph{disjoint} subsets $\HtoL(\syn{V})$ at the low-level, \rl{via a \varalign, it has been argued that more general kinds of abstractions without this feature, which we may call `distributed', may be more practical, e.g. in ML \cite{geiger2024finding, geiger2023causal}. 
Indeed, exact transformations and strong causal abstraction from Section \ref{subsec:exact-intv} generically allow for such `distributedness'. In this section we discuss a specific class of distributed abstractions, inspired by \cite{geiger2024finding, geiger2023causal}, making use of the following way to define a causal model.}\

\subsubsection{Model induction}

Throughout this section, $\modelM$ will always refer to a deterministic open causal model over variables $\modelV$ with $\Vout = \Vnin$, such that each variable $X$ has at least one normalised state in $\catC$.

\begin{definition}[Models from \rl{deterministic channels and vice versa}]  
We define the \emph{parallel mechanism channel} $\parmechdiag_\modelM \colon V \to V$ as: 
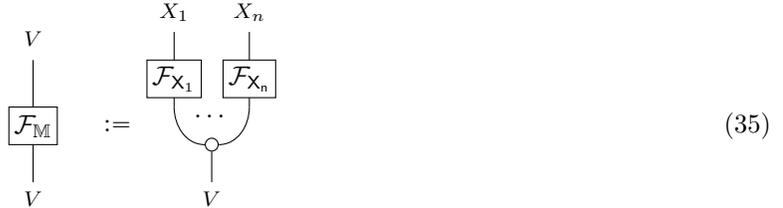
\begin{equation} \label{eq:par-mech}
% \tikzfig{par-mech-2-conc}
\input{par-mech-2-conc-nodots.tikz}
\end{equation}
where $\syn{V} = \syn{X_1,\dots,X_n}$. Here for each non-input variable $\syn{X}$ we set 
\begin{equation}
\input{FX-only-nodots.tikz} \ \ \ := \ \ \  
 \input{FX.tikz} \label{eq:F-mechanism}
\end{equation} 
and for each input variable $\syn{X}$ we define: 
\begin{equation} 
\input{FX-only-nodots.tikz} \ \ \ := \ \ \  \input{mech-input-only.tikz} \label{eq:mechanism-inputs}
\end{equation}
Conversely, let $\parmechdiag \colon V \to V$ be a deterministic channel in $\catC$ and write $\parmechdiag_X$ for each marginal on $X$. We define a signature $\Sig_\parmechdiag$ and model $\modelfrom{\parmechdiag}$ as follows. We declare that $\syn{X}$ is an input variable when \eqref{eq:mechanism-inputs} holds. Otherwise we define $\Pa(\syn{X})$ to be the least subset of variables, and $c_X$ the corresponding channel such that $\parmechdiag_X$ factors as in \eqref{eq:F-mechanism}, which exists by \cite[Lemma 112]{lorenz2023causal}. We then include a generator $c_\syn{X} \colon \syn{\Pa(X) \to X}$ with this representation. 
 We call the channel $\parmechdiag$ \emph{acyclic} when $\Sig_\parmechdiag$ defines a valid causal signature, i.e.~these parent sets $\Pa(X)$ induce a directed graph on $\syn{V}$ which is acyclic. 
\end{definition}

It is easy to see that if $\modelM$ is faithful\footnote{\rl{In the sense of mechanism-faithfulness of \cite{lorenz2023causal}: the parental sets $Pa(X)$ of $\modelM$ are already the minimal ones such that \eqref{eq:F-mechanism} holds, i.e. $X$ depends through $c_X$ on every of its parent variables.}} as an open deterministic causal model then $\parmechdiag_\modelM$ is acyclic with $\modelfrom{\parmechdiag_\modelM} = \modelM$.  The above construction lets us conversely define a causal model $\modelfrom{\parmechdiag}$ from any such acyclic channel $\parmechdiag$ on $\varVconc$, with \rl{$\parmechdiag_{\modelfrom{\parmechdiag}} = \parmechdiag$}. 

A useful fact connecting the input-output and parallel mechanism views of a model is the following.
\begin{lemma} \label{lem:fixpoints}
For any sharp state \rl{$v = (i, o)$ of $\Vinconc \otimes \Vninconc$ we have 
\begin{equation}
\input{io-view-vs-fixpoints.tikz}
\end{equation}
}
\end{lemma}
\begin{proof} 
Taking marginals, the latter holds iff $v_X = c_X \circ v|_\Pa(X)$ for all variables $\syn{X} \in \Vnin$. By carrying out induction over the longest chain of ancestors of a variable, one can see that this is equivalent to \rl{$o = \modelMio \circ i$}.  
\end{proof}

\begin{example} \label{ex:II-PMs} 
It is straightforward to verify that Do-interventions and interchange interventions alter $\parmechdiag_\modelM$ as follows.
\begin{equation} \label{eq:II-diag}
% \tikzfig{par-mech-do}
\input{par-mech-do-nodots.tikz}
\qquad \qquad \qquad 
% \tikzfig{par-mech-II-nodots}
\input{par-mech-II-nodots.tikz}
\end{equation}
\end{example}

We can now make use of the construction above to define the following.

\begin{definition}[Model Induction] \label{def:model-induction}
Let $\modelM$ be an open deterministic causal model over \rl{variables $\modelV$}. 
Let \rl{$\model{W}$ be a further set of concrete}\ variables and $\modin \colon V \simeq W$ a deterministic isomorphism. We say that $\modin$ \emph{respects} $\modelM$ when the channel: 
\begin{equation} \label{eq:induced-model}
% \tikzfig{induced-model}
\input{induced-model-nodots.tikz}
\end{equation}
is acyclic, with inputs $\Win$,\footnote{\rl{Recall that as concrete variables, $\model{W}$ comes with subset ${W^{\text{in}}}$ and the condition here is that the inputs as defined via \eqref{eq:mechanism-inputs} coincide with that set ${W^{\text{in}}}$.}} and moreover: 
\begin{equation} \label{eq:respects-inputs}
\input{modelin.tikz}
\end{equation}
for isomorphisms $\modinin$, $\modinnin$. In this case we define the \emph{induced} causal model $\inmod{\modelM} := \modelfrom{\parmechdiag_{\inmod{\modelM}}}$. 

Similarly, given a set of interventions $\Intset$ on $\modelM$, we say that $\modin$ \emph{respects} $\SigModel{\Intset}{\modelM}$ when it respects $\transform{\inti}{\modelM}$ for each $\inti \in \Intset$. In this case for each $\inti \in \SigI$ we define the induced intervention $\inducedint{\inti}$ on $\inmod{
\modelM}$ via $\transform{\inducedint{\inti}}{\inmod{\modelM}} := \inmod{\transform{\inti}{\modelM}}$. 
 We denote the set of such interventions on $\inmod{\modelM}$ by $\inmod{\Intset}$. 
\end{definition}

Now let us observe that such model induction does indeed yield a special case of \rlc{\qupshort}. Recall that any interventions $\Intset$ on $\modelM$ form a \rlb{signature of queries}\ $\Intsetio$ with types \rl{$\syn{\invar, \outvar}$}\ and a query $\syn{\inti \colon \invar \to \outvar}$ for each $\inti \in \Intset$. In this section \rlc{we}\ always assume each such set $\Intset$ contains the trivial intervention $\noint$ corresponding to $\modelM$ itself.  
Recall that, as in the proof of Proposition \ref{prop:ET-equiv-qup}, for \qup s with respect to query signatures $\Intsetio$ there is only one possible mapping $\HtoL$ on types ($\HtoL(\syn{in_H}) \rlc{=} \syn{in_L}$ and $\HtoL(\syn{out_H}) \rlc{=} \syn{out_L}$).  
We will hence drop the corresponding $\HtoL$ throughout the rest of this section and refer to data of \rlc{an}\ \qup\ by just $(\LtoHquery, \LtoH)$. 

Finally, we\ will assume that the semantics category $\catC$ has \emph{enough states}, meaning that whenever we have $f \circ x = g \circ x$ for all sharp states $x$ we have $f = g$; this is true for $\FStoch$.

\begin{proposition} \label{prop:Model-Induction}
Let $\modelM$ be an open deterministic causal model \rl{over $\modelV$}\ with set of interventions $\Intset$. Suppose that $\modin \colon V \to W$ is a deterministic isomorphism which respects $\SigModel{\Intset}{\modelM}$. Then 
\[
\qupabslarge{}{\modin}{\inducedintmap}{\SigModel{\IntQsimple{\Intset}}{\modelM}}{\SigModel{\IntQsimple{\inducedintmap(\SigI)}}{\inmod{\modelM}}}
\]
defines a (bijective) \qup.  
Explicitly, for all $\inti \in \SigI$ we have: 
\begin{equation} \label{eq:ex-trans-induced}
\input{induced-ex-trans.tikz}
\end{equation}
In particular $\modinnin \circ \modelMio = \io{\inmod{\modelM}} \circ \modinin$.
\end{proposition}

\begin{proof} Appendix \ref{app:proofs}.
\end{proof}

\subsubsection{Composing induction and abstraction} \label{subsubsec:isoCCA} 
\rl{Seeing as \qup{}s are closed under composition, the above discussed notion of model induction yields a natural construction: 
consider a model $\modelML$ over $\modelV_L$ and a deterministic isomorphism, which induces an `intermediate' model, which in turn is abstracted by a higher-level model $\modelMH$.}\ 
The resulting composite can yield a `distributed' abstraction relation between $\modelML$ and $\modelMH$. 
In more detail, suppose we have:
\begin{itemize}
\item A `low-level' model $\Mone=\modelfrom{\parmechdiag_{\Mone}}$ \rl{on $\modelVone$;} 
\item A determinstic isomorphism $\modin \colon \Vone \to \Vtwo$ which respects $\Mone$;
\item A set of interventions $\Intwo$ on $\Mtwo := \inmod{\Mone}$ respected by $\modin^{-1}$;
\item \rlc{An}\ \qup{} $\qupabs{}{\LtoH}{\LtoHquery}{\SigModel{\IntQsimple{\Intwo}}{\Mtwo}}{\SigModel{\IntQsimple{\Inthree}}{\Mthree}}$ to a model $\Mthree$ \rl{on $\modelVthree$}.
\end{itemize}
Then defining $\Inone := \inducedintmap^{-1}(\Intwo)$ on $\Mone$ we have a sequence of \qup s:\ 
\begin{equation} \label{eq:induction-circ-u-abs}
	% \tikzfig{exact-trans-pair}
	\input{exact-trans-pair-new.tikz}
\end{equation}
which composes to yield an overall  \qup{} \rlb{$(\LtoH \circ \modin, \LtoHquery \circ \inducedintmap)$}\ from $\Mone$ to $\Mthree$. Explicitly, for each intervention $\inti\in \Intwo$ on which $\LtoHquery$ is defined we have the following:

\begin{equation} \label{eq:compose-exact-trans-explicit}
\input{composing-exact-trans_b.tikz}
\end{equation}

This construction is of particular interest when the \qup{} \rlb{$(\LtoH, \LtoHquery)$}\ arises from a constructive or interchange abstraction. 
This yields an overall notion which we discuss next, implicitly considered in the works of \cite{geiger2024finding, geiger2023causal},\footnote{\rl{They introduced a `distributed version' of interchange interventions, though without a derivation of its form and without a definition of an associated notion of abstraction.}}. 

\begin{definition}[Iso-constructive/interchange abstraction] \label{def:iso-CCA} %\isoCCAcaps
	By an \emph{iso-constructive} (resp.~\emph{iso-interchange}) \emph{abstraction} we mean the above setup as in \eqref{eq:induction-circ-u-abs}, where $(\LtoHquery,\LtoH)$ is induced by a constructive (resp. interchange) abstraction.  
\end{definition}

\paragraph{Iso-constructive abstraction.}
It is worth spelling out what \isoCCA\ amounts to in more detail. 
Let $(\HtoL, \LtoH)$ be the data of the \CCA\ that induces $(\LtoHquery,\LtoH)$, $S \subseteq \Vthree$ and $p$ a sharp state of $\HtoL(S)$.  
The composite \qup{} then relates $\Do$-interventions $\Do(S = \LtoH \circ p)$ on $\Mthree$, 
via corresponding $\Do$-interventions $\Do(\HtoL(S)=p)$ on $\Mtwo$, 
to interventions on $\Mone$ of the form  $\inducedintmap^{-1}(\Do(\HtoL(S)=p))$, which we refer to as \emph{distributed Do-interventions}, defined generally for any $\Stwo \subseteq \Vtwo$: 
\[
	\DDo(\Stwo=\stwo,\rho) := \inducedintmap^{-1}(\Do(\Stwo=\stwo)) 
\] 
Thus we are interested in distributed Do-interventions on $\Mone$ with $S'=\HtoL(S)$. One may compute (proven in Appendix \ref{app:proofs}) that the parallel mechanism channels for these intervened low-level models are as follows. 
\begin{equation} \label{eq:DDo-pic}
\input{DI-relabel2_b.tikz}
\end{equation}

\isoCCAcaps s yield a well-characterised class of exact transformations with genuine `distributedness' in general.\footnote{This also motivates the general notion of exact transformation (Def.~\ref{def:exact-trans}) going beyond mere do-interventions.}
In fact by construction it is of the strong kind. 
\begin{proposition}\label{prop:isoCCA-is-strongCA}
	For an \isoCCA\ the overall \qup{} \rlb{$(\LtoH \circ \modin, \LtoHquery \circ \inducedintmap)$}\ from $\Mone$ to $\Mthree$ is a \strongCA\ (see Def.~\ref{def:strong-CA}).   
\end{proposition}

\paragraph{Iso-interchange abstraction.} Similarly for \isointabs, let $\dvaralignpair{\HtoL}{\LtoH}$ be the \varalign\ of the \restrictedCCA\ that induces \rlb{$(\LtoH, \LtoHquery)$}, let $S_1,\dots,S_n$ be disjoint subsets of $\Vthree$ and $\{x_1,\dots,x_n\}$ input states of $\Vone^{\text{in}}$. 
The composite \qup{} then relates interchange interventions $\II((S_j,\LtoH \circ \modinin \circ x_j))^n_{j=1}$ on  $\Mthree$, 
via interchange interventions $\II((\HtoL(S_j), \modinin \circ x_j))^n_{j=1}$ on $\Mtwo$ 
to interventions on $\Mone$ of the form $\inducedintmap^{-1}(\II((\HtoL(S_j), \rlb{ \modinin \circ} x_j))^n_{j=1})$, which we refer to as \emph{distributed interchange interventions}, defined generally for any disjoint subsets $Y_1,\dots,Y_n$ of $\Vtwo$:
\[
%\DII((Y_j,x_j)^n_{j=1},\modin) := \inducedintmap^{-1}(\II(Y_j, x_j)^n_{j=1})
\DII((Y_j,x_j)^n_{j=1},\modin) := \inducedintmap^{-1}(\II(Y_j,\modinin \circ x_j)^n_{j=1})
\]
Thus we are interested in distributed interchange interventions on $\modelL$ with the $Y_j = \HtoL(S_j)$. Setting $S:=\cup_j S_j$, one may again compute the parallel mechanism channels for these intervened low-level models to be:
\begin{equation} \label{eq:DII-pic}
\input{DII-relabel_b.tikz}
\end{equation}
which we prove in Appendix \ref{app:proofs}.

\paragraph{Consistency for iso-abstractions.}
\rl{Finally, observe that above we showed the models yielded by low-level `distributed' interventions in terms of parallel mechanism channels, but did not spell out the consistency condition for the overall \qup. 
The latter refers to models as input-output processes $\Vin$ to $\Vout$, while the former can only be made explicit in terms of $\parmechdiag_\modelM$ channels. The relation between the two would involve a quantification over fixed points of $\parmechdiag_\modelM$ channels as in Lemma~\ref{lem:fixpoints}. 

There actually is a way to render the relation between the two views `structural' allowing for a diagrammatic consistency condition of, e.g., \isoCCA s, without explicit reference to fixed points. Since it however needs a little more formalism we present this in App.~\ref{app:abs-in-pm-via-tracedC}.

In any case, given candidate data for an \isoCCA\ we can determine whether it holds as follows. 
For each high-level intervention $\LtoHquery(\inti) = \Do(S = \LtoH \circ p)$ we first consider $\io{\Mthree_{\LtoHquery(\inti)}}$ and the composite on the RHS of \eqref{eq:compose-exact-trans-explicit}. 
On the other hand, with $\inducedintmap^{-1}(\inti) = \inducedintmap^{-1}(\Do(\HtoL(S)=p))$ we can compute the induced model $\transform{\inducedintmap^{-1}(\inti)}{\Mone}$ via the channel \eqref{eq:DDo-pic} and reconstruct $\io{\transform{\inducedintmap^{-1}(\inti)}{\Mone}}$ and the composite on the LHS of \eqref{eq:compose-exact-trans-explicit}, comparing it against the RHS on various inputs. 
An analogous procedure works for \isointabs.}

\section{\strucdowncaps} \label{sec:structure-refinements}

So far we have considered abstraction relations at the level of queries. The categorical perspective however also suggests \rlb{to study}\ stronger notions of abstraction that apply at the \emph{structure} level of compositional models, i.e.~in terms of their individual components. Throughout this section we consider all compositional models \rlc{to be models in $ \catC$ and}\ to come with given \rl{structure \emph{and} query}\ signatures \rlc{$(\modelM,\SigS, \SigQ, \catC)$, where}\ all queries are abstract (Def.~\ref{def:abstract-query}).

\begin{definition}[\strucdowncaps]  \label{def:complevel}
A \emph{\strucdownfull}: 
\[
\strucdownabslarge{\HtoL}{\HtoLS}{\LtoH}{(\modelML,\SigS_L,\sigQ_L)}{(\modelMH,\SigS_H,\sigQ_H)}
\]
is given by functors $\HtoL, \HtoLS$ 
and an epic natural transformation $\LtoH$ as below, such that $\HtoL$ maps queries \rl{in $\sigQ_H$ to queries in $\sigQ_L$}\ and the upper square commutes:
\begin{equation} \label{eq:structure-comm-square}
% \tikzfig{mech-ref}
\input{mech-ref-swap_b.tikz}
\end{equation}
A \strucdown{} is \emph{strict} when $\LtoH=\id{}$. 
\end{definition}

As the terminology suggests, any such data indeed is a special case of an abstraction at the query \rlb{level, too}, by the following straightforward result.

\rl{
\begin{proposition} \label{prop:c-level-is-d-abs}
	Any \strucdownfull~\eqref{eq:structure-comm-square}  
	induces a \qdown{} $\qdownabs{\HtoL}{\LtoH'}{\SigModel{\sigQL}{\modelML}}{\SigModel{\sigQH}{\modelMH}}$ where $\LtoH' := {\LtoHS} \circ 1_{\qabsMHfunc}$. 
\end{proposition}
}
\begin{proof} Appendix \ref{app:proofs}. \end{proof} 

Explicitly then, a \strucdown{} assigns to each high-level variable \rl{$\syn{X}$} a list of low-level variables \rl{$\HtoLS(\syn{X})$}\ from $\SigSL$ and a (deterministic) channel $\LtoHS \colon \rl{\HtoLS(X) \to X}$ in $\catC$, extending to lists of high-variables via products as usual \eqref{eq:tau-one}, \eqref{eq:tau-factors}. Next, it sends each high-level component $\syn{f}$ to a low-level diagram $\HtoLS(\syn{f})$ in the low-level structure category $\strucSL$, such that the following holds in $\catC$: 
\[
\input{mechanism-refinement-equation.tikz}
% \tikzfig{mechanism-refinement-equation-nodots} % Optional-nodots
\]
Then $\HtoLS$ automatically extends to a functor, and $\LtoH$ a natural transformation, or in other words the above holds for arbitrary diagrams $\syn{f}$ in $\strucSH$, i.e. composites of components.

We also as usual have a map $\HtoL$ from the high to low level on queries and types sending $\syn{Q} \colon \syn{X} \to \syn{Y}$ to $\HtoL(\syn{Q}) \colon \HtoL(\syn{X}) \to \HtoL(\syn{Y})$. 
Finally, the \rl{component}\ mapping respects the diagrams for these (abstract) queries, in that for each such $\syn{Q}$ we have the following equality in the low-level structure category $\strucSL$: 
\[
\input{implementation-3.tikz}
% \tikzfig{implementation-3-nodots} % Optional-nodots
\]
In summary, as well as mapping high-level queries to low-level queries with consistent representations in $\catC$, a \strucdownfull~ now also comes with a consistent mapping from high-level formal diagrams to low-level formal diagrams. Moreover when applied to the diagram for a high-level query, we obtain an identical low-level diagram \rlb{for}\ its corresponding low-level query.

 \subsection{Strict \strucdown s}

An important special case are strict \strucdown{}s, where the natural transformations are trivial, with $\LtoH=\id{}$. Though we have defined them last, these are perhaps the simplest way in which we can relate a high-level model to a low-level model. Here we simply `unbox', `refine' or `implement' each high-level component directly as a diagram of low-level components (while respecting the diagrams associated with queries).
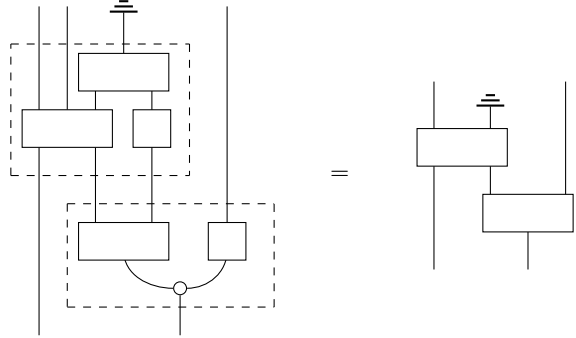
\[
\input{unboxing-generic_b.tikz}
\]

\begin{examples} \label{ex:strict-a-abs} \
\begin{enumerate}

\item 
By an open causal model $\modelM$ \emph{of} a channel $c$ in $\catC$ we mean one with $\modelMio = c$. 
\[
\input{modelof_b.tikz}
\]
Equivalently, there is a strict \strucdown{} $\strucdownabsshort{\modelM}{\syn{c}}$ preserving the inputs and outputs. Here we consider both models to have a single query $\syn{io} \colon \syn{in} \to \syn{out}$, sent to $\modelMio$ and $c$ respectively, and $\syn{c}$ to have variables the $\syn{X_i}, \syn{Y_j}$, and a single \rl{component}\ $\syn{c}$ which abstractly represents $\syn{io}$. \footnote{Hence $c$ is itself a causal model only if $\modelM$ and $c$ have only \rlb{one}\ output.}

For example, consider a distribution $\omega$ as a state in $\FStoch$. A causal model of $\omega$ in $\catC$ then is one whose network diagram composes to yield $\omega$, i.e.~the usual notion of a `causal model for a distribution'. More concretely, consider a distribution $\omega$ as below where $S$ represents whether someone smokes or not, $L$ whether they develop lung cancer or not and $A$ their age. The left-hand causal model below, with an additional latent variable $B$ for socio-economic backgrond conditions, would be a causal model of $\omega$ when the following equality holds (see \cite{CorreaEtAl_2020_CalculusForStochasticInterventions, lorenz2023causal}).

\[
%\tikzfig{causal-model-of}
\input{causal-model-of_b.tikz}
\]

\item \label{enum:FCM-refine}
Strict \strucdown{}s can also capture the way in which we can see an FCM as `refining' a causal model. \rl{Often an FCM $\modelF$ with endogenous and exogenous variables \rlb{$\FMCVen = \{X_i\}^n_{i=1}$}\ and $U = \{U_i\}^n_{i=1}$, respectively (Def.~\ref{def:FCM}), is considered to represent the (candidate) deterministic mechanisms $f_i$ that underlie a (non-functional) causal model $\modelM$ over the variables \rlb{$\FMCVen$}, where all randomness comes from our ignorance about the variables $U$, captured in the respective `noise' distributions $\lambda_i$.   
Indeed, given any FCM $\modelF$ we can define a causal model \rlb{$\modelF|_{\FMCVen}$ over just $\FMCVen$ by assigning each $X_i \in \FMCVen$}\ the mechanism $c_i$ defined by:
\begin{equation} \label{eq:deSCM}
%\tikzfig{de-SCM}
\input{de-SCM_b.tikz}
\end{equation} 
where $\Pa'(X_i)=\Pa(X_i) \setminus \{U_i\}$ with $\Pa(X_i)$ the parental sets with respect to $\modelF$.  
Conversely, given a (non-functional) causal model $\modelM$ over \rlb{$\FMCVen$}\ there are in general many ways to refine it to an FCM $\modelF$ such that \rlb{$\modelF|_{\FMCVen} = \modelM$}.  
One then obtains a strict \strucdown{} \rlb{$\strucdownabsshort{(\modelF,\SigS_{\modelF},\Openqueries(\syn{\FMCVen}))}{(\modelM,\SigS_{\modelM},\Openqueries(\syn{\FMCVen}))}$}, 
where we identify the $\syn{X_i}$ and $X_i$ in both models and map each $\syn{c_i}$ to the composite of $\syn{f_i}$ and $\syn{\lambda_i}$ as syntax, such that \eqref{eq:deSCM} holds in $\catC$. 
We can depict such a relation as below, where we identify variables with the same names and corresponding dashed boxes.

\begin{equation}
%    \tikzfig{FCMrefinement} 
	\input{FCMrefinement_c.tikz}
\end{equation}
}
\end{enumerate}
\end{examples}

\subsection{\strongCCAtitle} \label{sec:strong-constructive-abstraction}

In principle one may ask when any of the causal \qdown s from Section~\ref{sec:causal-abstraction} more strongly extend to the component level. In this section we will do so for constructive abstractions, calling a constructive causal abstraction which extends to a \strucdown~a \emph{\strongCCA}. This gives us a novel stronger notion of  causal abstraction not so far considered in the literature.\footnote{This also generalises Example~\ref{ex:strict-a-abs} \eqref{enum:FCM-refine}, which is an instance with $\LtoH=\id{}$.}

Explicitly, consider a constructive abstraction as a \qdown{} 
$\qdownabs{\HtoL}{\LtoH}{\SigModel{\Openqueries(\VL)}{\modelML}}{\SigModel{\Openqueries(\VH)}{\modelMH}}$ 
(Thm.~\ref{Thm:CCA-as-refinement}). To form a \strongCCA~would mean that there is a further mapping $\HtoLS$~sending each high-level mechanism to a low-level diagram: 
\begin{equation} \label{eq:strong-CA}
%\tikzfig{strong-CA-1}
% \tikzfig{strong-CA-1_b}
\input{strong-CA-1_b-nodots.tikz}
\end{equation}
such that the following must hold in $\catC$ (as usual using our font convention to omit $\semcompM{\modelML / \modelMH}{-}$):
\begin{equation} \label{eq:strong-CA-nat}
%\tikzfig{strong-CA-2}
%\tikzfig{strong-CA-2_b}
% \tikzfig{strong-CA-2_c}
\input{strong-CA-2_c-nodots.tikz}
\end{equation}
We then freely \rlb{extend}\ \rlc{$\HtoLS$}\ to a functor from high-level diagrams to low-level diagrams, with $\HtoLS(X) = \HtoL(X)$ for all $X \in \VH$.\footnote{Here we use that the variables of $\SigS_L$ are the types of $\Openqueries(\VL)$. Note that as before the naturality condition \eqref{eq:strong-CA-nat} then automatically extends to arbitrary high-level diagrams.} 
Moreover, for every \rl{subset $\syn{S} \subseteq \VninH$}\ we have an equality in $\strucSL$, purely at the diagrammatic level:
\begin{equation} \label{eq:strong-CA-diag}
\input{MDSL.tikz}
\qquad = \qquad
\HtoLS \left( \input{MDSH.tikz} \right)
% \qquad \qquad \qquad \qquad 
% \HtoLS \left( \tikzfig{MDSH2} \right)
% \qquad = \qquad 
% \tikzfig{MDSL2}
\end{equation}
That is, opening the low-level model at $\HtoL(S)$ yields the same formal diagram as applying $\HtoLS$ to the diagram for opening the high-level model at $S$. 

We will see \rl{shortly}\ that, under very mild assumptions on the model, the mapping $\HtoLS$ is unique whenever it exists, and hence we can view this as a \emph{property} of the abstraction $(\HtoL,\LtoH)$, as our terminology suggests. 
\rl{To this end the following purely graph-theoretic definition will be useful.}\  

\rl{
\begin{definition} \label{def:abstraction-abstract-conds}
Let $\pi$ be the partition of a constructive causal abstraction $(\HtoL,\LtoH)$ between $\modelML$ and $\modelMH$, and for each $\syn{X} \in \varH$ define:
\begin{align*}
\mechlevelset(\syn{X}) :&= 
% \An(\HtoL(X))  \ \setminus \  \An(\HtoL(\Pa(X))) \\ 
%  &=
\{\syn{Z} \in \varL \mid \exists \text{ a directed path } Z \to \pi(\syn{X}) 
\text{ in $\dagGL$ which does not pass through $\pi(\Pa(\syn{X}))$} 
\} 
\end{align*}
We say that $\pi$ is:
\begin{itemize}
\item 
\emph{\strong} if $\mechlevelset(\syn{X}) \cap \pi(\syn{Y}) = \emptyset$ for $\syn{X} \neq \syn{Y}$; 
\item 
\emph{\extrastrong}\footnote{Note that $\pi(\syn{X}) \subseteq \mechlevelset(X)$ by definition.} if $\mechlevelset(\syn{X}) \cap \mechlevelset(\syn{Y}) = \emptyset$ for $\syn{X} \neq \syn{Y}$; 
\item 
\emph{full} if for all $\syn{Y} \in \pi(\Pa(X) \setminus \VinH)$ there exists a directed path $\syn{Y} \to \pi(X)$ in $\dagGL$. 
\end{itemize}
\end{definition} 
\noindent Intuitively, $\mechlevelset(X)$ captures the degree to which the low-level representation of $\Pa(X)$ `screens off' that of $X$.}\ 

We \rl{can now}\ name the conditions on $(\HtoL,\LtoH)$ to obtain the equality of diagrams \eqref{eq:strong-CA-diag}. In fact we will present several results in one, since the conditions will depend on precisely how one chooses to define the structure category $\strucS$ for each model, generated by its signature $\Sig$, for which we will consider a few options with a corresponding result for each. 

For the first of these options we consider a slight weakening of the notion of Markov category from e.g.~\cite{cho2019disintegration}. A \emph{cd-category} $\catC$ is defined just like a Markov category except with a chosen (rather than unique) discard morphism $\discard{A} \colon A \to I$ on each object $A$, satisfying:
% such that :
\[
\input{disc-nat.tikz} \qquad \qquad  \qquad \input{disc-I.tikz}
\]
A \emph{cd-functor} $F \colon \catC \to \catD$ must also respect these discarding morphisms.\footnote{ 
%We also require $\discard{I} = 1$ and $\discard{A \otimes B} = \discard{A} \otimes \discard{B}$. 
In a cd-category we call a morphism $f$ a \emph{channel} when $\discard{} \circ f = \discard{}$. Cd-functors $F$ and natural transformations $\alpha$ are like the Markov case, but also with $F(\discard{A}) = \discard{F(A)}$ and all transformation components $\alpha_X$ (and structure isomorphisms for $F$) being channels.}
%between cd-categories is a strong monoidal functor which respects discarding morphisms.%A \emph{cd-functor} $F$ must satisfy $F(\discard{A}) = \discard{F(A)}$.

\begin{definition} \label{def:struc-well-behavedness}
\rl{Given causal models $\modelML, \modelMH$ with respective free structure categories $\strucSL, \strucSH$, call}\ a 
collection $(\HtoL(\syn{X}))_{X \in \syn{\VH}}$ of disjoint subsets of $\syn{\VL}$ \emph{structurally well-behaved} if $\HtoL$ is: 
\begin{itemize}
\item 
\extrastrong\ and full, when using the free cd-category $\strucS = \FreeCD(\Sig)$ generated by $\Sig$ (as in \rl{\cite{lorenz2023causal, tull2024towards}} and essentially \cite{jacobs2019causal}); 
\item 
\extrastrong, when using the free Markov category $\strucS = \FreeMarkov(\Sig)$, where we equivalently add axioms stating that every mechanism is a channel (as in \cite{fritz2023free}):
\[
\input{mechanisms-channels.tikz}
\]
\item 
\strong, when using the free Cartesian category $\strucS = \FreeCart(\Sig)$, where we equivalently add yet further axioms stating that each mechanism is a deterministic channel: 
\[
\input{mechanisms-det.tikz}
\]
\end{itemize}
\st{
We call a cd-functor $\HtoLS \colon \strucSH \to \strucSL$ \emph{structurally well-behaved} if each $\HtoLS(\syn{c_X})$ is a normalised network diagram in $\strucSL$ and $(\HtoL(\syn{X}))_{X \in \syn{\VH}} := (\HtoLS(\syn{X}))_{X \in \syn{\VH}}$ are disjoint subsets with $\HtoL(\VinH) = \VinL$, $\HtoL(\VoutH) \subseteq \VoutL$, which are 
structurally well-behaved, in each case.
}
\end{definition}

\color{black}

Going down \rl{the list in Def.~\ref{def:struc-well-behavedness}}, we include more diagrammatic axioms in $\strucS$ and hence the equality \eqref{eq:strong-CA-diag} is easier to obtain. 
\rl{We are now able to prove our main result characterising \strongCCA{}}.
\begin{theorem} \label{Thm:strongCCAnew}
\st{Let $\modelL, \modelH$ be causal models in $\catC$ such that $\VninH \subseteq \VoutH$}. Suppose every input \rl{variable in $\VinL$ of}\ $\modelML$ has at least one normalised state in $\catC$.\footnote{More generally, for the result we only require that for every input variable $\syn{V} \in \Vin$ and all $f, g \colon X \to Y$ in $\catC$ we have $\discard{V} \otimes f = \discard{V} \otimes g \implies f = g$.} 
For each structure-type (cartesian, Markov or cd) it is equivalent to specify the following: 
\begin{enumerate}
\item \label{enum:mech-level-CCA}
\rl{A \strongCCA{}}\ $(\HtoL,\HtoLS,\LtoH)$;
\item \label{enum:constr-CCA-well}
\rl{A \CCA}\ $(\HtoL,\LtoH)$ such that $\HtoL$ is structurally well-behaved; 
\item \label{enum:pis-functor}
A structurally well-behaved cd-functor $\HtoLS$ and natural transformation $\LtoH$ as below:
\[
%\tikzfig{struc-refinement}
\input{struc-refinement_b.tikz}
\]
 \end{enumerate} 
\end{theorem}

\begin{proof}
Appendix \ref{app:strongCA-proof}.
\end{proof}

The result tells us that a \CCA~extends to the mechanism-level iff its partition $\HtoL$ satisfies the corresponding conditions: being \rl{\strong, \extrastrong, or \extrastrong\ and full, depending on the}\ choice of structure type. While we have stated a result for $\strucS = \FreeCD(\Sig)$, for a typical causal model it is perhaps more natural to use the free Markov category $\strucS = \FreeMarkov(\Sig)$, since we know every mechanism will be a channel in $\catC$, and in this case $\HtoL$ must be \extrastrong. For deterministic models however, such as neural networks, it is natural to take $\strucS = \FreeCart(\Sig)$, and now we only require $\HtoL$ to be \strong. 

Equivalently, the final condition expresses such an abstraction instead purely in terms of a functor at the structure level. When using free Markov categories as structure, this final condition can be captured in a yet simpler form.

\begin{lemma} \label{lem:Markov-equiv}
Let $\modelML, \modelMH$ be causal models such that $\VninH \subseteq \VoutH$. A cd-functor $\HtoLS \colon \FreeMarkov(\Sig_\modelMH) \to \FreeMarkov(\Sig_\modelML)$ is structurally well-behaved iff $\HtoLS(\netdiag{\modelMH})$ is again a network diagram in $\strucSL$ with $\HtoLS(\VinH) = \VinL$, $\HtoLS(\VoutH) \subseteq \VoutL$ and $\HtoLS(\syn{\VninH}) \subseteq \syn{\VninL}$. 
\end{lemma}

\begin{proof}
Appendix \ref{app:strongCA-proof}.
\end{proof}

When using free Markov categories, specifying a \strongCCA\ is therefore equivalent to simply giving a cd-functor $\HtoLS$ and transformation $\LtoH$ which respect the overall network diagrams in the above sense. Perhaps surprisingly, this tells us that we can define this potentially broad class of constructive abstractions purely in terms of mechanisms and input-output behaviour, and thus without explicit reference to interventions or (Do-)queries at all. An open question is then whether all interesting cases of constructive abstraction in practice are of this form, i.e.~such that their partition $\HtoL$ is \strong. 

Let us now give some examples and non-examples of \strongCCA s. 

\begin{example}
Example~\ref{ex:CCA} of a constructive abstraction is in fact mechanism-level when taking $\strucS = \FreeCart(\Sig)$, which is an appropriate choice since the model is deterministic, and in fact more strongly also when taking $\strucS = \FreeMarkov(\Sig)$. This is easy to check via $(1) \Leftrightarrow (2)$ of Thm.~\ref{Thm:strongCCAnew}, since $\HtoL$ is clearly simple, and in fact it is straightforward to see it is \extrastrong~also.\footnote{Note that the conditions in Eq.~\eqref{eq:CCA-example-mechanism-consistency} of Ex.~\ref{ex:CCA} make it straightforward to also establish the relation more directly at the mechanism level due to $(1) \Leftrightarrow (3)$ of Thm.~\ref{Thm:strongCCAnew} together with Lem.~\ref{lem:Markov-equiv}. That $\HtoLS(\netdiag{\modelMH})$ is a network diagram is obvious.} 
\end{example}

Further pedagogical examples, which illustrate the different cases of Def.~\ref{def:struc-well-behavedness}, follow below.

\begin{example} \label{example:notsimple}
Let \rl{$\modelML$ and $\modelMH$ be causal models}\ defined by the following \rl{respective network diagrams in $\catC$:} 
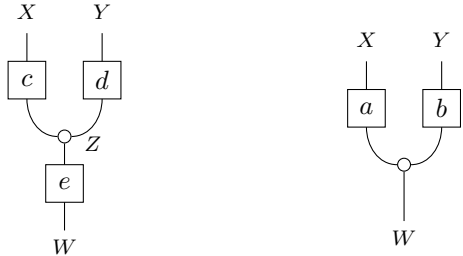
\[
% \tikzfig{CA-ex-1-DAG}
% \qquad \
% \tikzfig{CA-ex-1-DAGL}
% \qquad
\input{CA-ex-1-DL.tikz}
\qquad \qquad \qquad \qquad 
\input{CA-ex-1-DH.tikz}
\]
Suppose that, in $\catC$, the morphisms $a,b,c,d,e$ are all deterministic and that $a=c \circ e$ and $b = d\circ e$. 
Then there is a \CCA\ with $\LtoH=\id{}$ given by setting $\syn{\HtoL :: X \mapsto \{X\}, Y \mapsto \{Y\}, W \mapsto \{W\}}$. For example for the high-level input-output channel we obtain $\modelMHio = \modelMLio$ since:  
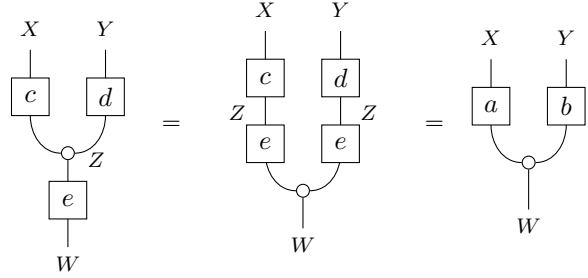
\[
%\tikzfig{CA-ex-1-equality}
\input{CA-ex-1-equality_b.tikz}
\]

This abstraction is not \extrastrong{} since for this $\pi$ we have $\mechlevelset(\syn{X}) = \{\syn{X, Z}\}$, $\mechlevelset(\syn{Y}) = \{\syn{Y,Z}\}$ and $\mechlevelset(\syn{W}) = \{\syn{W}\}$ so $\mechlevelset(\syn{X})$ and $\mechlevelset(\syn{Y})$ are not disjoint. If for the structure we use \rl{a free Markov category $\strucS=\FreeMarkov(\SigS)$ (or free cd-category)}\ we see that no \strucdown{} is possible, since one finds that \rl{$\HtoLS(\netdiagout{\modelMH}{X,Y}) \neq  \netdiagout{\modelML}{X,Y}$}. Indeed, as diagrams we have:
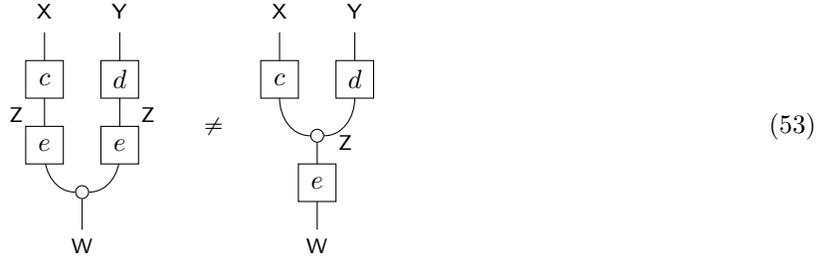
\begin{equation} \label{eq:equality-CA-ex}
% \pi_S(\diagD^\syn{X,Y}_{\modelMH})  \ \  = \ \  
\input{CA-ex-1-inequality.tikz}
 % \ \ = \ \  
% \diagD^{X,Y}_{\modelML}
\end{equation}
Nonetheless, this abstraction is \strong{}, and if for the structure we instead take a free Cartesian category $\strucS=\FreeCart(\SigS)$, \rl{that is requiring each box be}\ deterministic, \eqref{eq:equality-CA-ex} becomes an equality and we do obtain a \strucdown.

Note that one may alternatively define a \CCA, again with $\tau = \id{}$, by setting $\syn{\HtoL :: X \mapsto \{X\}}$, $\syn{Y \mapsto \{Y\}}$, $\syn{W \mapsto \{W,Z\}}$. Now \rl{$\mechlevelset(\syn{N}) = \HtoL(\syn{N})$ for each of $\syn{N}=\syn{X},\syn{Y},\syn{Z}$}, making the abstraction \extrastrong{}, and this defines \rl{a \strongCCA}\
even when using free Markov categories. 
\end{example}

\begin{example}
Let \rl{$\modelML$ and $\modelMH$}\ be causal models defined by the following \rl{respective network diagrams in $\catC$:}
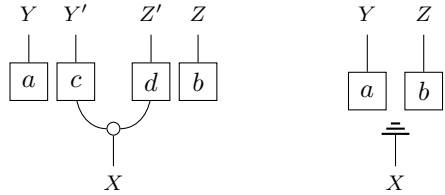
\[
%\tikzfig{CA-ex-2-nds} 
\input{CA-ex-2-nds_b.tikz} 
\]
There is a constructive abstraction between $\modelML$ and $\modelMH$ given by $\syn{\pi :: X \mapsto \{X\}, Y \mapsto  \{Y,Y'\}, Z \mapsto \{Z,Z'\}}$ and \rl{($\LtoH_X=id_X$)}: 
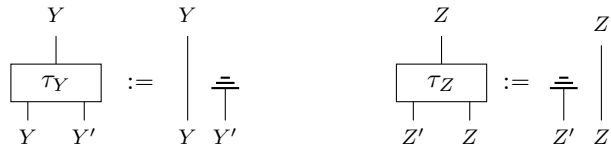
\[
\input{CA-ex-2-taus.tikz}
\]
However this is not \strong{} since $\syn{X \in \pi(X) \cap \mechlevelset(Y)}$.\footnote{This means condition \eqref{eq:cond-1} fails, so indeed $\HtoLS$ is not definable.} This abstraction does not extend to a \strucdown{}, and indeed we cannot even define $\HtoLS$ as in \eqref{eq:strong-CA} since $\HtoLS(\syn{a})$ would be a state of $\syn{Y} \otimes \syn{Y'}$ in $\strucSL$, but no such state exists.
\end{example}

\begin{remark}
	Note that any \CCA\ indeed implies a condition for each high-level mechanism close to Eq.~\eqref{eq:strong-CA-nat}. For any $\syn{X} \in \VninH$, set $\syn{S}:=\Pa(X) \setminus \VinH$. Then in $\catC$ we have the right-hand equality below:  
	\begin{equation}
		% \tikzfig{remark-on-ML-CCA} 
		\input{remark-on-ML-CCA-v2.tikz} 
	\end{equation}
where $R = \VinL \setminus \HtoL(\Pa(X))$. We would like to then define a normalised network diagram $\syn{h_X}$, as on the LHS, as our candidate for $\HtoLS(\syn{c_X})$. For such a normalised diagram to exist, with inputs $\HtoL(\Pa(X))$ as required for $\HtoLS$ to form a functor, would mean that every path from an input in $\VinL$ to $\HtoL(X)$ passes through $\HtoL(\Pa(X))$, i.e.~that $\mechlevelset(X) \cap \VinL = \emptyset$. 
This is thus an equivalent condition to defining the functor $\HtoLS$ in the first place. To establish that $\HtoLS$ satisfies the conditions of a \strongCCA~in each case requires precisely the stronger conditions in Definition \ref{def:struc-well-behavedness}; for more details see the proofs in Appendix \ref{app:strongCA-proof} 
\end{remark}

\section{Abstraction for quantum models} \label{sec:quantum-abstraction}

A benefit of our categorical definition of abstraction is that it can be applied not merely to causal models, but any compositional model whatsoever. Another major class of models are those based on quantum processes. In this section we present new notions of abstraction between \rl{quantum and classical models}, touching on potential applications in explainable quantum AI. To do so, we will work \rl{in the following example of a terminal d-category that can capture 
both finite-dimensional quantum processes (as in quantum computation) and finite classical probabilistic processes (in particular classical computation) \cite{coecke2018picturing}}.

\begin{example}[\QC] \label{ex:QC-catgeory}
In the category $\QC$ the objects are pairs $(\hilbH, X)$ where $\hilbH$ is a finite-dimensional complex Hilbert space (a quantum system) and $X$ is a finite set (a classical system), depicted as thick and thin wires respectively, placed side-by-side. A morphism $f \colon (\hilbH, X) \to (\hilbK, Y)$, drawn as: 
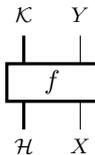
\[
\input{q-mor.tikz}
\]
is given by a \emph{controlled quantum instrument}, i.e.~a collection of linear maps \rl{$\left(f(y \mid x) \colon L(\hilbH) \to L(\hilbK)\right)_{x \in X, y \in Y}$ that}\ are each \emph{completely positive} and jointly trace-preserving:  
\[
\sum_{y \in Y} \Tr(f(y \mid x)(a)) = \Tr(a)
\]
for all \rl{$x \in X$ and}\ operators $a \in L(\hilbH)$. Here $L(\hilbH)$ is the set of linear operators on $\hilbH$.\footnote{Recall that a linear map $f \colon L(\hilbH) \to L(\hilbK)$ is positive when it sends positive operators $a=b^\dagger b \in L(\hilbH)$ to positive operators $f(a)$ in $L(\hilbK)$, and is completely positive when each map $f \otimes \id{\hilbH'}$ is also positive, for any other Hilbert space $\hilbH'$.} Given any $g \colon (\hilbK, Y) \to (\hilbW, Z)$ we define $g \circ f \colon (\hilbH, X) \to (\hilbW, Z)$ \rl{via}\ summation: 
\[
(g \circ f)(z \mid x) := \sum_{y \in Y} g(z \mid y) \circ f(y \mid x)
\]
We define $(\hilbH, X) \otimes (\hilbK, Z) := (\hilbH \otimes \hilbK, X \otimes Y)$ where $\hilbH \otimes \hilbK$ is the tensor product of Hilbert spaces, and on morphisms define $(f \otimes g)(y, z \mid x, w) := f(y \mid x) \otimes g(z \mid w)$, with unit object $I=(\mathbb{C},\{\star\})$. This category has discarding \rl{$\discard{} \colon (\hilbH, X) \to (\mathbb{C}, \{\star\})$, where $\discard{}(\star \mid x) \colon L(\hilbH) \to \mathbb{C}$ is defined by $a \mapsto \Tr(a)$}\ for all $a \in L(\hilbH)$. By definition every morphism is a channel due to trace-preservation. 

In this category \rl{the `fully quantum'}\ systems $\hilbH = (\hilbH, \{\star\})$, correspond to finite dimensional Hilbert spaces (thick wires), such as qubits $\hilbH=\mathbb{C}^2$, and `fully classical' systems $X=(\mathbb{C},X)$, corresponding to finite sets (thin wires). Quantum channels, i.e.~those between quantum systems $f \colon \hilbH \to \hilbK$, are precisely \rl{completely positive and trace-preserving (CPTP) maps}, with $\discard{}$ corresponding to the partial trace. These include quantum states, and quantum unitary circuits. Classical channels $c \colon X \to Y$ between finite sets are precisely those of $\FStoch$, giving an embedding $\FStoch \hookrightarrow \QC$, which allows us to describe classical finite causal models within $\QC$. Some processes of interest between both kinds of systems in $\QC$ are shown in Figure \ref{Fig:QC}. 
\begin{figure}
\centering
\[
\input{Hdisc.tikz}
\qquad
%\tikzfig{unitary}
\rl{\input{unitary_b.tikz}}
\qquad
%\tikzfig{controlled-unitary}
\rl{\input{controlled-unitary_b.tikz}}
\qquad
\input{CPTP.tikz}
\qquad
\input{DM.tikz}
\qquad
\input{purezero.tikz}
\qquad
\input{meas.tikz}
\qquad
\input{encode.tikz}
\qquad
\input{channel.tikz}
\qquad
\input{omegadist.tikz}
\]
\caption{\rl{Morphisms in $\QC$, in order: partial trace on $\hilbH$; unitary gate, i.e.~CPTP map $U(-)U^\dagger$ for unitary $U$, on multiple \rlc{quantum systems}; controlled unitary $U$ with classical control  $X$; CPTP map $f \colon L(\hilbH) \to L(\hilbK)$; density matrix $\rho \in L(\hilbH)$; pure state in computational basis $\ket{0}\bra{0}$;  measurement $\meas$ on $\hilbH$ with finite outcome set $X$; encoder $\enc$ from finite set $X$ into $\hilbH$; probability channel $c \colon X \to Y$; probability distribution $\omega$ over $X$.}}
 \label{Fig:QC}
\end{figure} 
\end{example}

By a \emph{quantum model} we will mean a compositional model $\modelM$ with semantics in $\QC$ making use of at least some quantum systems $\hilbH$, as well as (potentially) classical ones. By a \emph{classical model} we mean a model $\modelM$ in $\QC$ only using classical systems, such as any (finite) causal model. It now makes sense to explore abstraction relations where the models on each side are each of either kind.

\subsection{Quantum to quantum abstraction}

\rl{Let us first consider abstraction relations $\modelML \to \modelMH$ where \emph{both} $\modelML$ and $\modelMH$}\ are quantum. 
\rlc{A ubiquitous instance}\ is a strict \strucdown{}, capturing how a quantum model $\modelMH$ \rl{(say, at a computational or logical level)}\ is implemented in terms of a lower-level quantum model $\modelML$ of specific quantum gates \rl{(say at a quantum hardware level)}, as in the following example. 
\rlc{Although this example is, in a sense trivial, it shows that abstraction is not alien to quantum theory and quantum computing; we leave the study of further, more interesting quantum to quantum cases of abstraction for future work.}

\begin{example}[Quantum model implementation] \label{ex:quantum-model-impl}
By a \emph{\qcircmodel{}} $\modelM$ from classical inputs $X$ to classical outputs $Y$ we mean a model in $\QC$ given as a \rl{sequential}\ composite of an encoder $E$ from $X$ to a quantum system $\hilbH$ 
followed by a unitary circuit $U$ on $\hilbH$, and then a final measurement $\measbig$ with output $Y$. In practice $\hilbH$ is a tensor of some number of qubits, and $\modelM$ is implemented by a number of lower-level processes acting on these, captured by a strict \strucdown{} \rl{$\strucdownabsshort{\modelML}{\modelM}$ with the}\ lower-level model $\modelML$ including the individual qubits. An example where $\modelML$ acts on four qubits, and $X$ factors as four inputs and $Y$ as three outputs is shown below, using processes defined in Figure \ref{Fig:QC}.   
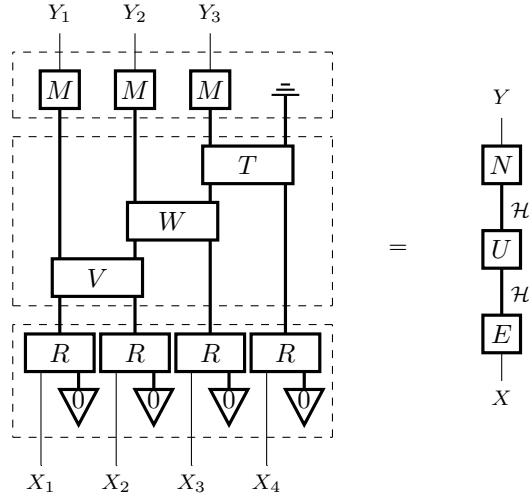
\[
% \tikzfig{refine-PQC}
%\tikzfig{refine-PQC-two}
\input{refine-PQC-two_b.tikz}
\]
The encoder $E$ is implemented via controlled rotation gates $R$ acting on initial states $\ket{0}$, the unitary $U$ via \rl{(typically)}\ entangling gates $V,W,T$, and the measurement $\measbig$ by individual qubit measurements on three \rl{of the qubits}. 
\end{example}

\subsection{Quantum to classical abstraction} \label{sec:q-to-class-reduction}

An interesting new possibility to explore are abstractions $\modelML \to \modelMH$ from a \emph{quantum} model $\modelML$ to a high-level \rl{\emph{classical} model}\ $\modelMH$. There are several motivations to consider such abstractions:
\begin{enumerate}

	\item \emph{Classical reduction.} One simple motivation would be to explain precisely how the quantum model $\modelML$ is `really' classical, being fully accounted for by the classical model $\modelMH$. This would be most apparent \rl{when}\ the queries on $\modelML$ are sufficiently `rich', \rl{or the abstraction holds}\ even at the component-level, capturing all the ways we intend to use the quantum model. One reason for such a reduction would be to argue against a claim that a quantum semantics is necessary to achieve certain behaviours of a model (e.g. computationally or in settings such as quantum cognition \cite{QCog}); for an instance see Example \ref{ex:quant-dec-model} in Appendix \ref{app:q-abs}. 
	\item \emph{Causal decoherence.} Another related motivation from quantum foundations would be to describe a form of decoherence in which a `causally consistent' classical description of a physical system emerges from an underlying quantum one. Crucially this would require $\modelML$ to be a quantum \emph{causal model}, which are not straightforwardly defined as compositional models; we discuss this future direction further in Section \ref{sec:discussion}. 
	\item \emph{Explanation.} A final motivation, pertinent to the intersection of AI and quantum computing, is to consider how a classical causal model $\modelMH$ may provide an `explanation' for (part of) the behaviour of a quantum model $\modelML$. Indeed, within the field of explainable AI (XAI) it has been argued that \rl{ideal `explanations'}\ for a low-level neural network $\modelML$ are best described as causal abstractions $\modelML$ to a high-level (interpretable) causal model $\modelMH$ \rl{\cite{geiger2023causal, geiger-etal-2020-neural, geiger2024finding}}. The abstraction captures structurally how $\modelMH$ reflects the behaviour of the network, \rl{that is the latter has to produce outcomes in keeping with the causal structure of $\modelMH$, which thereby provides causal explanations}, often argued to be the `richest' form of explanations possible \rlb{(see also Sec.~\ref{sec:intro})}. 

In just the same way, we may hope to explain a paramaterised quantum circuit (PQC) or more general quantum model $\modelML$ via an abstraction to a classical causal model $\modelMH$. This may offer an approach to explainable quantum AI (XQAI) in which AI models given by PQCs are explained via such abstractions to interpretable classical causal models. 
 Note that a difference from the fully classical XAI case is that, while \emph{any} neural network can be formally seen as a causal model, this is not the case for \rl{quantum models (see discussion in Section \ref{sec:discussion}.)  
Nonetheless, this may not necessarily be a barrier --  
what brings about the explanation (i.e. the explanans) has two aspects: the high-level model $\modelMH$, which \emph{is} assumed to be causal, and the link between the two models.}   
\end{enumerate}

In the remainder of the section we will focus on the \rlb{above third and final}\ motivation. In any case, the meaning and richness of an abstraction will depend precisely on the set of low-level queries  that are being associated to the high-level ones. Hence we must ask: \emph{what meaningful ways can one query a quantum circuit} 
and moreover such that it can yield meaningful abstractions to a classical (ideally causal) model?\   
\rlb{Here we will propose a partial answer to this question, by generalising causal models to a broader class of compositional models allowing for a quantum semantics, and by suitably generalising two of our earlier classes of queries, namely abstract do-queries and interchange queries.}

\paragraph{Models of DAGs, revisited.} 
In order to define this class of compositional models, 
we first specify a new more general recipe for defining a string diagram from a DAG.  
Let $\dagG$ be an open DAG with vertices $\varV$ and let $\syn{\Vout} \subseteq \syn{\varV}$ be a further subset of vertices called the \emph{outputs}. 

We define the diagram \rl{$\DD{G,\syn{\Vout}}$}\ to have a wire $\syn{X}$ for every vertex $\syn{X} \in \varV$, a wire $\syn{\edgevar{X}{Y}}$ for each edge $\syn{X} \to \syn{Y}$ in $\dagG$, and a further wire labelled $\syn{\edgeout{X}}$ for each output $\syn{X}$. Then for   
each non-input $\syn{X}$ we draw a composite of boxes $\smech_X \circ \mmech_X$, and for each input $\syn{I} \in \syn{\Vin}$ just a box $\smech_I$, as below, and \rl{$\DD{G,\syn{\Vout}}$}~is given by connecting these boxes along their matching wires. 
\[
\input{GV.tikz} 
\quad 
\mapsto 
\quad 
\input{split-mechV.tikz}
% \ \ \qquad  \qquad 
% \tikzfig{genmechV} \qquad = \qquad 
\qquad \qquad \qquad \qquad \qquad 
\input{input-dag-slice.tikz} 
\quad
\mapsto
\quad
\input{inputsplit.tikz}
\]
Here \rl{$\syn{\ePaX} := (\syn{\edgevar{Y}{X}})_{\syn{Y \in \Pa(X)}}$ and $\syn{\eChX} := (\syn{\edgevar{X}{Y}})_{\syn{Y \in \Ch(X)}}$, or if $\syn{X} \in \syn{\Vout}$ then $\syn{\eChX}$ additionally contains $\syn{\edgeout{X}}$. 
Hence $\DD{G,\syn{\Vout}}$}\ is a diagram from $\syn{\Vin}$ to outputs $\DDGout = \syn{(\edgeout{X})_{X \in \Vout}}$. We follow a convention that whenever an output $\syn{X} \in \syn{\Vout}$ has no children in $\dagG$ we omit to draw $\syn{\smech_X}$ and identify $\syn{X = \edgeout{X}}$: \footnote{Formally we assume a model with $X=\edgeout{X}$ and $\smech_X = \id{}$.} 
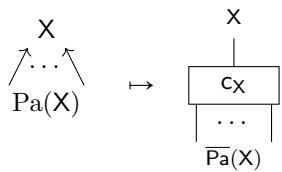
\[
% \qquad \qquad \qquad 
%\tikzfig{output-dag-slice}
\input{output-dag-slice_b.tikz}
\quad
\mapsto
\quad
%\tikzfig{output-mech}
\input{output-mech_b.tikz}
\]
We will see examples of $\DD{G,\syn{\Vout}}$ in Examples \ref{ex:PQD-DDG-ex-1} and \ref{ex:geoguesser} shortly. We can now specify a \rl{general notion}\ of model of a DAG. 

\begin{definition} \label{Def:comp-model-more-general}
Given an open DAG $\dagG$ and subset of output vertices $\syn{\Vout}$, by a \emph{compositional model of $\dagG$ with outputs $\syn{\Vout}$} \rlb{in $\catC$}\ we mean a compositional model $\modelM$ of the signature \rl{$\sigDD{G,\syn{\Vout}}$}\ \rlb{given by the wire labels and boxes of 
$\DD{G,\syn{\Vout}}$}, which futhermore sends every box to a channel in $\catC$. 
\end{definition}

For brevity, for each non-input $\syn{X}$ we denote $\genmech_X := \smech_X \circ \mmech_X$ in $\catC$. As before we denote the overall channel from inputs to outputs by:  
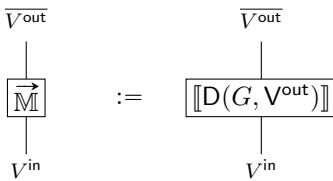
\[
% \tikzfig{Mioreduxsimple}
\input{Mioreduxsimple-nodots.tikz}
% \tikzfig{Mioredux}
\]

\begin{example}[\rl{Classical}\ Causal Models]
An (open) causal model of $\dagG$ in a Markov category $\catC$ \rl{with outputs $\syn{\Vout}$ (see Def.~\ref{def:causal-model})}\ is equivalent to a compositional model of $\dagG$ in the above sense such that $\edgevar{X}{Y} = X$ for all \rl{$\syn{X, Y} \in \syn{V_G}$ with $\syn{X} \in \Pa(\syn{Y})$, $\edgeout{X} = X$ for all $\syn{X} \in \syn{\Vout}$, and the following holds:} 
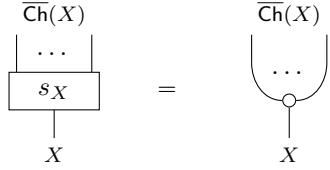
\[
\input{sxiscopy.tikz}
\]
Making \rl{the corresponding}\ substitutions to $\DD{G,\syn{\Vout}}$ directly then yields the network diagram $\diagD_G$.\footnote{While for a (classical) causal model, a choice of outputs $\syn{\Vout }$ is `innocent' in that any change amounts to just `copying out' more or fewer variables without altering the causal mechanisms, a model in the sense of Definition \ref{Def:comp-model-more-general} may be defined in the absence of copy maps (as for quantum theory), making the choice of outputs $\syn{\Vout}$ a non-trivial aspect.}
%a non-trivial aspect of the model.}  
\end{example}

\begin{example}[Quantum Circuits] \label{ex:PQC-as-DDG} 

Consider a \qcircmodel{} $\model{Q}$ in the sense of Example~\ref{ex:quantum-model-impl}, with classical inputs $X$ and classical outputs $Y$, which furthermore is already decomposed into a composite of boxes, where each box is either an encoder $E$ (i.e.~classical input to quantum), a unitary $U$ (quantum to quantum), or a measurement $\meas$ (quantum to classical output). 

We can view $\model{Q}$ as a compositional model of an open DAG $\dagG_{\model{Q}}$ defined as follows.   
We specify an input vertex $\syn{I_j}$ for the input of each encoder $E_j$, an output vertex $\syn{O_j}$ for the output of each measurement $\meas_j$ and a vertex $\syn{X_j}$ for each unitary $U_j$, with $\syn{\ePa{X_j}}$ and $\syn{\eCh{X_j}}$ given respectively by the inputs and outputs of  $U_j$. Thus there is an edge in $\dagG_{\model{Q}}$ for each internal wire of the circuit. 
Then $\model{Q}$ is a compositional model of $\dagG_{\model{Q}}$ with outputs $\syn{O}$, via \rlb{$X_j := \bigotimes_{Y \in \eCh{X_j}} Y$}, $\smech_{X_j} = \id{}$ and $U_j = \mmech_{X_j} = \genmech_{X_j}$.\footnote{While we could simply view $U_j$ as the channel $\genmech_{X_j}$, formally introducing $X_j$ makes explicit the manner in which we are emphasising the outputs of $U_j$ as seen as a single variable.} 
\ 
\[
\input{PQC-as-DDG-2.tikz}
\qquad \qquad \qquad 
\input{PQC-as-DDG-3.tikz}
\qquad \qquad
\qquad
\input{PQC-as-DDG-short.tikz}
\]
\end{example}

\begin{example}[\rl{Subclass of quantum causal models}] \label{ex:QCM-as-comp-model-over-DAG} 
Every quantum causal model with DAG $G$ in the specific sense of \cite{costa2016quantum} (provided some natural dimensional assumption at each `lab') is equivalent to the data of a compositional model over $G$ in above sense, via an identification similar to the treatment of unitary gates in Ex.~\ref{ex:PQC-as-DDG}. These are a special case of the \rlb{more}\ general notion of quantum causal model from \cite{barrett2019quantum, allen2017quantum} which in generality are not (yet) defined as compositional models; see the discussion in Section \ref{sec:discussion}.  
\end{example}

Note that while one can freely turn any classical causal model into a new one with an arbitrary new subset of outputs (using the same mechanisms, simply adding copy maps to the network diagram for the former for any new outputs), this is no longer the case for compositional models of open DAGs as in Def.~\ref{Def:comp-model-more-general}, with the choice of outputs non-trivial. Essentially this is due to no longer assuming the presence of copy maps, in the quantum case due to the famous  `no-cloning' theorem. 

\rlb{While classical causal models and the above special subclass of quantum causal models are special cases of these DAG-based models, the latter are strictly more general -- its instances, classical or quantum, are in general \emph{not} causal models. Yet, they allow for well-defined generalisations of some \emph{causal} queries in interesting ways.}

\paragraph{Opening queries.}
We can \rl{now finally}\ generalise abstract Do-queries as follows. 
\rl{Given a compositional model $\modelM$ of DAG $G$, for}\ any subset of non-input vertices $\syn{S}$ we define a query $\Mopen(\syn{S})$, with representation depicted graphically again as left-hand below. The diagram for $\Mopen(\syn{S})$ \rl{is given by altering $\DD{G,\syn{\Vout}}$ by deleting for all $\syn{X \in S}$ the box $\syn{\mmech_X}$ as well as all wires \rlb{$\syn{\edgevar{Y}{X}}$}\ for $\syn{Y} \in  \Pa(\syn{X})$, now setting $X$ as an input to the diagram (right-hand below)}.\footnote{As for abstract Do-queries we can view opening as a transformation to a new model. For any non-inputs $\syn{S}$ define the open DAG $\dagG_{\Mopen(S)}$ by labelling $\syn{S}$ as inputs and deleting all their incoming edges. Any model $\modelM$ of \rl{$\DD{G,\syn{\Vout}}$}\ induces a model $\modelM' = \modelM_{\Mopen(S)}$ \rl{of $\SigS_{ \DD{\dagG_{\Mopen(S)},\syn{\Vout}}}$}\ with $\semM{\Mopen(S)} = \io{\modelM_{\Mopen(S)}}$, by defining $\sem{\smech_Y}_{\modelM'}$ as the marginal of $\semM{\smech_Y}$ on $\syn{\eChY} \setminus \edgevar{Y}{X}$, discarding its $\edgevar{X}{Y}$ output.} 
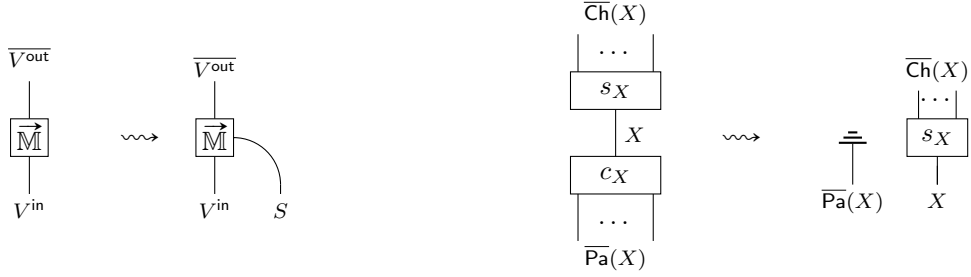
\[
% \tikzfig{Miosimple} 
\input{Miosimple-nodots.tikz} 
\quad \ \ \rightsquigarrow \ \ 
\input{openquery-nodots.tikz}
% \tikzfig{openquery}
\qquad 
\qquad 
\qquad
\qquad \qquad 
%\tikzfig{GV} 
\rl{\input{split-mechV-sem.tikz}}
\quad \ \  
\rightsquigarrow
\quad  \ \ 
\input{open-surgery-2.tikz}
\]
\paragraph{Interchange queries.}

We can also define interchange queries with respect to $\DD{G,\syn{\Vout}}$ as follows. Call a subset of vertices $\syn{W} \subseteq \syn{V}$ \emph{parallelisable} if there is no path $\syn{X} \to \syn{Y}$ in $\dagG$ for any $\syn{X, Y \in W}$ with $\syn{X} \neq \syn{Y}$. In this case we can factor the diagram $\DD{G,\syn{\Vout}}$ as on the left-hand below for some diagrams $\syn{f}, \syn{g}$, and further wires $\syn{W'}$.\footnote{\rl{Note that for each $\syn{X} \in \syn{W}$ the box $\mmech_{X}$ is then contained in $\syn{f}$ and  $\smech_{X}$ in $\syn{g}$. One may then prove that the marginal of $f$ used to define $\modelMio^{W}$ is indeed independent of the (non-unique) chosen $\syn{f}$, $\syn{g}$ and subset $\syn{W'}$.}} We next define a channel \rl{$\modelMio^{W}$}\ as right-hand below.
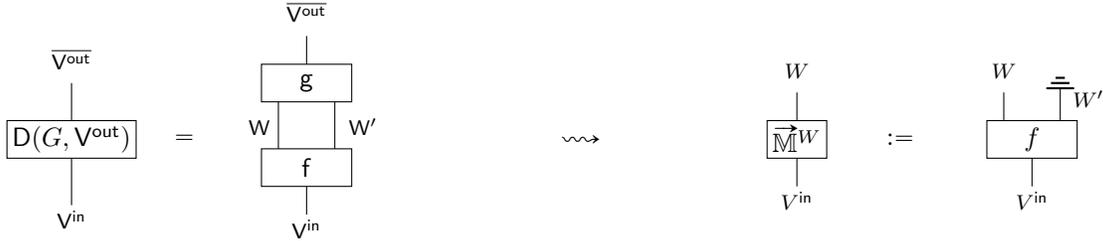
\[
% \tikzfig{DDGfactor}
\input{DDGfactor-nodots.tikz}
\qquad \qquad  \qquad 
%\implies 
\rightsquigarrow
 \qquad \qquad  \qquad 
%\tikzfig{DDGdef}
% \tikzfig{DDGdef_b}
\input{DDGdef_b-nodots.tikz}
\]
\rl{Finally, for any collection of disjoint subsets $\syn{S_1},\dots, \syn{S_n}$, which are each parallelisable, and for any}\ subset $O \subseteq \DDGout$ of outputs, we then define an \emph{interchange query} constructed just as in \eqref{eq:interchange-query}, i.e.:\footnote{\rl{The need to consider parallelisable sets $W$ and corresponding channels $\modelMio^{W}$ contrasts with interchange queries for (classical) causal models and reflects the non-trivial step in general of `turning variables into outputs' already emphasised earlier.}}
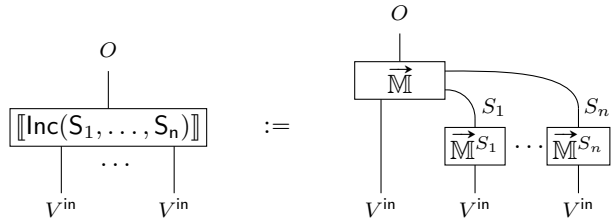
\[
%\tikzfig{weak-open-two}
\input{weak-open-two_b.tikz}
\] 

Opening and interchange queries thus provide two ways to query compositional models of open DAGs, beyond merely input-output behaviour, which may feature in corresponding \qdown{} relations. 

This generalisation is a promising route, for explanatory purposes, to circumvent the subtleties around generalisations of causal models, and we illustrate their use with two examples in which we explain aspects of a given \qcircmodel{} (in the sense of Examples \ref{ex:quantum-model-impl} and \ref{ex:PQC-as-DDG}), via abstraction to a causal model.

\begin{example}[Interpreting wires in a \qcircmodel{}] \label{ex:PQD-DDG-ex-1}
Suppose we wish to interpret a set of quantum wires (e.g.~qubits) within a \qcircmodel{} $\modelML$ with inputs $I_L$ and outputs $O_L$ as corresponding to some high-level classical variable $X$. To do so we first state a classical causal model $\modelMH$ featuring $X$, structured by the simple DAG $\dagG$ shown left-hand below. A general compositional model $\modelM$ of $\dagG$ would be structured by the diagram $\DD{G,\syn{\Vout}}$ center below, while the \rl{(classical)}\ causal model $\modelMH$ specialises to the right-hand form. 

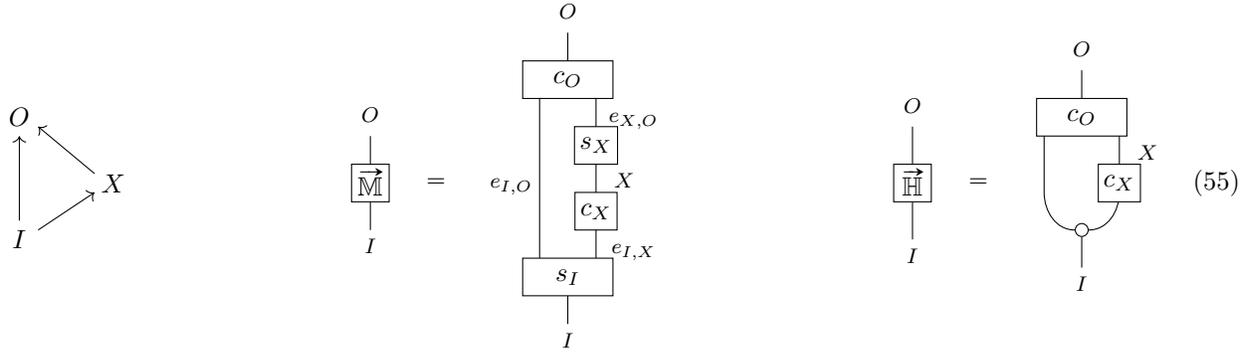
\begin{equation}
\input{ex1DAG.tikz}
\qquad \qquad \qquad \qquad
% \tikzfig{ex1ddgonly} 
\input{ex1ddgM.tikz}
\qquad \qquad \qquad \qquad
 \input{ex1ddgCH.tikz}
\end{equation}
For any such model $\modelM$ over $\dagG$ the queries for opening and interchange at $X$ are respectively as follows. 
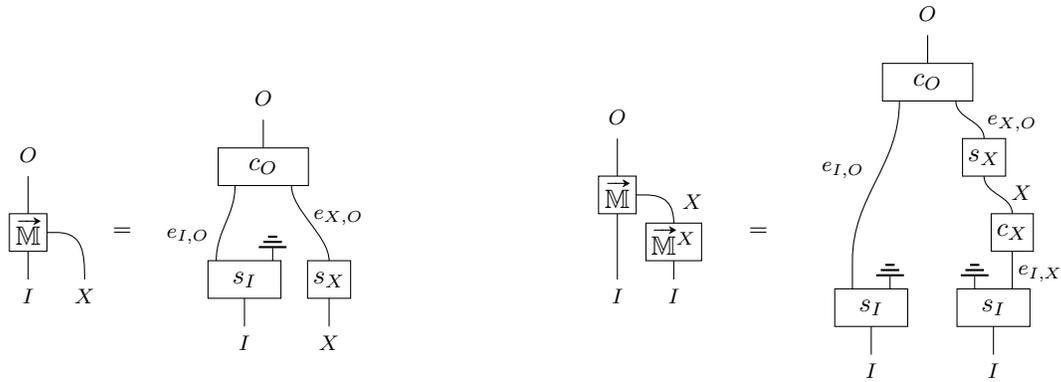
\[
% \tikzfig{modelMiosimpler} \qquad \qquad 
\input{ex1OpenX.tikz}
\qquad \qquad \qquad \qquad 
%\tikzfig{ex1IIX}
\rl{\input{ex1IIX_b.tikz}}
\]
Suppose that the quantum model $\modelML$ is specified by a circuit diagram of encoders, unitaries and measurements. We must first view $\modelML$ as a model of $\dagG$, by grouping together these channels into composites; formally this is a strict \strucdown{} as in Example \ref{ex:quantum-model-impl} (there depicted as dashed boxes). After doing so we can view $\modelML$ as a model of $G$ as above, factorising as below. 
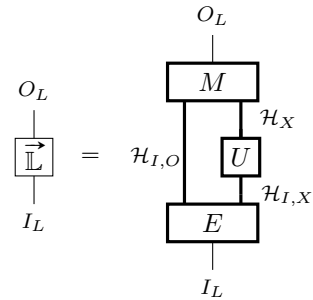
\[
\input{circuit-simple.tikz}
\]

We have now described both $\modelML$ and $\modelMH$ as models of $\dagG$, and hence the same signature, and seek to interpret the outputs of the unitary $U$ in terms of the classical variable $X$ of $\modelMH$. To do so, we must exhibit an abstraction $\modelML \to \modelMH$ for some choice of queries. We now consider three forms of queries and corresponding abstractions of varying strengths.

 Firstly, an abstraction only at the level of input-output queries requires merely classical channels $\LtoH \colon I_L \to I_H$ and $\LtoH \colon O_L \to O_H$ on inputs and outputs, \rl{respectively, satisfying the following condition}\ ensuring the same input-output behaviour.  
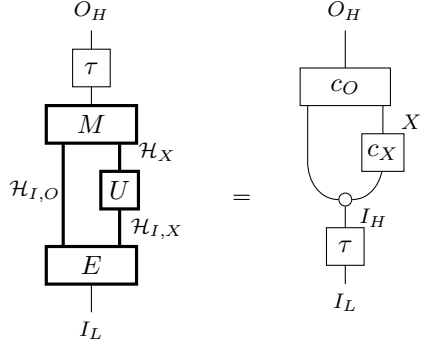
\[
% \tikzfig{Lio} 
% \qquad = \qquad 
% \tikzfig{ex1abstr1}
\input{ex1abstr1simpler.tikz} 
% \qquad = \qquad 
% \tikzfig{Hio} 
\]
To actually refer directly to $X$, we should consider further queries.  
Secondly then, an abstraction based on interchange queries requires a single extra condition, corresponding to interchange at $X$. This states that inserting a second input for which we only use its embedding at $\hilbH_X$ (LHS) yields the same output as if we determined $X$ from that second input (RHS).
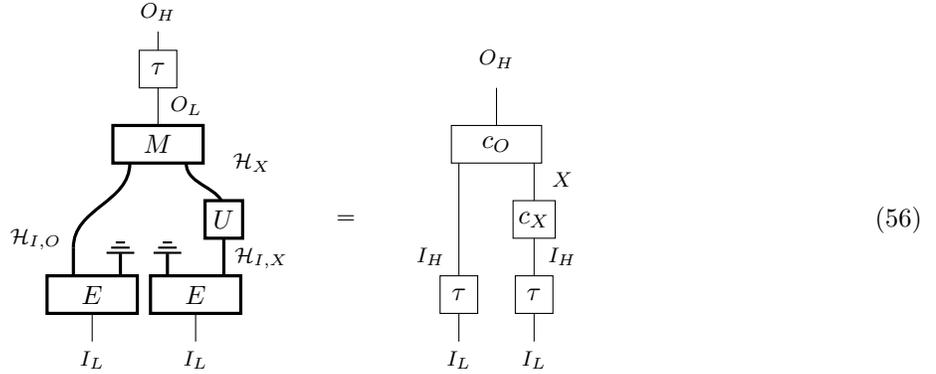
\begin{equation} \label{eq:int-ex-1}
% \tikzfig{ILinterchange}
% \qquad = \qquad
% \tikzfig{ex1abstr3}
\input{ex1abstr3simpler.tikz}
% \qquad = \qquad
% \tikzfig{IHinterchange}
\end{equation}
 Finally, a stronger requirement would be an abstraction based on opening queries. This requires a channel $\LtoH \colon \hilbH_X \to X$ (in practice implemented as a measurement) satisfying the following consistency condition for opening at $X$. This strengthens the interchange condition to one at the level of arbitrary states of $\hilbH_X$ with states of $X$, rather than merely those obtained from inputs via the encoder.  
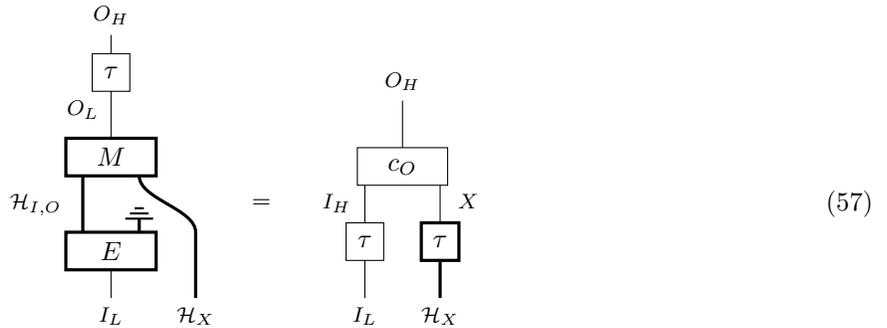
\begin{equation} \label{eq:int-ex-2}
% \tikzfig{MLopen}
% \qquad = \qquad
% \tikzfig{ex1abstr2}
\input{ex1abstr2simpler.tikz}
% \qquad = \qquad
% \tikzfig{Hopen}
\end{equation}
Intuitively, either of \eqref{eq:int-ex-1}, \eqref{eq:int-ex-2} can be used as extra conditions on the quantum model which \rlb{justify an}\ association between the outputs $\hilbH_X$ of $U$ and the variable $X$. 

\end{example}

\begin{example}[Quantum Geoguesser] \label{ex:geoguesser}
Consider a quantum model $\modelML$ which from an input image $I$ returns a guess of its location $O$.  
\rlb{Suppose we wish to check whether}\
the model works by first assessing the weather $W$, the road type in the image $R$, and from these the country $C$ and final output $O$. We can represent this \rlb{hypothesis}\ by the following DAG $\dagG$ and model structure $\DD{G,O}$.  
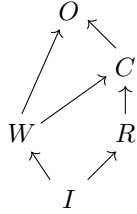
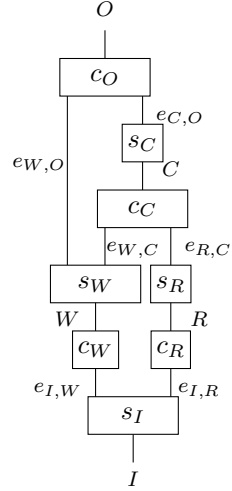
\[
\input{ex2DAG.tikz}
\qquad \qquad \qquad \qquad \qquad 
\input{ex2DDG.tikz}
\]
\rlb{First of all, it must be possible to}\
factorise the gates of the quantum model $\modelML$ into the form of a model of $\dagG$, \rlb{similarly to}\ Example \ref{ex:PQC-as-DDG}, and exhibit a classical causal model $\modelMH$ of $\dagG$ (where in this model we can directly interpret the states of $R$ in terms of roads, $C$ as countries etc). 
 
Now, \rlb{in order to establish whether $\modelML$ behaves in keeping with $\modelMH$, when identifying the output $\hilbH_X$ of $U_X$ with $X$, for each variable $X$ of $\modelMH$, we have to}\ give an abstraction $\modelML \to \modelMH$ for some choice of queries. 
We again consider our three possibilities. Firstly, considering only input-output queries would again mean that the models share input-output behaviour up to classical channels on inputs $I_L \to I_H$ and outputs $O_L \to O_H$,which amounts to the following. 
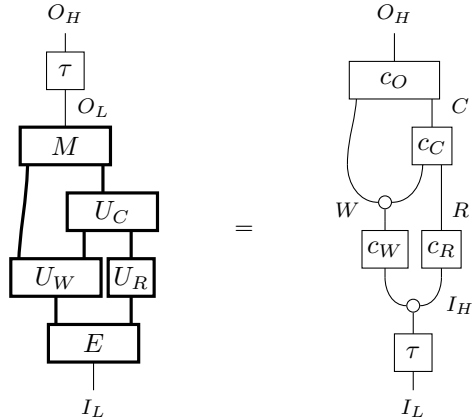
\[
\input{ex2abstr1simplest.tikz}
\]
To actually relate the spaces $\hilbH_R$ with $R$, and so on, we should go further to at least an abstraction based on interchange queries. For this we obtain a condition for each parallelisable subset of vertices. The condition for $\{R\}$ is left-hand below, which is the same as that for $\{W\}$ (up to swapping the left and right inputs on both sides of the equation), and that for $\{C\}$ is to the right. Other subsets are either not parallelisable, or lead to trivial conditions since they span an entire horizontal slice of the \rl{above respective diagrams}. 
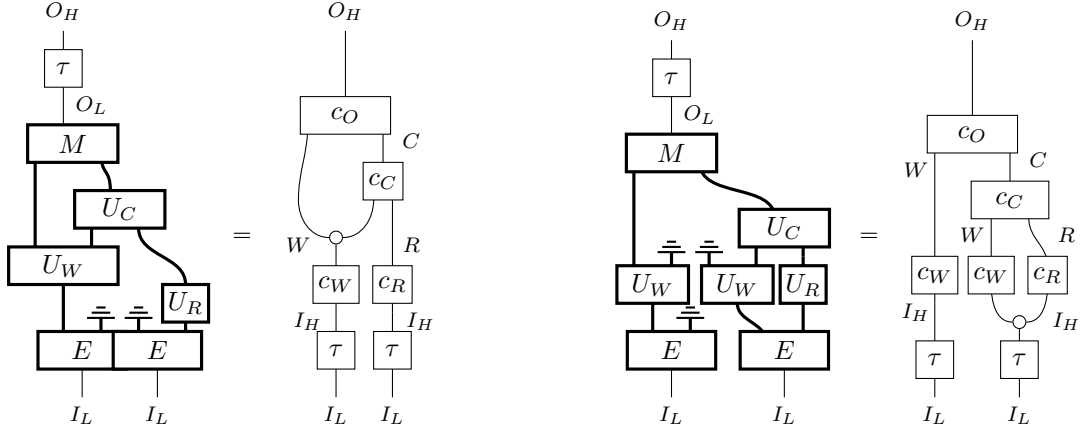
\[
% \tikzfig{ex2abstr2simplest}
\input{ex2abstr2simplestsmall.tikz}
\qquad \qquad \qquad 
% \tikzfig{ex2abstr3simplest}
\input{ex23small.tikz}
\]
Finally, a more strong abstraction based on opening queries would include a channel $
\LtoH \colon \hilbH_X \to X$ for each variable $X$, satisfying a consistency condition for each subset of non-input variables (several of which may well be trivial). For example, the consistency condition for opening at $\{R\}$ would be the following. 
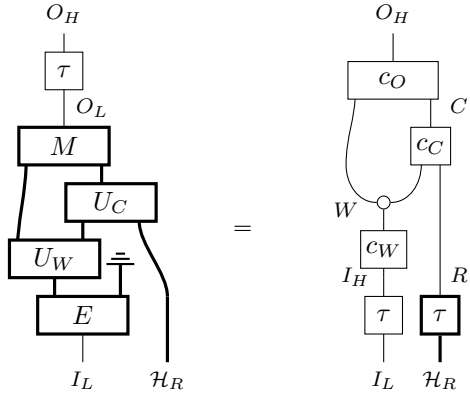
\[
\input{ex3openR.tikz}
\]
\end{example}

\color{black}

\section{Discussion} \label{sec:discussion}

We close with some points of discussion and directions for future work. 

\paragraph{Downward vs upward abstraction.} 
A key finding of this work is that abstraction relations naturally fall into two formally distinct forms, \qdownshort{}s, in which queries are mapped from high to low level, and \qupshort{}s, in which they are mapped in the converse direction, from low to high level. 
The literature on causal abstraction has, to our knowledge, worked entirely in terms of \qupshort{}s, which capture causal abstractions at the level of specific `concrete' (Do-)interventions. 

However, our analysis suggests that it is the notion of \qdownshort{}, which is perhaps the most natural and fundamental. 
Firstly, a \qdownshort{} has a simple formal description, consisting of a functor and natural transformation, which a category theorist might naturally bundle up into a `morphism' between models (see below and App.~\ref{app:further-formalisation}). 
Secondly, most more specific notions of causal abstraction, let alone actual examples, are of the \qdownshort{} type. 
Apart from a generic exact transformation without further assumptions and explicitly distributed causal abstractions like \isoCCA\ and \isointabs, any other noteworthy notion we are aware of can be seen as, in essence, a \qdownshort{}. 
The latter class includes \CCA, \restrictedCCA\ and \CF{} abstraction, where the common presentation of them as \qupshort{}s that make reference to concrete interventions, arises from a down abstraction via Proposition \ref{prop:Exact-trans-from-ref}. 

While the notion of \qupshort{} may be important for certain scenarios, its prominence may simply be due to a focus on concrete interventions and lack of familiarity with open models. Instead \qdownshort{} may prove a more helpful default notion of abstraction in future research.

\paragraph{Mechanism-level causal abstraction.} 
A new contribution from our approach is the notion of a \strucdownfull\ and in particular its instance as a \emph{\strongCCA}, a novel strengthened form of constructive causal abstraction. 
As per our definition this is a \CCA, complemented with consistency conditions for each high-level causal mechanism. 
Thanks to Theorem \ref{Thm:strongCCAnew} we saw that this is in fact equivalent, on the one hand, to a \CCA, where the partitioning $\HtoL$ of low-level variables satisfies a certain purely graph-theoretic condition, and on the other hand, to a functor $\HtoLS \colon \strucSH \to \strucSL$ and a natural transformation respecting the overall network diagrams for the models. 
Noticeably, the latter formulation means that one can define a \strongCCA\ without directly mentioning Do-queries (not even abstract ones in our sense) or any interventions at all. 

While it is a strictly stronger notion than \CCA, it captures special cases of \CCA\ that appear to inform common intuitions, namely cases where picking some partition of the variables of a given (low-level) causal model is used in a relatively straightforward way to \emph{define} a higher-level causal model such that the pair then forms a \CCA. 
Indeed the main examples of \CCA\ from, e.g., \cite{BeckersEtAl_2019_AbstractingCausalModles} are mechanism-level (Ex.~\ref{ex:CCA} being one of them). 
In future it would thus be desirable to understand this novel notion better from a practical perspective -- are most practical examples in fact of this strong, mechanism-level kind, and is it something we should expect of any `good' instance of constructive abstraction?

\paragraph{Iso-constructive and strong causal abstraction.}

Inspired by work in \cite{geiger2024finding, geiger2023causal} we defined the notion of \rlb{an}\ \emph{\isoCCA} - the composition of an `isomorphism-induced' abstraction with a \CCA. 
These form a class of exact transformations based on general (non-Do-)\allowbreak{interventions}, featuring distributed low-level `representations' of high-level variables, with many nice properties yet not requiring a \varalign. Separately we also outlined some natural conditions one would arguably like any practical causal abstraction to satisfy under the name of \emph{\strongCA}. \rlb{The latter}\ generalises \CCA\ and \rlb{is closely related to notions in the literature, but not}\ identical, above all due to allowing general interventions (see Sec.~\ref{sec:lit-comparison}). 

Now, \isoCCA{} also generalises \CCA{} and is indeed an instance of \strongCA. A natural question thus is what other (if any) examples of strong causal abstraction are that are \emph{not} \isoCCA{}s. 
%Finally, we note that understanding conditions that in turn single out

\rlc{We note that understanding conditions that in turn single out the constructive case among \isoCCA{} (such as restriction to only Do-interventions) is related to a conjecture formulated in \cite{BeckersEtAl_2019_AbstractingCausalModles} concerning the equivalence of what are there called strong and constructive $\tau$-abstraction (again, see Sec.~\ref{sec:lit-comparison}). 
Also with a view to a better understanding of the landscape of abstractions with \varalign{s}, it would be interesting to spell out in our setup the significance of violations of the abstract invariance condition from \cite{xia2024neural} and the compositional nature of projected abstractions from \cite{xia2025CAUnderLossyReps}, as well as the notions of graphical consistency, cluster DAGs and the results from \cite{schooltink2024aligning}.}

\paragraph{Abstractions for queries with algebraic structure.} 
While our main definitions of abstraction refer to just sets of queries organised into signatures $\sigQ$, the latter may have further non-trivial algebraic structure. 
We defined what it means for $\sigQ$ to have a \emph{monoid structure} and for the map $\queryLtoH$ of \rlc{an}\ \qup{} to be a \emph{homomorphism} that preserves the algebraic structure of the queries (Sec.~\ref{sec:alg-structure}).  
This allows one to treat the `order-preservation' discussed in the causal model literature as a special case and suggests the stronger homomorphism property to be a perhaps more natural condition to require. 
A further categorical treatment of this extra level of composition on queries would be interesting to explore in future.

{\paragraph{Abstractions for quantum AI.}
In Section \ref{sec:quantum-abstraction} we introduced abstractions between quantum circuits and high-level classical causal models. We speculated that these may provide a principled way to constrain such a quantum model to behave according to, and thus be interpretable in terms of, a high-level classical causal model. Enforcing such constraints from the outset (e.g. in training) may provide a way to obtain explainable quantum AI models, just as causal abstraction may be applied when training neural networks to behave according to a causal model \rlb{(also see Sec.~\ref{sec:intro}).}

If so, an important question is which queries should be taken in such an abstraction, beyond simply matching input-output behaviour. We illustrated consistency conditions corresponding to different queries in Examples \ref {ex:PQD-DDG-ex-1} and \ref{ex:geoguesser}. Interchange queries are perhaps the most practical, since they can simply be tested on classical inputs. However, since they may only be performed for parallelisable subsets they may not be sufficient to fully capture certain high-level structures (DAGs). Another issue with the approach may be that in some cases the constraints imposed by the abstraction may remove necessary coherence in the quantum model, undermining its \rlb{use}\ in the first place. 

In future work it would be interesting to address these questions and more fully assess the potential value of abstraction as a means for assigning interpretable structure to quantum models.

\color{black}
\paragraph{Abstraction with continuous PQCs.}
While in Section \ref{sec:quantum-abstraction} we worked only with finite classical systems, in practice quantum models are typically implemented as PQCs whose inputs $I$ and outputs $O$ are valued \emph{continuously}, in $\mathbb{R}^n$. Moreover, the overall function described by a PQC is usually given by taking the \emph{expectation value} of a measurement (or collection of measurements), obtaining a deterministic map $I \to O$ sending each input to these expectation value(s). In fact one can upgrade $\QC$ to a similarly defined category $\Hyb$, which includes infinite classical systems such as $\mathbb{R}$, and expectation values modelled as an extra graphical ingredient on the classical subcategory of $\Hyb$ \cite{HybPaper}. In future it would be interesting to extend our approach to cover abstractions acting on such PQCs.

\paragraph{Quantum causal abstraction.} 
While in Section \ref{sec:quantum-abstraction} we met abstractions from quantum circuits, seen as compositional models, to high-level classical causal models, to describe a truly \emph{causal} abstraction from quantum to classical would instead require an abstraction from a \emph{quantum causal model} \cite{barrett2019quantum, allen2017quantum}. 
There does exist what is essentially a certain sub-class of quantum causal models due to Costa and Shrapnel which are definable as compositional models (Example \ref{ex:QCM-as-comp-model-over-DAG}) \cite{costa2016quantum, giarmatzi2019quantum}, being essentially defined in terms of compositional rather than causal structure. However, in general quantum causal models are defined instead in terms of `higher-order maps', and so cannot straightforwardly (yet) be seen as compositional models (due to technical open questions \cite{vdLEtAL_2025_UnitaryCausalDecompositions,lorenz2021causal}). Nonetheless notions of quantum causal abstraction would be interesting from the perspective of foundations of quantum theory and philosophy of causation, namely as a possible means to study when decoherence may yield a `causally-consistent emergence of classical reality'.\footnote{\rl{There also are other frameworks of quantum causal models, which however do not appear suitable for studying the suggested foundational question here. As well as those of \cite{barrett2019quantum, allen2017quantum} and the special case in \cite{costa2016quantum, giarmatzi2019quantum}, the notion due to Henson, Lal and Pusey \cite{henson2014theory} is not `fully quantum' seeing as classical variables (such as measurement outcomes) are already baked into the notion  itself; 
further frameworks that share the latter feature include, for instance, \cite{fritz2012beyond, laskey2007quantum, pienaar2015graph,ried2015quantum, ferradini2025cyclic}; see \cite{lorenz2020quantum, allen2017quantum} for overviews of these and further frameworks.}}

The development of genuine quantum causal abstraction is thus left for future work, with Sec.~\ref{sec:quantum-abstraction} and App.~\ref{app:abs-in-pm-via-tracedC} providing possible starting points by first extending the notion of compositional model so as to incorporate higher-order maps.

\paragraph{Cyclic structures and approximate versions.} 
The notion of causal model this work used -- albeit more general in many respects than the ones in much of the literature -- did assume acyclicity. 
There are however reasons to study causal abstraction on the basis of in general cyclic models such as, e.g., the study of equilibration of time-evolving systems from a causal model perspective \cite{rubenstein2017causal}, but also the interpretability of machine learning models in \rlb{the}\ field of XAI has been argued to motivate cyclicity \cite{geiger2023causal}.\footnote{\rl{Note though that, unlike the rest of the literature on cyclic causal models, the notion in \cite{geiger2023causal} is the weakest possible -- essentially any function on the variables, without demanding unique solutions or that it yields a well-defined probability distribution on outputs give some stochasticity on exogenous variables. Allowing for interventions more general than Do-interventions and also for cyclic structures but excluding paradoxical or inconsistent models is a non-trivial task without simple characterisations thus far.}} 
It would be interesting to generalise the categorical treatment of causal models and abstraction to cyclic causal models. 
A natural way to do so is via the parallel-mechanism process $\parmechdiag_\modelM$ induced by a model $\modelM$ together with the `trace-trick' sketched in App.~\ref{app:abs-in-pm-via-tracedC}; a further possible starting point is the framework from \cite{ferradini2025cyclic}. Both of these also naturally link with the study of abstraction for quantum causal models. 

A further direction in which one may try to extend the categorical treatment of abstraction is to allow approximate abstractions in the sense as studied in \cite{beckers2020approximate, geiger2024finding, geiger2023causal}, which is important for practical machine learning contexts. Here the consistency condition for each query is not required to hold on the nose anymore, but instead one demands respective left- and right-hand sides be close to each other in some suitable sense.

\paragraph{Further formalisation.}
There is a further step of formalising the central definitions of this work which is natural to consider from a categorical perspective. This is based on defining a category $\Model(\catC)$ whose objects are models in $\catC$ and morphisms are such that they essentially include both \qdown{}s and \qup{}s. More details and the basic definitions of this approach can be found in App.~\ref{app:further-formalisation}, which would be interesting to fully develop in future. 

It would also be interesting to use our approach to explore connections between (causal) abstractions and categorical notions of abstraction from computer science (such as those between programming languages or types). The latter are often also formalised in terms of functors and natural transformations, which now moreover form adjunctions.

\addcontentsline{toc}{section}{References}
\bibliographystyle{alpha}
\bibliography{causabs}

\appendix

\section{Proofs}\label{app:proofs}

\begin{proof}[\proofof~Remark \ref{rem:var-alignment-ref}]
Any monoidal functor $\HtoL \colon \strucV^*_H \to \strucV^*_L$ sends each  $\syn{V}$ in $\varV_H$ to some list $\HtoL(\syn{V})$ in $\varV_L$, and must preserve identities. Since $\HtoL$ sends queries to queries, $\HtoL(Q_{(\syn{V})}) = Q_{\syn{X}}$ for some subset $\syn{X}$ in $\varV_L$, and then $\HtoL(\syn{V}) = \syn{X}$. But then $\HtoL(Q_{\varV_H}) = Q_{\syn{Y}}$ for some subset $Y$, which must be equal to the concatenation of $\HtoL(\syn{V})$ for ${\syn{V} \in \varV_H}$. Hence the $\HtoL(\syn{V})$ must all be disjoint as $\syn{Y}$ contains no repeats. The properties of $\LtoH$ are those of \qdown s. 
\end{proof}

\begin{proof}[\proofof~Theorem \ref{Thm:CCA-as-refinement}]
Given a constructive abstraction $(\HtoL, \LtoH)$ we define a \qdown{} as follows. We define $ \HtoL$ on types by concatenation and on queries by 
\begin{equation} \label{eq:query-map-CA}
\syn{\HtoL \left(\Do(S) \colon \VinH \otimes S \to X \right) := \left( \Do(\HtoL(S)) \colon \HtoL(\VinH \otimes S) \to \HtoL(X) \right)}
\end{equation} 
using that $\pi(\VinH) = \VinL$. We define $\LtoH$ on arbitrary types by taking products as \rl{in \eqref{eq:tau-factors}}. Taking marginals \rl{of \eqref{eq:CAabstraction}}\ yields all the necessary naturality equations stating precisely that $(\HtoL,\LtoH)$ forms a \qdown. 

Conversely, let $(\HtoL, \LtoH)$ be a \qdown. Since $\HtoL(\VinH) = \VinL$ and $\HtoL$ is a monoidal functor sending queries to queries, it must be as in \eqref{eq:query-map-CA}, as the RHS is the only query of input type $\VinL \otimes \HtoL(\syn{S})$. Now consider the high-level query $\syn{Q}$ with outputs $\VoutH$ intervening on all of $\VintH$. Then $\HtoL(\syn{Q})$ is the low-level query with outputs $\HtoL(\VoutH)$ and intervening on $\HtoL(\VintH)$. Hence the $\HtoL(\syn{X})$ are all subsets, with $\HtoL(\VoutH) \subseteq \VoutL$. Since the low-level inputs $\HtoL(\VinH \cup \VintH)$ have no repeats, the $\HtoL(\syn{X})$ must all be disjoint.  
 Hence $\dvaralignpair{\HtoL}{\LtoH}$ is a \varalign.   
Naturality then ensures that \rl{\eqref{eq:CAabstraction}}\ holds. 
\end{proof}

\begin{proof}[\proofof~Proposition \ref{prop:interchange-queries}]
Given an \rl{\restrictedCCA}, define $\HtoL$ on queries via \eqref{eq:weakCA-query-map} and extend $\HtoL, \LtoH$ to arbitrary types by taking products as usual. All query consistency conditions follow from \eqref{eq:weak-CA-natural} by taking marginals. 

Conversely, let $(\HtoL, \LtoH)$ be such a \qdown{}. Since it must send queries to queries, considering their outputs shows that we must have $\HtoL(\VoutH) \subseteq \VoutL$, and we have $\HtoL(\VinH) = \VinL$ and $\HtoL(\VintH) \subseteq \VintL$ by assumption. Hence the $\HtoL(\syn{X})$ must all be disjoint subsets for $X$ within either $\VinH$ or $\VintH$, and hence all disjoint since $\VinL$ and $\VintL$ are. Hence $\dvaralignpair{\HtoL}{\LtoH}$ form a \varalign. Then \eqref{eq:weak-CA-natural} follows from naturality of $\LtoH$.
\end{proof}

\begin{proof}[\proofof~Proposition \ref{prop:Model-Induction}]

By assumption the map is well-defined, and it remains for us to verify the naturality condition \eqref{eq:ex-trans-induced}. 

It suffices to show that $\modinnin \circ \modelM = \inmod{\modelM} \circ \modinin$, since then for any intervention $\inti$ we can simply substitute $\modelM$ with 
$\transform{\inti}{\modelM}$ to obtain \eqref{eq:ex-trans-induced}. Let $\vin$ be a sharp state of $\Vin$. Let $\vnin = \modelMio \circ \vin$, and define $v = \vin \otimes \vnin$ as a state of $V$. By Lemma \ref{lem:fixpoints} we have $v = \parmechdiag_\modelM \circ v$. Hence $\modin \circ v = \modin \circ \parmechdiag_\modelM \circ v = \parmechdiag_{\inmod{\modelM}} \circ \modin \circ v$.  By Lemma \ref{lem:fixpoints} again the latter holds iff, writing $\modin \circ v = \win \otimes \wnin$ we have $\wnin = \io{\inmod{\modelM}} \circ \win$. Since $\modin$ factors as in \eqref{eq:respects-inputs} we have $\win = \modinin \circ \vin$ and $\wnin = \modinnin \circ \vnin = \modinnin \circ \modelMio \circ \vin$. Hence $\io{\inmod{\modelM}} \circ \modinin \circ \vin = \modinnin \circ \modelMio \circ \vin$. Since $\vin$ was arbitrary \rl{(and $\catC$ is assumed to have enough states)}\ we conclude $\io{\inmod{\modelM}} \circ \modinin = \modinnin \circ \modelMio$ as required. 
\end{proof}

\begin{proof}[\proofof~Equations \eqref{eq:DDo-pic} and \eqref{eq:DII-pic}] 
We first verify \eqref{eq:DDo-pic}. Let $\doendo{p}$ denote the LHS of \eqref{eq:II-diag}, replacing state $s$ by $p$ and $S$ by $S'$. Then by definition: 
\[
\parmechdiag_{\DDo(S'=p,\rho)} := \modin^{-1} \circ 
%\parmechdiag_{\Do(S'=p)(\inmod{\modelL})} 
\rlb{\parmechdiag_{\inmod{\modelL}_{\Do(S'=p)}}}
\circ \modin  = \modin^{-1} \circ \doendo{p} \circ \parmechdiag_{\inmod{\modelL}} \circ \modin = \modin^{-1} \circ \rlb{\doendo{p}} \circ \modin \circ \parmechdiag_\modelL
\]
using \eqref{eq:II-diag} and then \eqref{eq:ex-trans-induced} in the last step. Next, the equation \eqref{eq:DII-pic} is a special case, using the RHS of \eqref{eq:II-diag}, 
once we observe that for any subset of outputs $Y_j$ (in particular for $Y_j = \HtoL(S_j)$) we have: 
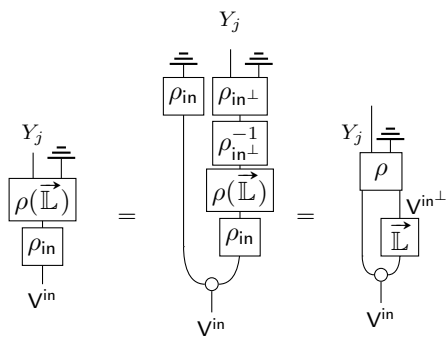
\[
\input{DII-proof.tikz}
\]
in the final step again using \eqref{eq:ex-trans-induced}. 
\end{proof}

\begin{proof}[Proof of Proposition \ref{prop:c-level-is-d-abs}]
Composing $\LtoH \colon \semcompM{\modelL}{-} \circ \HtoLS \Rightarrow \semcompM{\modelH}{-}$ with the identity transformation $1_{\qabsMHfunc}$ on $\qabsMHfunc$ we obtain a transformation 
$\LtoH \circ 1_{\qabsMHfunc} \colon \semMLfunc \circ \HtoL \Rightarrow \semMHfunc$, since $\semcompM{\modelL}{-} \circ \HtoLS \circ \qabsMHfunc = \semcompM{\modelL}{-} \circ \qabsMLfunc \circ \HtoL = \semMLfunc \circ \HtoL$ and $\semcompM{\modelH}{-} \circ \qabsMHfunc = \semMHfunc$.
\end{proof}

\subsection{Mechanism-level causal abstraction} \label{app:strongCA-proof}
We now work towards a proof of our characterisation of mechanism-level causal abstraction. 

We will use the following. Firstly, recall the network diagram $\netdiag{\modelM}$ associated to an open DAG $G=G_\modelM$ for a causal model $\modelM$. We can characterise such diagrams as follows. 

 \begin{definition} \label{def:network-diagram}
A \emph{network diagram} is a string diagram $\diagD$ built from single-output boxes, copy maps and discarding:
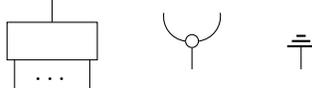
\[
\input{nd-box.tikz}
\qquad 
\input{nd-copy.tikz}
\qquad 
\begin{tikzpicture}[tikzfig]
	\begin{pgfonlayer}{nodelayer}
		\node [style=upground] (0) at (1.5, 0) {};
		\node [style=none] (1) at (1.5, -0.75) {};
	\end{pgfonlayer}
	\begin{pgfonlayer}{edgelayer}
		\draw (1.center) to (0);
	\end{pgfonlayer}
\end{tikzpicture}

\] 
along with labellings on the wires, such that any wires not connected by a sequence of copy maps are given distinct labels, and each label appears as an output at most once and as an input to any given box at most once.
\end{definition}

Such diagrams are equivalent to open DAGs via the correspondence outlined earlier in Section \ref{subsec:causal-models}. We call a network diagram $\diagD$ \emph{normalised} if every non-input variable in the diagram has a path to an output. Up to swappings of input wires, any such diagram takes the form: 
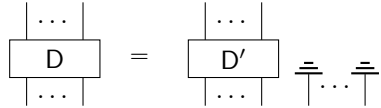
\[
\input{normalised-network-diag.tikz}
\]
where $\diagD'$ is a network diagram which does not feature $\discard{}$. 

\color{black}

\begin{lemma} \label{lem:unique-ND}
Let $\SigS=\SigS_\modelM$ be the signature of an open causal model $\modelM$ with input variables $\syn{\Vin}$ and let $\syn{X}, \syn{Y}$ be lists of variables from $\Sig$, each of which feature no repeated variables.
\begin{enumerate}
\item \label{enum:NND-equivalent-conditions}
A diagram $\diagD \colon \syn{X} \to \syn{Y}$ in $\strucScd = \FreeCD(\Sig)$ is a normalised network diagram iff for every variable $\syn{V}$ such that $\syn{c_V}$ appears in $\diagD$ we have that $\syn{c_V}$ appears only once, $\syn{V}$ is not an input to $\diagD$, and $\syn{V}$ has an upward path to an output in $\diagD$.
\item \label{enum:path-to-diag}
Suppose there exists a path $\syn{V \to Y}$ in $\dagG$ not passing through $\syn{X}$. Then $\syn{c_V}$ belongs to any diagram $\diagD \colon \syn{X} \to \syn{Y}$.
\item \label{enum:LemFreeCD}
There is at most one normalised network diagram $\diagD \colon \syn{X} \to \syn{Y}$ in $\FreeCD(\Sig)$. 
\item \label{enum:LemFreeCart}
If $\syn{X} \subseteq \Vin$ there is at most a unique morphism $\diagD \colon \syn{X} \to \syn{Y}$ in $\FreeCart(\Sig)$.
\end{enumerate}
\end{lemma}

\begin{proof}
\eqref{enum:NND-equivalent-conditions}: By definition any normalised network diagram must satisfy the first two conditions, else the diagram would contain two $\syn{V}$ wires not connected only by copy maps, and the final condition is required for normalisation. Conversely, suppose that they hold. By definition $\syn{X}, \syn{Y}$ have no repeated variables so any input and output appears at most once, and $\diagD$ is built from the necessary pieces. Let $\syn{V}$ be an input to the diagram. Since it is not repeated and cannot arise as an output from any box, all $\syn{V}$ wires must be connected by copy maps. Suppose $\syn{V}$ is not an input and appears in the diagram. Then by the nature of the signature $\syn{V}$ must arise as an output of some $\syn{c_V}$ box, with no other boxes between $\syn{V}$ and $\syn{c_V}$ as none other have output type $\syn{V}$, so this must be reducible to a sequence of copy maps. Moreover by assumption there is a unique such box $\syn{c_V}$ so this covers all instances of $\syn{V}$. Hence $\diagD$ is indeed a network diagram. 

\eqref{enum:path-to-diag}: Suppose there exists such a path $\syn{V \to Y_i}$ where $\syn{Y_i \in Y}$. We prove that $\syn{c_V}$ appears in $\diagD$ by induction over the shortest length $n$ of such a path. 

If $n=0$ then $\synV \in \syn{Y}$, with $\syn{V = Y_i}$. Since $\syn{Y_i}$ is not an input to the diagram, it must arise as an output of some box, and so $\syn{c_{Y_i}}$ must appear in the diagram. Suppose the result is true for $n$. Then if $\synV$ has a shortest path of length $(n + 1)$ it has a child $\syn{W}$ with shortest path of length $n$ and so $\syn{c_W}$ appears in $\diagD$. Then a $\syn{V}$ wire appears in $\diagD$, and so similarly since $\syn{V}$ is not an input to the diagram, $\syn{c_V}$ must appear in $\diagD$.

\eqref{enum:LemFreeCD}:
Let $\diagD \colon \syn{X} \to \syn{Y}$ be such a normalised network diagram. Since $\diagD$ is a network diagram, for each of its inputs $\syn{X_i} \in \syn{X}$ it cannot be that $\syn{c_{X_i}}$ appears in $\diagD$. Hence without loss of generality we can assume each $\syn{X_i}$ is an input to the model itself. 

Let $\dagG=\dagG_\modelM$ be the open DAG corresponding to $\modelM$ (for any choice of outputs). Then $\diagD$ also corresponds to an open DAG $H$ with vertices given by a subset of those of $\dagG$. Moreover, since $\diagD$ is built from $\Sig$, $H$ has edges given by a subset of those of $\dagG$, and if $V \in H \setminus \invar(H)$ then $\Pa^\dagG(V) \subseteq H$ with all edges $\Pa^\dagG(\synV) \to H$ in $\diagD$. Hence $\Pa^H(\synV) = \Pa^G(\synV)$ for each $V \in H$. 

We claim that $H$ is precisely the full sub-DAG of $\dagG$ given by all vertices $\syn{V}$ for which there exists a path $\syn{V} \to \syn{Y}$, plus any further inputs of $\diagD$, making $H$, and hence $\diagD$ unique. Indeed since $\diagD$ is normalised, if $\synV \in H \setminus \invar(H)$ then there is a path $\synV \to \syn{Y}$ in $H$, and hence also such a path in $G$. Conversely if $\synV$ has a path to $\syn{Y}$ in $G$ then $\syn{c_V}$ appears in $\diagD$ by \eqref{enum:path-to-diag} and so $\synV \in H$, as required.

\eqref{enum:LemFreeCart}: 
There is a unique diagram $\diagD \colon \syn{X} \to ( )$ given by $\discard{X}$. Since every morphism in $\strucScart$ is deterministic, any diagram $\diagD \colon \syn{X} \to (\syn{Y_1},\dots,\syn{Y_n})$ takes the form:
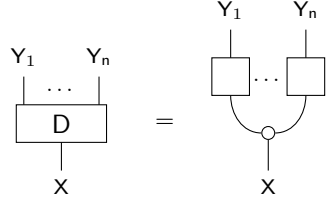
\[
\input{cart-factors.tikz}
\]
Hence it suffices to consider when $n=1$. We proceed by induction over the DAG. If $\syn{Y}$ has no parents then there are no generators with output $\syn{Y}$, and so the unique diagram is given, up to swaps, by $\id{Y} \otimes \discard{}$. Now suppose $\Pa(\syn{Y})$ is non-empty and the result holds for all parents of $\syn{Y}$. Then $\syn{Y}$ is not an input to the model and hence not one to the diagram, and so as an output must arise from a $\syn{c_Y}$ box. Then we can factor the diagram as follows, and by induction this is its unique form. 
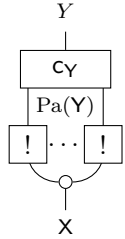
\[
% % \tikzfig{cart-factors} \qquad \qquad \qquad \qquad 
% \tikzfig{discardbuty} \qquad \qquad \qquad \qquad 
\input{pay-factor.tikz}
\]
\end{proof}

\begin{proposition} \label{prop:caus-mech-abs-main}
Let $\HtoLS \colon \strucSH \to \strucSL$ be a cd-functor. Write $\HtoL(\syn{X}) := \HtoLS(\syn{X})$ on variables and types $\syn{X}$. Then $\HtoLS$ is structurally well-behaved iff the subsets $\HtoL(\syn{X})$ are all disjoint with $\HtoLS(\VoutH) \subseteq \VoutL$, $\HtoLS(\VinH) \subseteq \VinL$, and the following square commutes: 
\begin{equation} \label{eq:CCA-square}
\input{CCA-square.tikz}
\end{equation}
when as in \ref{eq:query-map-CA} we define $\syn{\HtoL (\Do(S) \colon \VinH \otimes S \to X) := (\Do(\HtoL(S)) \colon \HtoL(\VinH \otimes S) \to \HtoL(X))}$. Moreover this is the only possible such functor $\HtoL$ mapping queries to queries making the diagram commute. 
\end{proposition}
\begin{proof}
For the final statement, since $\qabsMLfunc$ and $\qabsMHfunc$ are the identity on types (variables), $\HtoL$ and $\HtoLS$ must agree on objects. Then the above is the only well-typed functorial map from queries to queries, so the functor $\HtoL$ is unique if it exists. 

Suppose that the square \eqref{eq:CCA-square} does commute, we claim that each $\HtoLS(\syn{c_X})$ must be a normalised network diagram. Indeed consider the query $\synQ = \openquery{\syn{X}}{\VinH}{\syn{Y}}$ where $\syn{Y := \Pa(X) \setminus \VinH}$. By definition this has $\qabsMH{\synQ} = \syn{c_x \otimes \discard{\VinH \setminus \Pa(X)}}$. But then $\qabsML{\HtoL(\synQ)} = \HtoLS(\qabsMH{\synQ}) =  \HtoLS(\syn{c_X}) \otimes \discard{\HtoL(\Vin \setminus \Pa(X)}$ and so the latter must be a normalised network diagram, and hence $\HtoLS(\syn{c_X})$ must be also, as required. 

Suppose that $\HtoLS$ satisfies the above conditions. By definition structural well-behavedness requires that $(\HtoL(\syn{X}))_{X \in \VH}$ satisfies the subset and disjointness conditions above. By assumption each $\HtoLS(\syn{c_X})$ is a normalised network diagram. Moreover we claim it has non-input variables $\mechlevelset(\syn{X})$. Indeed suppose that $\syn{c_Y}$ appears in the diagram. Then there must be a path from $\syn{Y}$ to the output $\syn{\HtoL(X)}$ within the diagram. Since it is a network diagram it cannot pass through an input to the diagram, and so $\syn{Y} \in \mechlevelset(\syn{X})$. Conversely if $\syn{Y} \in \mechlevelset(\syn{X})$ then $\syn{c_Y}$ appears in $\HtoLS(\syn{c_X})$ by Lemma \ref{lem:unique-ND} \eqref{enum:path-to-diag}.

We now show that for each structure type (Cartesian, Markov, cd), the square \eqref{eq:CCA-square} commutes iff $\HtoL$ is also structurally well-behaved in the corresponding sense. 

\ \ 

\underline{Cartesian case $\FreeCart(\Sig)$:} Suppose that $\HtoL$ is not strong. Let $\syn{V} \in \HtoL(\syn{Y}) \cap  \mechlevelset(\syn{X})$. Then $\syn{Y} \not \in \Pa(\syn{X})$. Let $\syn{S} :=  \Pa(\syn{X}) \cup \syn{Y} \setminus \VinH$. Consider the query $\syn{Q} = \openquery{\syn{X},\syn{Y}}{\VinH}{\syn{S}}$, i.e.~defined as in \eqref{eq:Q-useful} but now with $\syn{Y}$ as both an output and input. Then: 
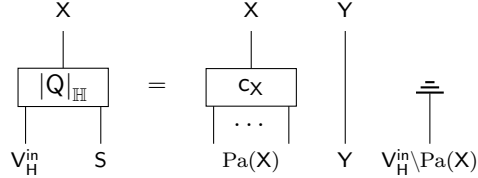
\[
\input{cond-again.tikz}
\]
 Then $\HtoLS(\qabsMH{\syn{Q}})$ contains two disconnected instances of $\syn{V}$, one internal to $\HtoLS(\syn{c_X})$ and one to $\HtoL(\id{\syn{Y}})$. Hence it is not a network diagram.

Conversely, suppose that $\HtoL$ is strong. Consider an arbitrary query $\syn{Q} = \openquery{\syn{X}}{\VinH}{\syn{S}}$. We will show that $\HtoLS(\qabsMH{\syn{Q}})$ does not contain $\syn{c_V}$ for any $\syn{V} \in \HtoL(\syn{S})$. Indeed if $\syn{c_V}$ belongs to this diagram then there is some $\syn{c_Y}$ appearing in $\qabsMH{\syn{Q}}$ with $\syn{c_V}$ belonging to $\HtoLS(\syn{c_Y})$, i.e.~$\syn{V} \in \mechlevelset(\syn{Y})$. But if $\syn{c_Y}$ appears in this way then $\syn{Y} \not \in \HtoL(\syn{S})$. Since then $\mechlevelset(\syn{Y}) \cap \pi(\syn{S}) = \emptyset$ by strength, we have $\syn{V} \not \in \HtoL(\syn{S})$. Thus we have that both $\HtoLS(\qabsMH{\syn{Q}})$ and $\qabsML{\HtoL(\syn{Q})}$ are (network) diagrams of the same type not featuring $\syn{c_V}$ for any $\syn{V} \in \HtoL(\syn{S})$. Thus letting $\modelM' := \openmodel_\syn{S}(\modelML)$ we can see both as diagrams in $\strucS_{\modelM'}$, for which all of their inputs are inputs to the model $\modelM'$. Hence by Lemma \ref{lem:unique-ND} they are equal in $\strucS_{\modelM'}$ and hence equal in $\strucSL$ also. 

\ \ 

\underline{Markov case $\FreeMarkov(\Sig)$:} Suppose that $\HtoL$ is not extra strong. Then there exist non-input variables $\syn{X},\syn{Y} \in \varH$, with $\syn{V} \in \mechlevelset(\syn{X}) \cap \mechlevelset(\syn{Y})$. Then the mechanism $\syn{c_V}$ appears in both $\HtoLS(\syn{c_X})$, $\HtoLS(\syn{c_Y})$ and so for any network diagram $\diagD$ in $\strucSH$ containing both $\syn{c_X},\syn{c_Y}$, $\HtoLS(\diagD)$ contains two instances of $\syn{c_V}$ and is not a network diagram. 

Conversely, assume $\HtoL$ is extra strong. Note that every diagram is automatically normalised in the Markov case. We will show that for any network diagram $\diagD$ in $\strucSH$, $\HtoLS(\diagD)$ is again a network diagram in $\strucSL$. Indeed for any such diagram $\diagD$ in $\strucSH$, each $c_\syn{Y}$ can appear in at most one $\HtoLS(\syn{c_X})$, since the $\mechlevelset(\syn{X})$ are all disjoint as $\HtoL$ is extra strong, and hence at most once in $\HtoLS(\diagD)$. If it does not have a path to an output it is automatically removed from the diagram in the Markov case. Hence by Lemma \ref{lem:unique-ND} \eqref{enum:NND-equivalent-conditions}, $\HtoLS(\diagD)$ is a network diagram. In particular, for each query $\syn{Q} \in \Do(\VH)$, $\HtoLS(\qabsMH{\syn{Q}})$ is a network diagram, and then we are done by Lemma \ref{lem:unique-ND} \eqref{enum:LemFreeCD}.

\underline{Cd-case $\FreeCD(\Sig)$:} Suppose that fullness fails, so that for some $\syn{X}$ we have $\syn{Y} \in \Pa(\syn{X})$ with $\syn{Y}$ not an input, and some $\syn{V} \in \HtoL(\syn{Y})$ with no path $\syn{V} \to \HtoL(\syn{X})$. 
Consider the query 
\[
\syn{Q} = \openquery{\syn{X}}{\VinH}{(\syn{\Pa(Y) \setminus \VinH) \cup (\Pa(X) \setminus Y)}}
\] 
By construction $\qabsMH{\syn{Q}}$ is the following diagram. 
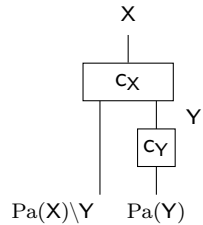
\[
\input{free-cd-pic.tikz}
\]
Then $\HtoLS(\qabsMH{\syn{Q}})$ contains $\syn{\HtoLS(c_Y)}$ which contains $\syn{c_V}$. However since there is no path $\syn{ V \to \HtoL(X)}$, in the diagram $\HtoLS(\qabsMH{\syn{Q}})$ there is no path from $\syn{c_V}$ to the output, making this not a normalised network diagram, and therefore not equal to $\qabsML{\HtoL(\syn{Q})}$. Hence fullness is necessary. Extra strength is necessary by the Markov case, since if $\HtoLS$ exists then it restricts to a functor between Free Markov categories.

Now suppose that both hold. For any query $\syn{Q}$ in $\Openqueries(\varH)$, the proof of the Markov case showed that $\HtoLS(\qabsMH{\syn{Q}})$ is a network diagram, it remains for us to show that fullness ensures it is also normalised, and then we are done by Lemma \ref{lem:unique-ND} \eqref{enum:LemFreeCD}. 

But for this, note that every diagram $\syn{\HtoLS(c_X)}$ involves no discarding, since it normalised and fullness ensures that every input to the diagram has a path to an output. Now for any query $\syn{Q}$, $\qabsMH{\syn{Q}}$ is a normalised network diagram, and so (up to swappings of inputs) factors as $\diagD' \otimes \discard{}$ where $\diagD'$ is a network diagram with no discarding. 
Then $\HtoLS(\qabsMH{\syn{Q}}) = \HtoLS(\diagD') \otimes \discard{}$ where now again $\HtoLS(\diagD')$ again features no discards, since each $\syn{\HtoLS(c_X)}$ lacks them. Hence $\HtoLS(\qabsMH{\syn{Q}})$ is also normalised, as required. 
\end{proof}

We are now able to prove our main result characterising mechanism-level causal abstraction.

\begin{proof}[\proofof~Theorem \ref{Thm:strongCCAnew}]

\eqref{enum:pis-functor} $\implies$ \eqref{enum:mech-level-CCA}:
Given $(\HtoLS, \LtoH)$ satisfying these conditions, by Proposition \ref{prop:caus-mech-abs-main} there is a unique $\HtoL \colon \Do(\VH) \to \Do(\VL)$ making $(\HtoL, \HtoLS, \LtoH)$ a mechanism-level abstraction, and hence mechanism CCA. 

\eqref{enum:mech-level-CCA} $\implies$ \eqref{enum:constr-CCA-well}:
Given $(\HtoL, \HtoLS, \LtoH)$ by assumption $(\HtoL, \LtoH)$ is a constructive causal abstraction. Since the former is a mechanism-level abstraction, the square 
\eqref{eq:CCA-square} commutes and so $\HtoLS$ is structurally well-behaved and hence the alignment $\HtoL$ is structurally well-behaved also, for each structure type. 

\eqref{enum:constr-CCA-well} $\implies$ \eqref{enum:mech-level-CCA}:
Let $(\HtoL, \LtoH)$ be a constructive abstraction with $\HtoL$ structurally well-behaved. Then by  $\HtoL$ acts on queries by $\Do(S) \mapsto \Do(\HtoL(S))$. 

We claim there is a unique functor $\HtoLS$ which is structurally well-behaved which agrees with $\HtoL(\syn{X}) = \HtoLS(\syn{X})$ for all variables $\syn{X}$. 

Indeed if $\HtoLS$ is structurally well-behaved then by Proposition \ref{prop:caus-mech-abs-main} the square \eqref{eq:CCA-square} commutes for the query map $\HtoL \colon \Do(S) \mapsto \Do(\HtoL(S))$, which is indeed the definition of $\HtoL$ in a constructive abstraction by Theorem \ref{Thm:CCA-as-refinement}. So the square must commute. 

Now for each non-input variable $\syn{X} \in \VninH$ we must define $\HtoLS(\syn{c_X})$ in \eqref{eq:strong-CA}. Define $\syn{S} := \Pa(\syn{X}) \setminus \VinH$ and queries
\begin{equation} \label{eq:Q-useful}
\syn{Q_X} := {\openquery{\syn{X}}{\VinH}{\syn{S}}} \qquad \qquad \HtoL(\syn{Q_X}) = {\openquery{\HtoL(\syn{X})}{\VinL}{\syn{\HtoL(S)}}}
\end{equation}
using that $\VninH \subseteq \VoutH$ to define the former query, and $\HtoL(\VoutH) \subseteq \VoutL$ for the latter, since $\HtoL$ belongs to a \varalign. 

Then $\qabsMH{\syn{Q_X}}$ is a normalised network diagram, and in $\strucSH$ we have: 
\begin{equation} \label{eq:S-property}
\input{diagD-cond.tikz}
\end{equation}
Now since \eqref{eq:CCA-square} commutes we must have: 
\begin{equation} \label{eq:pis-def}
% \tikzfig{pis-unique-2}
\input{pis-unique-2-nodots.tikz}
\end{equation}
for the diagram $\diagD_\syn{X} := \HtoLS(\syn{c_X})$, in the last step using \eqref{eq:Q-useful}, and that $\HtoLS$ is a cd-functor and $\HtoL(\VinH) = \VinL$. Now we claim the existence of any such factorisation as on the RHS is equivalent to the statement that: 
\begin{equation} \label{eq:cond-1}
\mechlevelset(\syn{X}) \cap \VinL = \emptyset 
% \text{ for } \syn{X} \not \in \VinL
\end{equation}
Indeed, by definition the LHS is a normalised network diagram and so each input variable without a path to $\syn{X}$ must be simply discarded. Then \eqref{eq:cond-1} states precisely that the remainder are contained in $\HtoL(\Pa(\syn{X}))$, i.e.~that we have such a factorisation. Hence in order to define $\HtoLS$, \eqref{eq:cond-1} must hold for all $\syn{X} \not \in \VinH$, and in this case $\HtoLS(\syn{c_X})$ must be the unique diagram $\diagD_\syn{X}$ in \eqref{eq:pis-def}. This makes $\HtoLS$ unique when it exists. Note that this property holds, and so $\HtoLS$ is definable, if $\HtoL$ is strong, since $\HtoL(\VinH) = \VinL$. Thus it holds whenever $\HtoL$ is structurally well-behaved. We must finally check that $\LtoH$ does indeed satisfy naturality \eqref{eq:strong-CA-nat}. But in $\catC$ we have: 
\begin{align*}
% \tikzfig{pis-proof-simpler-1}
\input{pis-proof-simpler-1-nodots.tikz}
%  \\
% \tikzfig{pis-proof-1}
% \tikzfig{pis-proof-2} 
\end{align*} 
using \eqref{eq:S-property}, \eqref{eq:pis-def}, and cd-naturality of $\LtoH$ for queries. Finally, using our assumption on $\catC$, applying any normalised states to the right-hand discards cancels them out, to obtain $\HtoL_X \circ \HtoLS(c_X) = c_X \circ \HtoL$. Hence $\LtoH$ indeed restricts to such a cd-natural transformation as required. 
\end{proof}

\begin{proof}[\proofof~Lemma \ref{lem:Markov-equiv}]
The condition is clearly necessary by Proposition \ref{prop:caus-mech-abs-main} since $\netdiag{\modelMH} = \qabsMH{\syn{Q}}$ for the query defined by $\syn{Q = \openquery{\VoutH}{\VinH}{}}$ and so if \eqref{eq:CCA-square} commutes then $\HtoLS(\netdiag{\modelMH}) = \qabsML{\HtoL(\syn{Q}}$ making it a normalised network diagram, and by definition of $\HtoLS$ being structurally well-behaved it satisfies these conditions. Indeed, note that since $\HtoLS(\VinH) = \VinL$ and all the $\HtoLS(\syn{X})$ are disjoint we must have $\HtoLS(\VninH) \subseteq \VninL$. 

Conversely, we now show that these conditions are sufficient. Suppose that they hold. Then since $\HtoLS(\VninH) \subseteq \VninL$ the $\HtoLS(\syn{X})$ are disjoint subsets for $\syn{X} \in \VninH$. Since $\HtoLS(\VinH) = \VinL$ the same holds for $\syn{X} \in \VinH$. Since $\VninL$ and $\VinL$ are disjoint the $\HtoLS(\syn{X})$ are therefore all disjoint.

We need to show each $\HtoLS(\syn{c_X})$ is a network diagram. Since the $\HtoL(\syn{X})$ are all disjoint the diagram has no repeated inputs or outputs. Now suppose two instances of $\syn{c_V}$ belong to $\HtoLS(\syn{c_X})$ (after normalisation). Since they both have a path to an output of $\HtoLS(\syn{c_X})$ they also have a path to an output of $\HtoLS(\netdiag{\modelMH})$ and so cannot be discarded, making the latter not a network diagram, a contradiction. Next suppose $\syn{V} \in \HtoL(\syn{Y})$ where $\syn{Y} \in \Pa(\syn{X})$. Then $\syn{Y}$ must not be an input (else $\syn{V}$ would be an input and so $\syn{c_V}$ is not part of the signature). Then by Lemma \ref{lem:unique-ND} \eqref{enum:path-to-diag} we have $\syn{c_V} \in \HtoLS(\syn{c_Y})$. Then $\syn{c_V}$ appears in both $\HtoLS(\syn{c_X})$ and $\HtoLS(\syn{c_Y})$ (due to their inputs and output types), with a path to the respective outputs $\HtoL(X)$ and $\HtoL(Y)$ of each. Since these are outputs to $\HtoLS(\netdiag{\modelMH})$, these two instances of $\syn{c_V}$ cannot be removed from the diagram $\HtoLS(\netdiag{\modelMH})$ (i.e.~are not discarded) and so this cannot be a network diagram, a contradiction. Hence by Lemma \ref{lem:unique-ND} \eqref{enum:NND-equivalent-conditions} each box $\HtoLS(\syn{c_X})$ is indeed a network diagram.

Conversely, suppose that $\HtoL$ is not extra strong. Then there exist non-input variables $\syn{X},\syn{Y} \in \varH$ and some $\syn{V} \in \mechlevelset(\syn{X}) \cap \mechlevelset(\syn{Y})$. By Lemma \ref{lem:unique-ND} \eqref{enum:path-to-diag} $\syn{c_V}$ appears in both $\HtoLS(\syn{c_X})$ and $\HtoLS(\syn{c_Y})$ and just as above this makes $\HtoLS(\syn{c_X})$ not a network diagram, a contradiction. 
\end{proof}

\section{Preservation of algebraic structure} \label{sec:alg-structure}

Organising interventions into query signatures such as $\Intset^{\mathsf{io}}(\modelV)$, $\Do(\modelV)$ and $\II(\modelV)$ has been useful to formalise causal abstraction. However (the free categories generated by) these tend to have little interesting compositional structure, with the only new composites beyond the queries themselves being products of queries.%\footnote{For abstract queries $\Do(\syn{V})$ or $\II(\syn{V})$, the category $\StrucQ$ does contain sequential composites such as of $\syn{\openqueryshort{O}{\Vin}{S}} \colon \syn{\Vin \otimes S \to O}$ and $\syn{\openqueryshort{O'}{\Vin}{S'}} \colon \syn{\Vin \otimes S' \to O'}$ whenever, say $S'$ and $O$ overlap. However, these wouldn't normally be considered quantities of interest that one wishes to compute (all those are all themselves already in the signature $\sigQ$).}  

However there is a second sense in which interventions have algebraic structure: we can apply them successively in sequence. Since this has played a role in causal abstraction, it is natural to ask if we can capture it in our setup. One useful way of doing so for general models and queries is the following. 

\begin{definition} 
	Let us say a query signature $\SigQ$ has \emph{monoid structure} when there is a query signature \rlb{$\sigfont{T}$}\ with the same types and a monoid \rlb{$(M_\SigQ,\bullet)$}\ such that $\SigQ$ has \rlb{queries $\sigfont{T}_{mor} \times M_\SigQ$}. We say $\SigQ$ \emph{forms a monoid} when \rlb{$\sigfont{T}$}\ has two types $\syn{in}, \syn{out}$ and precisely one query $q \colon \syn{in} \rightarrow \syn{out}$.
	For two query signatures with monoid structure $\SigQ, \SigQ'$, we say that a map $F \colon \SigQ \to \SigQ'$ is a \emph{homomorphism} when it factors as $F((\syn{Q},\syn{m})) = (F_Q(\syn{Q}),F_M(\syn{m}))$ where $F_M \colon M_\SigQ \to \rlb{M_{\SigQ'}}$ is a monoid homomorphism. 
\end{definition}

 Let us first observe some key examples of queries with monoid structure. 

\begin{examples} \label{ex:queries-with-monoid-structure}
Let $\modelV$ be a set of concrete variables. 
\begin{enumerate}[label=(\arabic*)]

\item \label{itm:mon-of-abs-Do}
The query signature $\Openqueries(\syn{V})$ \emph{has monoid structure} $\ioqueries(\syn{V}) \times \mathbb{P}(\syn{\Vint})$ where $\ioqueries(\syn{V})$ denotes the input-output queries on $\syn{V}$ and $\mathbb{P}(\syn{\Vint})$ the monoid of subsets of \st{$\syn{\Vint}$} under $\syn{S} \bullet \syn{T} := \syn{S} \cup \syn{T}$. In this case we can in fact extend the monoid multiplication to a partial operation on all of $\Openqueries(\syn{V})$ via $\openquery{\syn{O}}{\Vin}{\syn{S}} \bullet \openquery{\syn{O}}{\Vin}{\syn{T}} = \openquery{\syn{O}}{\Vin}{\syn{S \cup T}}$, defined whenever the output subsets of both queries agree.

\item \label{itm:mon-of-concrete-Do}
The query signature $\Do(\modelV)$ \emph{forms a monoid} via the trivial `empty' intervention $\noint$ 
and composition $\Do(X =x) \bullet \Do(Y=y) := \Do(X = x, {Y \setminus X} = y|_{Y \setminus X} )$. 

\item \label{itm:mon-of-gen-int}
More generally, for any set of interventions $\Intset$ the query signature $\Intsetio$  \emph{forms a monoid} if it is closed under composition: 
for two interventions $\inti,\inti'$  with respective sets of `new' mechanisms $\{c_{X}\}_{X \in A_{\inti}}$ and $\{c_{Y}'\}_{Y \in A_{\inti'}}$  
the composite $\inti \circ \inti'$ is defined by $\{c_{X}\}_{X \in A_{\inti}} \cup \{c_{Y}'\}_{Y \in A_{\inti'}\setminus A_{\inti}}$; the unit is given by $A_{\inti} = \emptyset$. 
\end{enumerate}
\end{examples}

Many of the notions of causal abstraction \rlb{from Section~\ref{sec:causal-abstraction}}\ indeed preserve this algebraic structure present in the corresponding sets of interventions.  
\rlb{The following is straightforward to verify and we omit the proof.}

\begin{proposition}	\label{prop:CA-notions-with-hom} 
Given the monoid structures from Example~\ref{ex:queries-with-monoid-structure} it holds that:
\begin{enumerate}[label=(\alph*)]
	\item \label{itm:CCA-gives-hom}
	For any constructive causal abstraction the query map $\HtoL \colon \Do(\syn{\varH}) \to \Do(\syn{\varL})$ is a homomorphism;  
		and the \qup\ it defines  
		has a homomorphism as query map $\LtoHquery \colon \Do(\modelV_L) \to \Do(\modelV_H)$.\footnote{Here $\HtoL$ is the map between signatures that, by virtue of the constraint in Def.~\ref{def:query-ref}, is defined by the functor $\HtoL$ of a constructive causal abstraction; and saying the map $\LtoHquery \colon \Do(\modelV_L) \to \Do(\modelV_H)$ of the corresponding \qup\ is a homomorphism means it is one on the domain of $\LtoHquery$, which by construction is indeed closed under composition.}  
	\item \label{itm:iso-CCA-gives-hom} 
	Any \isoCCA\ defines an overall \qup\  whose map $\LtoHquery$ is a homomorphism. 
\end{enumerate}
\end{proposition}

\paragraph{Homomorphisms vs order-preservation.}
Given an exact transformation (Def.~\ref{def:exact-trans}) between query signatures with monoid structures, i.e. closed under the composition of interventions, the query map $\LtoHquery$ may or may not be a homomorphism.  
In the literature this map is 
actually required to be \emph{order-preserving} under the partial order $\leq$ with $\inti \leq \inti'$ whenever $\inti \circ \inti' = \inti'$ \cite{rubenstein2017causal, BeckersEtAl_2019_AbstractingCausalModles, geiger2023causal}.\footnote{\rl{In \cite{rubenstein2017causal} and \cite{BeckersEtAl_2019_AbstractingCausalModles}, which do not consider interventions more general than Do-interventions, the partial order that is required to be \rlc{preserved is $\Do(S=s) \leq \Do(S'=s')$ iff $S \subseteq S'$ \emph{and} $s'|_{S}=s$. Clearly, order-preservation wrt that partial order is implied by the preservation of the partial order considered above.}}}
Studying the algebraic structure of interventions, the work in  \cite{geiger2023causal} alternatively considers commutation relations between interventions as significant. 
Being a homomorphism implies preserving both commutation and order relations: for the latter,  
whenever $\inti \circ \inti' = \inti'$  we have 
$\LtoHquery(\inti') = \LtoHquery(\inti \circ \inti') = \LtoHquery(\inti) \circ \LtoHquery(\inti')$. 
As such, it may be more natural to simply require $\omega$ to be a homomorphism. 
This could also give a natural condition for a further strengthening of 
\emph{\strongCA} (Def.~\ref{def:strong-CA}), being satisfied in all examples of the latter in this work.

\section{Comparison with the literature} \label{sec:lit-comparison}
 
\begin{remark}[Exact transformations] \label{rem:ET-notions} 
	The notion of exact transformation given in Definition~\ref{def:exact-trans} differs from those in the literature as follows. 
	\begin{enumerate}
				
		\item \emph{Algebraic structure}: As discussed in detail in Section \ref{sec:alg-structure}, the literature assumes an order-preserving map  $\LtoHquery$ on interventions, while we do not. 

		\item \emph{Surjectivity}: In \cite{rubenstein2017causal, BeckersEtAl_2019_AbstractingCausalModles} (unlike \cite{geiger2023causal}) the map $\LtoH$ of an exact transformation is \emph{not} required to be surjective. However in 
		 \cite{BeckersEtAl_2019_AbstractingCausalModles} it is demonstrated that to qualify as a causal \emph{abstraction}  $\tau$ \rlb{should be required to}\ be surjective. 
		 Indeed we take this to be an essential ingredient of any abstraction relation, causal or not, and hence $\tau$ is necessarily an epimorphism for both kinds of basic abstraction in Sec.~\ref{sec:abstraction}. 
		In this respect our notion of an exact transformation is closer to a $\tau$-abstraction in \cite{BeckersEtAl_2019_AbstractingCausalModles}. 		
		
		\item \emph{Types of causal models}: 
		Apart from \CF{} abstraction, none of our definitions make assumptions on the kind of causal models -- other than that there are no variables considered fundamentally hidden: 
		\rlb{variables are either input or not and while $\Vint = V\setminus \Vin$ may not necessarily all be outputs, a priori one can consider interventions on any of them}.\footnote{Hence, a key observation in \cite{BeckersEtAl_2019_AbstractingCausalModles}, namely that allowing `probabilistic models' in their sense can lead to `unwarranted abstraction' from fine-tuning probability distributions on root nodes does not apply here.} 		
		Seeing as any notion of causal abstraction is only `as strong as' the sets of (intervention) queries with respect to which it obtains, restricting the type of model from the outset seems not essential. 
		
		\item \emph{Kinds of interventions}: Unlike \cite{rubenstein2017causal, BeckersEtAl_2019_AbstractingCausalModles} in Def.~\ref{def:exact-trans} there is no restriction to Do-interventions for $\IntsetL$, $\IntsetH$. (In this respect ours is closer to the notion in \cite{geiger2023causal}.) 						
	\end{enumerate}
\end{remark}

\begin{remark}[Strong causal abstractions] \label{rmk:strong-CA}
	While we define exact transformations for arbitrary interventions, Do-interventions are of greatest conceptual significance, since with access to all of these one can identify an entire causal model. A `vanilla' instance of causal abstraction between models $\model{L}$ and $\model{H}$ is thus the case where $\IntsetH$ contains all Do-interventions 
	-- Beckers and Halpern then call their notion of a $\tau$-abstraction `strong' \cite{BeckersEtAl_2019_AbstractingCausalModles}. 
	Our notion of \emph{\strongCA} (Def.~\ref{def:strong-CA}) is in keeping but also differs in a few ways. 
	\begin{enumerate}
		\item Just as for 
		exact transformations (see above), we demand no restriction of the type of models and no assumption of order-preservation.

		\item More importantly, while $\IntsetH$ is all Do-interventions, $\model{L}$ is a priori unconstrained and, as a consequence, a map $\LtoHquery$ on queries is part of the given data, rather than induced by $\LtoH$ as in \cite{BeckersEtAl_2019_AbstractingCausalModles}. 
	\end{enumerate}
	The motivation comes from various considerations. 
	First, modern AI based on neural networks has broadened the scope of causal abstraction such that one should expect `distributed' low-level representations of interpretable variables. 
	Second, from a theory perspective one would want to understand causal abstraction in the context of general interventions, which the framework of causal models does feature -- independently from machine learning. 
	In fact Beckers and Halpern formulate a conjecture in \cite{BeckersEtAl_2019_AbstractingCausalModles}, which basically amounts to that `strongness' (all high-level Do-interventions) implies `constructiveness', i.e. a partitioning as in a \varalign. 
	Crucially, this intuition relies on a restriction to Do-interventions also at the low level from the start. 
	It is not hard to check that in our set-up of a \strongCA\ such a restriction is equivalent to finding a partitioning of $\VL$. 
	While this doesn't answer their conjecture yet -- due to the other mentioned differences in the respective setups -- it suggests that our notion offers a useful angle to further our understanding of causal abstraction. 
\end{remark}
	
	\color{black}

\st{
\subsection{Comparison with Englberger and Dhami} \label{sec:Engl-comparison}

Closely related is the work of Englberger and Dhami \cite{englberger2025causal}, who give a similar formalisation of causal abstraction based on Markov categories and natural transformations. However, while their categorical definition of causal abstraction aims to capture constructive causal abstraction (as in our Theorem \ref{Thm:CCA-as-refinement}), in fact (and after a fix) it instead captures mechanism-level constructive abstraction more specifically. 

% who also formalise constructive causal abstractions in terms of natural transformations between causal models (without inputs and with all variables as outputs), seen as Markov functors. Cruicially, howwever, their definition of causal abstraction, though aiming to capture constructive causal abstraction (as in our Theorem \ref{Thm:CCA-as-refinement}), after a slight fix in fact captures mechanism-level constructive abstraction more specifically. 

In more detail, they define a causal abstraction in terms of a functor $\iota = \HtoL \colon \strucSH \to \strucSL$ and determinsistic natural transformation $\LtoH \colon \semMLfunc \circ \iota \implies \semMHfunc$, where $\strucSH, \strucSL$ are the free Markov categories induced by causal signatures, and such that $\iota$ maps morphisms in $\restr(\strucSH)$ to those of $\restr(\strucSL)$. Here $\restr(\strucSL)$ is defined as the subcategory of $\strucSL$ `after restricting to those morphisms where every generating box appears at most once'. In their Prop.~2.1, they claim to show that there is a unique such morphism $\syn{X} \to \syn{Y}$ for all  $\syn{X, Y \subseteq \varL}$. 

Unfortunately, however this is not the case. For example, consider a causal model with DAG $\syn{X \to Y}$ and the distinct morphisms $\syn{c_Y} \neq \syn{c_Y} \circ \syn{c_X} \circ \discard{X} \colon \syn{X} \to \syn{Y}$.
% of type $\syn{X \to Y}$. 
One can instead define $\restr(\strucSL)$ as consisting of (normalised) \emph{network diagrams}, which however form a signature rather than a category, making it precisely our signature $\Do(\varL)$ of abstract Do-interventions. Then Prop.~2.1 corresponds to the uniqueness result of Lemma \ref{lem:unique-ND} \eqref{enum:LemFreeCD}.

Making this replacement for both models,
%Replacing $\restr(\strucSL)$, $\restr(\strucSH)$ with $\Do(\varL)$, $\Do(\varH)$ then makes 
their definition of causal abstraction is then precisely a component-level abstraction (Def.~\ref{def:complevel}) for the query sets $\Do(\varL)$, $\Do(\varH)$, i.e.~that of a mechanism level constructive abstraction (\ref{sec:strong-constructive-abstraction}) (apart from not requiring the $\LtoH$ to be epi). Their Theorem 3.1 claims an equivalence between these and a form of constructive causal abstractions. 

Because of the subtleties above, however, it is not the case that every constructive abstraction is of their form. Indeed by our characterisation Theorem \ref{Thm:strongCCAnew} the definition instead captures those constructive abstractions which are \extrastrong.\footnote{Note also Lemma \ref{lem:Markov-equiv} gives a further characterisation of mechanism-level abstractions based on Free Markov categories; in fact we must only check that $\HtoLS(\netdiag{\modelMH})$ is again a network diagram, rather than requiring a condition on all the diagrams of $\Do(\varH)$.} 

\begin{example}
An example of a constructive abstraction which does not fit the definition of a causal abstraction in \cite{englberger2025causal} is given by the models $\modelML, \modelMH$ of Example \ref{example:notsimple}, here setting $\syn{W} = I$, where now $a = c \circ e \colon I \to X$ and $b = d \circ e \colon I \to Y$ for a sharp state $e$ of $Z$. There is a constructive abstraction with $\LtoH=\id{}$ given by $\HtoL(\syn{X}) = \syn{X}$ and $\HtoL(\syn{Y}) = \syn{Y}$. If one were to extend this to a functor $\HtoLS = \iota \colon \strucSH \to \strucSL$ it would have to be given by $\HtoLS(\syn{a}) = \syn{c} \circ \syn{e}$ and $\HtoLS(\syn{b}) = \syn{d} \circ \syn{e}$. But then as diagrams we would have \rl{$\HtoLS(\netdiagout{\modelMH}{X,Y}) \neq  \netdiagout{\modelML}{X,Y}$} just as in \eqref{eq:equality-CA-ex}.
\end{example}

\begin{remark}
%Morally, this example highlights a purely syntactic difference, and indeed is fixed by instead using the free Cartesian category (rather than Markov category). 
As we have noted in Sec.~\ref{sec:discussion}, 
one may wonder whether many or most of  
%it is interesting however that precisely
the constructive abstractions of practical interest 
might turn out to be those which are indeed at the mechanism-level i.e.~which are \extrastrong, and indeed do fit the (corrected version of the) Definition of \cite{englberger2025causal}; we leave this investigation for future work. 
\end{remark}
}

\section{Parallel mechanism channels and consistency conditions in diagrams} 
\label{app:abs-in-pm-via-tracedC} 

Section~\ref{sec:distributed-abstraction} introduced the notion of a parallel mechanism channel $\parmechdiag_\modelM$ induced by (and conversely inducing) a causal model $\modelM$.  
Here we will clarify further the relationship between $\parmechdiag_\modelM$ and $\io{\modelM}$ in a way that will be useful to express the consistency conditions for distributed abstractions.

To this end the following notion will be useful. 
Recall that an SMC $\catC$ is \emph{compact closed} when each object $X$ comes with a dual $X^{\star}$ and units and counits on $X \otimes X^{\star}$ and $X^{\star} \otimes X$, subject to various axioms \cite{selinger2010survey}. For simplicity, and since it \rlc{fits our}\ main examples, we can take $X^{\star} = X$ and the definition then simplifies to that each object $X$ comes with a pair of a distinguished state and effect on $X \otimes X$,\footnote{An \emph{effect} on an object $X$ is a morphism $X \rightarrow I$ with $I$ the unit object.} depicted as a \emph{cup} $\tinycup$ and a \emph{cap} $\tinycap$, respectively, satisfying the following: 
\begin{equation} \label{eq:compact-closure}
	\input{cap-sym.tikz} \hspace*{2cm}
	\input{cup-sym.tikz} \hspace*{2cm}
	\input{snake.tikz}
\end{equation}
If $\catC$ is also a cd-category we moreover demand that: 
\begin{equation} \label{eq:comp-closure-cd-1}
	\input{cap-2.tikz} \hspace*{3.0cm}
	\input{mult-map2.tikz}
\end{equation}

For any sharp state $x$ there is a corresponding deterministic effect denoted $x^\dagger$ and depicted as `flipping $x$ upside-down': 
\begin{equation} \label{eq:xupsidedown}
	\input{cap-3.tikz}
\end{equation}
From \eqref{eq:comp-closure-cd-1} (a) we obtain that any normalised sharp state $x$ satisfies $x^\dagger \circ x = 1$ and from (b) that we have: 
\begin{equation} \label{eq:sharp-effect-w-copy}
	%\tikzfig{sharp-1}
	\input{sharp-2.tikz}
\end{equation}
It follows, by composing \eqref{eq:sharp-effect-w-copy} with $y$, that for all sharp states $x, y$ we have: 
\begin{equation} \label{eq:comp-closure-2}
\input{condition.tikz}
\end{equation}
Our main example is that the Markov category $\FStoch$ sits inside the broader compact closed cd-category $\MatR$ whose objects are again finite sets and morphisms $M \colon X \to Y$ are now arbitrary (not necessarily Stochastic) positive matrices $M(y \mid x)_{x \in X, y \in Y}$. Here the cup $\tinycup$ and cap $\tinycap$ both are just $\delta_{X,X}$, i.e. `perfect correlation'.

Finally, we recall that compact closure also induces a traced category \cite{selinger2010survey}, the categorical generalisation of the notion of a partial trace. 
For each object $C$ the corresponding cup and cap define for any objects $A,B$ the map $\Tr_C \colon \catC(A\otimes C, B \otimes C) \rightarrow \catC(A, B)$ with $\Tr_C(f)$ given by 
\begin{equation}
	\input{trace-in-diags_b.tikz}
\end{equation} 
All these pieces together give a relation between fixed points of a function and traces, which in our setup manifests as a relation between $\parmechdiag_\modelM$ and $\modelM$ as follows. 

Say that $\catC$ has \emph{enough} normalised sharp states if whenever $f \circ x = g \circ x$ for all normalised sharp states $x$ we have $f = g$.

\begin{lemma}
Let $\catC$ be a compact closed cd-category satisfying the above properties and with enough normalised sharp states. Then for any causal model $\modelM$ in $\catC$ over variables $\modelV$ with $\Vout=\Vnin$ we have: 
\begin{equation} \label{eq:trace-cond}
	\input{M-vs-PM-via-traces.tikz}
\end{equation}
\end{lemma}

\begin{proof}
Let $i, o$ be sharp states of $\Vinconc, \Vninconc$ respectively. Then by Lem.~\ref{lem:fixpoints} we have $\modelMio \circ i = o$ iff $\parmechdiag_\modelM \circ (i \otimes o) = (i \otimes o)$. Since $\parmechdiag_\modelM$ is the identity on $\Vin$ this is equivalent to having $\discard{\Vin} \circ \parmechdiag_\modelM \circ (i\otimes o) = o$. Thanks to \eqref{eq:comp-closure-2} this is equivalent to the LHS below being equal to $1$, where the equality below follows from \eqref{eq:sharp-effect-w-copy}.  
\[
\input{fixedflip.tikz}
\]
Hence the RHS above is equivalently equal to $1$. Applying \eqref{eq:comp-closure-2} again, this is equivalent to the LHS of \eqref{eq:trace-cond} composing with $i$ to give $o$. Now we are done since this holds for arbitrary states $i$, and $\catC$ has enough states.
\end{proof}

\color{black}
We can now write the consistency condition~\eqref{eq:compose-exact-trans-explicit} for an \isoCCA\ in Sec.~\ref{subsubsec:isoCCA}, \rlb{considering the}\ low-level distributed Do-interventions from Eq.~\eqref{eq:DDo-pic}, in the following more direct form: 
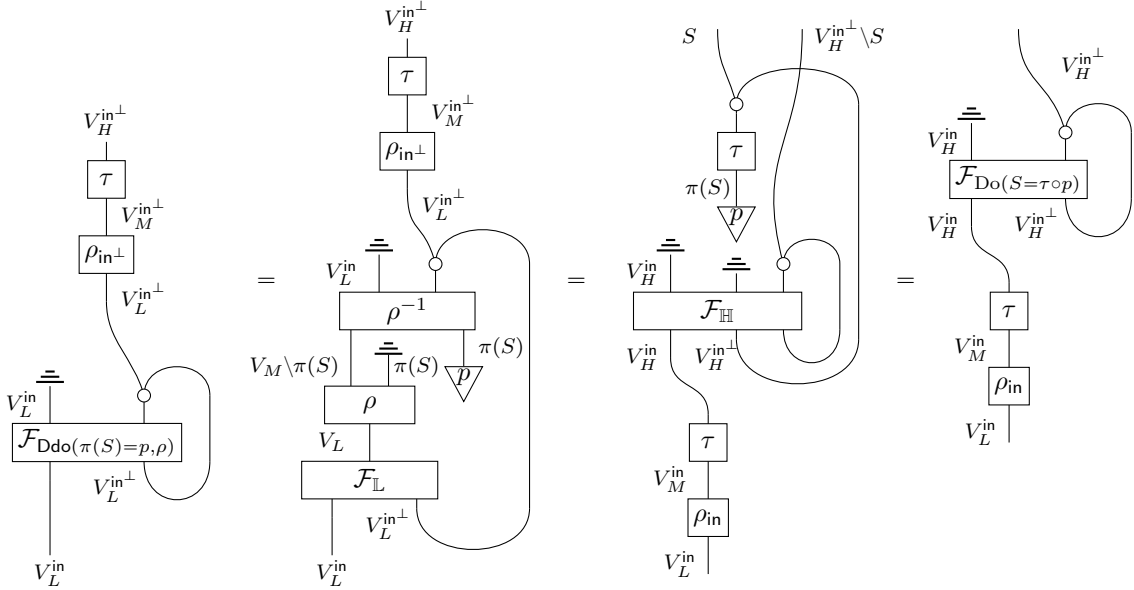
\begin{equation}
	\hspace*{-0.5cm} \input{isoCCA-condition.tikz} \nonumber
\end{equation}

The consistency condition for \isointabs\ can be written in an analogous way. 

\begin{remark} \label{rem:trace-trick-for-HO}
	There is an equivalent representation of classical causal models in the formalism of so called \emph{split-node causal models} where the data of a causal model defines a higher-order map \cite{barrett2019quantum, lorenz2023causal}. 
	What this work calls the parallel mechanism process $\parmechdiag_\modelM$ of a model $\modelM$ is essentially, up to some differences with bookkeeping, the same as one common way to represent a split-node model's   higher-order map. 
	Using a trace structure in the above way to go between $\parmechdiag_\modelM$ and the input-output process $\io{\modelM}$ of the model is then key to recast notions of causal abstraction from Sec.~\ref{sec:causal-abstraction} in the formalism of split-node models. 
	Seeing as split-node models are in turn also the way to make precise how classical causal models are special cases of quantum causal models \cite{barrett2019quantum}, the `trace-structure trick' is also key to studying notions of quantum causal abstraction (see Sec.~\ref{sec:quantum-abstraction} and upcoming work). 
\end{remark}

\color{black}

\section{Quantum-classical reduction} \label{app:q-abs}

The following example illustrates abstraction $\modelML \to \modelMH$ from a quantum model $\modelML$ to a classical model $\modelMH$ in $\QC$ which can be seen as a `reduction' of the quantum model to a classical one.

\begin{example}[Quantum Decision Models] \label{ex:quant-dec-model}
\cite{QCog} Consider an experiment in which human subjects make decisions $A$ and $B$ with classical outcomes $X_A, X_B$; for example answering questions with answer sets $X_A=X_B=\{\text{yes, no}\}$. This yields distributions \rl{$\orderedP{A,B}$}, over answer sequences for `$A$ then $B$', and \rl{$\orderedP{B,A}$}, for `$B$ then $A$'. A \emph{decision model} specifies an object $S$ in an initial state $\rho$ and channels $A, B$ for each decision, such that we recover the distributions in that the following holds:\footnote{Hence it is a compositional model with signature $\SigS$ consisting of $\syn{S, A, B, \rho}$ modelled as $(S, A ,B ,\rho)$ with queries $\sigQ$ given by $\syn{\orderedP{A,B}, \orderedP{B,A}}$ modelled as above.}
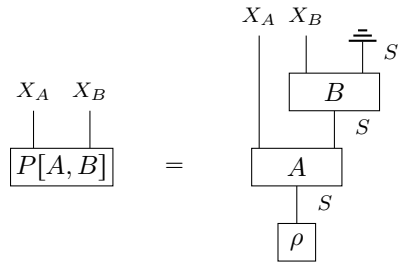
\[
\input{inst-model-general_b.tikz}
\]
along with the similar equation for `$B$ then $A$'. \emph{Quantum decision models}, those in $\QC$ with $\sem{S} = \hilbH$, have been offered as explanations for various psychological effects \cite{QCog}. 

Given such a model $\modelML$, what would an abstraction to a classical decision model $\modelMH$, given by finite set $S$, distribution $\omega$ and channels $A,B$, entail? A \qdown{} $\modelML \to \modelMH$, with $\HtoL =\id{\strucQ}$ and $\LtoH_{X_A}=\id{}$ and $\LtoH_{X_B} = \id{}$, would state that both models coincide on $\orderedP{A,B}$ and $\orderedP{B,A}$, i.e.~a classical model $\modelMH$ is in fact available. More strongly, a \strucdown{} with $\HtoLS = \id{\strucS}$ would yield a channel $\LtoH \colon \hilbH \to S$ such that:  
\[
\input{qinst-abs-1.tikz}
\qquad \qquad \qquad \qquad 
\input{qinst-abs-2.tikz}
\]
and similarly for $B$. Finally, stronger still would be to also offer a `preparation' channel $\LtoH^{-1} \colon S \to \hilbH$ with $\LtoH \circ \LtoH^{-1} = \id{S}$, so that each quantum channel $A$ reduces to a `measurement' channel $\LtoH$, classical channel $A$, and re-preparation $\LtoH^{-1}$.  
For example, taking $\LtoH, \LtoH^{-1}$ as the measurement and preparation channels for an orthonormal basis would make $A$ and $B$ diagonal in this basis. 
\end{example}

\section{Further formalisation of abstraction}
\label{app:further-formalisation}

\subsection{Formalising \qup{}} \label{app:up-abs}

\rlb{Let us now make precise the sense in which diagram \eqref{eq:exact-partial} holds}, 
by relating \qup{}s to natural transformations. 

\rlb{
\begin{proposition} \label{prop:ex-trans}
\rlc{An}\ \qup{} $\qupcomponents{\HtoL}{\LtoH}{\queryLtoH}$ is equivalent to the existence of $\SigModel{\QLsubset}{\model{M}}$ and two \qdown s: 
\[
\input{uabs.tikz}
\]
\begin{equation} \label{eq:query-alignment}
% \tikzfig{Query-alignment-simpler}
\input{Query-alignment-simpler-nodots_tilde.tikz}
% \tikzfig{Query-alignment}
\end{equation}
where $\SigModel{\QLsubset}{\model{M}}$ is the `high-level model' in both \qdown s, $\HtoL$ is injective on queries and $\widetilde{\queryLtoH}$ is the identity on types and surjective on queries. 
The correspondence with $\LtoHquery \colon \SigQL \pto \SigQH$ is realised via $\queryLtoH(\QL) = \widetilde{\queryLtoH}(\QM)$ in $\SigQH$ for any $\QL$ such that $\QL = \syn{\HtoL(\QM)}$.\footnote{\rlb{Note that the functor $\HtoL$ from the \qdown\ agrees with the \qup{}'s $\HtoL$, which is only a map on types.}} 
\end{proposition}
}

Above, each query $\QM$ in $\QLsubset$ has types $\syn{X}$, $\syn{Y}$ from $\SigQH$, but may be identified with a query $\QL := \syn{\HtoL(\QM)} \colon \HtoL(\syn{X}) \to \HtoL(\syn{Y})$ in $\SigQL$, thought of as one on which $\LtoHquery$ is defined with image \rlb{$\queryLtoH(\QL) := \widetilde{\queryLtoH}(\QM)$}\ in $\SigQH$. 

\begin{proof}
Suppose we are given \qdown{}s as above. Since the types in $\SigQH$ are the same as those in $\QLsubset$, we can view $(\HtoL, \LtoH)$ also as a type alignment between $\modelML$ and $\modelMH$. We now extend \rlb{$\widetilde{\queryLtoH}$}\ to $\sigQL$ as outlined in \eqref{eq:query-alignment}. For $\syn{Q} \in \SigQL$ set $\LtoHquery(\syn{Q})$ as defined iff $\syn{Q} = \HtoL(\syn{Q}')$ for some $\syn{Q}' \in \QLsubset$. In this case define \rlb{$\LtoHquery(\syn{Q}) := \widetilde{\queryLtoH}(\syn{Q}')$}. Then since \rlb{$\widetilde{\queryLtoH}$}\ is surjective from $\QLsubset$, \rlb{so is $\LtoHquery$ from $\SigQL$}\ and one may verify that whenever it maps a query to one of type $\syn{X} \to \syn{Y}$ that query has type $\HtoL(\syn{X}) \to \HtoL(\syn{Y})$. 

It remains to check that consistency holds. Suppose that $\syn{Q}_H = \LtoHquery(\syn{Q})$ for $\syn{Q} \in \sigQL$. Then $\syn{Q} = \HtoL(\syn{Q}')$ for $\syn{Q}' \in \QLsubset$. Then consistency holds since:
\[ 
Q_H \circ \LtoH = \rlb{\widetilde{\queryLtoH}(Q')} \circ \LtoH = Q' \circ \LtoH = \LtoH \circ \HtoL(Q') = \LtoH \circ Q
\]
In the second equality we used that $(\LtoH, \id{})$ is a \rlb{strict \qdown{}}.  
In the third equality we used that $(\HtoL, \LtoH)$ is a \qdown{}.

Conversely, given any \qup{}, define $\QLsubset$ to have the same types as $\sigQH$ and a query $\syn{Q}^* \colon \syn{X} \to \syn{Y}$ for each query $\syn{Q} \colon \HtoL(\syn{X}) \to \HtoL(\syn{Y})$ on which $\LtoHquery$ is defined; \rlb{and set $\sem{\syn{Q}^*}_{\model{M}} := \semMH{\LtoHquery(\syn{Q})}$ and $\sem{\syn{X}}_{\model{M}} := \semMH{\syn{X}}$ for all types $\syn{X}$}. 
Define $\HtoL(\syn{Q}^*) := \syn{Q}$ and \rlb{define $\widetilde{\queryLtoH}$ by $\widetilde{\queryLtoH}(\syn{Q}^*) := \LtoHquery(\syn{Q})$ and as the identity on types}. Then by definition $\HtoL$ is injective on types, and \rlb{$\widetilde{\queryLtoH}$ is surjective}. 
The fact that \rlb{$\sem{\syn{Q}^*}_{\model{M}} := \semMH{\LtoHquery(\syn{Q})}$ ensures that $(\widetilde{\queryLtoH}, \id{})$}\ forms a strict \qdown{}. Finally for $(\HtoL, \LtoH)$ to form a \qdown{} we must check consistency, which holds since we always have $\rlb{\sem{\syn{Q}^*}_{\model{M}}} \circ \LtoH = \semMH{\LtoHquery(\syn{Q})} \circ \LtoH = \LtoH \circ \HtoL(Q)$. 
\end{proof}

\subsection{Characterising abstractions}

We can characterise \rl{\restrictedCCA s}\ and \CF{} abstractions more categorically, as follows. 

\begin{lemma} \label{lem:weak-CA-chara}
For any causal model \rl{over $\modelV$}\ with structure category $\strucS$, the functor \rlb{$\qabsMfunc \colon \cat{\WOpenqueries(\syn{V})} \to \cat{\strucS}$}\ sending each query to its diagram factorises over that \rlb{$\qabsMfunc \colon \cat{\Openqueries(\syn{V})} \to \cat{\strucS}$ for Do-queries, via a functor $\qabsMfunc \colon \cat{\WOpenqueries(\syn{V})} \to \cat{\Openqueries(\syn{V})}$}.\footnote{Note that each generator of $\WOpenqueries(\syn{V})$ can indeed be mapped to a diagram in the free category $\cat{\Openqueries(\syn{V})}$, and so such a functor exists.} 
 Then \rl{an \restrictedCCA}\ is equivalent to a 
\qdown{} with $\HtoL(\VinH) = \VinL$ and such that there exists a (unique) query map $\HtoL'$ for which the following diagram commutes. 
\begin{equation} \label{eq:II-diagram}
\input{II-diagram.tikz}
\end{equation}
\end{lemma}
\begin{proof}
We check that \eqref{eq:weakCA-query-map} is equivalent to the existence of $\HtoL'$ in \eqref{eq:II-diagram}. Since $\qabsMHfunc, \qabsMLfunc$ preserve all variables, $\HtoL'(\syn{V}) = \HtoL(\syn{V})$ for all $\syn{V} \in \VH$. Just as in the proof of Theorem \ref{Thm:CCA-as-refinement} any query map $\HtoL'$ preserving inputs must send each $\openqueryshort{}{}{S}$ to $\openqueryshort{}{}{\HtoL'(S)}$, as the latter is the only query of correct inputs and outputs (i.e. be defined by \eqref{eq:query-map-CA} replacing $\HtoL$ by $\HtoL'$). Now, with this unique map $\HtoL'$, one may see that the diagram commutes iff \eqref{eq:weakCA-query-map} holds. 
\end{proof}

\begin{lemma} \label{lem:CF-characterisation}
 For any FCM $\modelM$ over \rlb{$\syn{\FMCVen}$, $\syn{U}$}\ define a set of abstract queries \rlb{$\Openqueries(\syn{\FMCVen}, \syn{U})_\lambda$}\ given by the disjoint union of \rlb{$\Openqueries(\syn{\FMCVen} \cup \{\syn{U}\})$}, with a single type $\syn{U}$ treated as the inputs, along with a single extra query $\syn{\lambda}$ given by a state of $\syn{U}$. The functor \rlb{$\qabsMfunc \colon \absLthree(\syn{\FMCVen}) \to \strucS$}\ sending each query to its diagram factorises over \rl{$\qabsMfunc \colon \Openqueries(\syn{\FMCVen}, \syn{U})_\lambda \to \strucS$}, via a functor we denote by $\qabsMfunc$. 
A \CF{} abstraction is precisely a \qdown{} 
\rlb{$\qdownabs{\HtoL}{\LtoH}{\SigModel{\absLthree(\FMCVen_L)}{\modelML}}{\SigModel{\absLthree(\FMCVen_H)}{\modelMH}}$}\  
such that there exists a (unique) query map $\HtoL'$ for which the following commutes: 
\begin{equation} \label{eq:CFabs}
\rlb{\input{CFabs-diagram.tikz}}
\end{equation}
and $\HtoL'(\syn{U}_H) = \syn{U}_L$. 
\end{lemma}
\begin{proof} 
We check that this is equivalent to the \qdown{} taking the \rlc{form \eqref{eq:CF-abs-query-map}. 
On}\ variables \rlb{$\syn{X} \in \FMCVen_H$}\ we have $\HtoL'(\syn{X}) = \HtoL(\syn{X})$ by construction. Just as in the proof of Theorem \ref{Thm:CCA-as-refinement} in order to preserve types any query map $\HtoL'$ preserving inputs must send each $\openqueryshort{}{}{S}$ to $\openqueryshort{}{}{\HtoL(S)}$, and by construction must send $\lambda$ to $\lambda$. So $\HtoL'$ is uniquely characterised by $\HtoL$. Then one may see that \eqref{eq:CFabs} commutes iff \eqref{eq:CF-abs-query-map} holds. Indeed, writing out the diagrams corresponding to each query \rlb{$\Lthreequery{Y_1|_{\CFOpenSet_1},\dots,Y_m|_{\CFOpenSet_m}}$}\ and applying $\HtoL'$ means applying $\HtoL$ to each opening box for $\modelF$, which precisely yields the diagram corresponding to \rlb{$\Lthreequery{\HtoL(Y_1)|_{\HtoL(\CFOpenSet_1)},\dots,\HtoL(Y_m)|_{\HtoL(\CFOpenSet_m)}}$}.
\end{proof}

\subsection{The category $\Model(\catC)$}

\rl{From a categorical perspective there is a succinct way to denote and capture this work's main definitions, which however was omitted in the earlier sections for reasons of simplicity. 
We now briefly present this approach. 
For what follows fix a d- or cd-category $\catC$ with all functors and natural transformations of the corresponding, appropriate kind.}

\begin{definition}
We write $\Model(\catC)$ for the category whose objects \rl{$\Gensig^{\modelM}$}\ are pairs of a signature \rl{$\Gensig$ along with a model $\modelM$ of $\Gensig$}\ in $\catC$. By a morphism $(F,\LtoH) \colon \Gensig^\modelM \to \Gensig'^{\modelM'}$ we mean a functor \rl{$F \colon \Gensigcat \to \Gensigcat'$}\ along with an epic natural transformation $\LtoH$ as below. 
\[
\input{morphism-sigs.tikz}
\]
Composition is given by $(G, \tau') \circ (F, \tau) = (G \circ F, \tau \circ \tau'_F)$ with $\id{\SigC} = {\id{}, \id{}}$.  
\end{definition}

For short we often denote such a morphism by $F$ in place of $(F,\LtoH^F)$. We say a morphism $F$ is: 
\begin{itemize}
\item 
\emph{strict} when $\tau=\id{}$, so that $\sem{-}' \circ F = \sem{-}$. It is easy to then see that strict morphisms are closed under composition. 

\item 
 a \emph{map} when $F$ sends generators to generators, and then \emph{surjective} when every generator of \rl{$\Gensigcat'$}\ is of the form $F(\syn{g})$ for some generator $\syn{g}$ of \rl{$\Gensigcat$}. 
\end{itemize}

Any compositional model $\modelM$ of structure $\SigS$ and queries $\SigQ$ in $\catC$ amounts to two distinct objects $\SigM$ and $\QueryM$ in $\Model(\catC)$. Then:
\begin{itemize}
\item 
Queries $\SigQ$ are abstract just when there is a strict morphism $\qabs{-} \colon \QueryM \to \SigM$. 

\item 
A \qdown{} $\modelML \to \modelMH$ is precisely a map $\HtoL \colon \QueryMH\to \QueryML$. 

\item 
A \qdown{} is \rl{\strucdownprop}\ when there is a morphism $\HtoLS$ such that the following commutes: 
\[
\input{struc-level-simple.tikz}
\]
The \rl{\strucdownprop}\ \qdown{} case is strict when $\HtoLS$ (and hence $\HtoL$) is strict. 
\end{itemize} 

\begin{example}
There is a strict map $\IOSig(\varV) \to \Openqueries(\varV)$ which is injective on variables. A constructive causal abstraction is precisely a morphism $\HtoL$ \rlb{such that the diagram below commutes}.  
\[
%\tikzfig{io-factor}
\input{io-factor_b.tikz}
\]
\end{example}

We can treat \qup{}s as follows. Define a \emph{transformation} to be a pair of maps: 
\begin{equation} \label{eq:exact-trans-again}
\rlb{\input{partial-map-sigs-2-flipped.tikz}}
\end{equation}
such that $\queryLtoH$ is strict, surjective, and the identity on variables. Then:
\begin{itemize}
\item  \rlc{An}\ \qup{} is the special case where $\HtoL$ is injective on types and generators. 
\item 
A \qdown{} is the special case where $\queryLtoH$ is a bijection on generators, so we may identify \rlb{$\QLsubset^{\model{M}} = \QHC$}. Thus it is simply a map $\HtoL \colon \QHC \to \QLC$. 
\end{itemize}

\rl{It is especially when working with relations between given queries and models that are more complex than this work focused on that}\ it may be useful to work in the category $\Model(\catC)$ and make use of the above terminology. We leave the full exploration of this category and its properties to future work.

\paragraph{\rl{A yet more general, functorial view of abstraction.}}
From a categorical perspective, perhaps the most general notion of model, and abstraction between \rl{models}, is as a collection of categories (not necessarily with given signatures), functors and a natural transformation: 
\[
\input{general-square.tikz}
\]
for which the pair $(\HtoL, \queryLtoH)$ form a \emph{relation} i.e.~such that $\langle \HtoL, \queryLtoH \rangle \colon \strucQ_M \to \strucQ_H \times \strucQ_L$ is faithful. This would mean that we can understand the abstraction as a structured manner of relating high-level and low-level `queries'. That is, for those morphisms (`queries') $(Q_H, Q_L)$ contained in the relation, we obtain the canonical consistency equation \eqref{eq:consistency}. Allowing for arbitrary relations of this form would yield more general ways of relating high and low level queries in a structured manner than we have considered here (not necessarily given by a partial function), but would be interesting to explore in future work.

\end{document}